\documentclass[a4paper,11pt,fleqn]{article}

\usepackage{geometry}

\usepackage{xcolor}
\definecolor{darkLink}{RGB}{0, 64, 128}
\usepackage{hyperref}
\hypersetup{
    colorlinks=true,
    linkcolor=darkLink,
    citecolor=darkLink,      
    urlcolor=darkLink,
    }
\usepackage{graphicx}

\usepackage{amsmath,amssymb,amsthm}

\usepackage{algpseudocode}

\usepackage{tikz}

\newtheorem{thm}{Theorem}[section]
\newtheorem{lem}[thm]{Lemma}
\newtheorem{clm}[thm]{Claim}

\usepackage{xspace} 
\newcommand{\QRAMc}{QRAM\textsubscript{C}\xspace}
\newcommand{\QRAMq}{QRAM\textsubscript{Q}\xspace}

\renewcommand{\>}{\rangle}
\newcommand{\<}{\langle}
\newcommand{\qqAnd}{\qquad\text{and}\qquad}

\newcommand{\Yes}{\mathcal{Y}}
\newcommand{\No}{\mathcal{N}}

\newcommand{\OO}{\mathrm{O}} 
\newcommand{\F}{\mathbb{F}} 
\newcommand{\C}{\mathbb{C}}

\newcommand{\Z}{\mathbb{Z}}

\newcommand{\fS}{{\mathfrak{S}}}
\newcommand{\fP}{{\mathfrak{P}}}
\newcommand{\fL}{{\mathfrak{L}}}
\newcommand{\fR}{{\mathfrak{R}}}

\newcommand{\fa}{{\mathfrak{a}}}
\newcommand{\fb}{{\mathfrak{b}}}
\newcommand{\fc}{{\mathfrak{c}}}
\newcommand{\fq}{{\mathfrak{q}}}

\newcommand{\cH}{{\mathcal{H}}}
\newcommand{\cE}{{\mathcal{E}}}

\DeclareMathOperator{\polylog}{polylog}

\newcommand{\beE}{\begin{equation}}
\newcommand{\enE}{\end{equation}}

\usetikzlibrary{shadings}
\usetikzlibrary{calc}

\renewcommand{\Yes}{S}
\renewcommand{\No}{\overline{S}}

\newcommand{\ii}{\ensuremath{\mathsf{i}}}
\newcommand{\ee}{\ensuremath{\mathsf{e}}}

\newcommand{\algostate}{\varphi}
\newcommand{\nmean}{\lambda_c} 

\DeclareMathOperator{\DFT}{DFT}
\DeclareMathOperator{\NTT}{NTT}
\DeclareMathOperator{\im}{im}

\usepackage{caption} 

\newcommand{\polyr}[1]{\left\lfloor #1 \right\rceil}
\newtheorem{fact}[thm]{Fact}

\title{A nearly linear-time Decoded Quantum Interferometry algorithm for the Optimal Polynomial Intersection problem}
\author{Ansis Rosmanis \bigskip
\\Quantinuum, Partnership House, Carlisle Place,
\\ London SW1P 1BX, United Kingdom
} 
\date{}

\begin{document}

\maketitle

\begin{abstract}
    Recently, Jordan et al.~(\href{https://doi.org/10.1038/s41586-025-09527-5}{Nature, 2025}) introduced a novel quantum-algorithmic technique called Decoded Quantum Interferometry (DQI) for solving specific combinatorial optimization problems associated with classical codes. They presented a constraint-satisfaction problem called Optimal Polynomial Intersection (OPI) and showed that, for this problem, a DQI algorithm running in polynomial time can satisfy a larger fraction of constraints than any known polynomial-time classical algorithm.

    In this work, we propose several improvements to the DQI algorithm, including sidestepping the quadratic-time Dicke state preparation. Given random access to the input, we show how these improvements result in a nearly linear-time DQI algorithm for the OPI problem. Concurrently and independently with this work, Khattar et al.~(\href{https://arxiv.org/abs/2510.10967}{arXiv:2510:10967}) also construct a nearly linear-time DQI algorithm for OPI using slightly different techniques.
\end{abstract}

\tableofcontents

\newpage

\section{Introduction}

Quantum algorithms provide exponential speedups over their classical counterparts for a wide range of computational problems. Notable examples include various number-theoretic problems, such as the integer factoring and the discrete logarithm, as well as the simulation of quantum mechanical systems. While these quantum algorithms already showcase the immense potential of quantum computers, achieving exponential quantum-over-classical speedups for NP-hard optimization problems---another important class of computational problems---has remained elusive for a long time.

A recent work by Yamakawa and Zhandry presented an NP search problem relative to a random oracle that provides exponential quantum-over-classical speedups~ \cite{yamakawa24:RandomOracleNP}. They considered a computational problem based on classical codes, for which the quantum speedups arose in part from the ability of quantum computers to perform the corresponding error correction in superposition. Building on Yamakawa and Zhandry's techniques, Jordan et al.~recently achieved similar results for the standard (i.e., non-relativized) computation~\cite{jordan25:DQI-original}. Specifically, they introduced an optimization problem called Optimal Polynomial Intersection (OPI) along with an efficient quantum approximation algorithm for the problem that achieves approximation ratios surpassing those of any known efficient classical algorithm. More generally, they developed a novel quantum-algorithmic framework called Decoded Quantum Interferometry (DQI) for solving specific combinatorial optimization problems associated with classical codes, with OPI being a key example of such a problem.
  
Given the significance of DQI, as evidenced by several recent works \cite{
chailloux25:SoftDecoders,
patamawisut25:Circuits-for-DQI,
bu25:DQI-under-noise,
sabater25:industiralLP-via-DQI,
marwaha25:complexity-of-DQI,
 anschuetz25:DQI-needs-structure,
schmidhuber25:HamiltonianDQI,
hillel25:quadratic-DQI,
khattar25:optimized_DQI}, 
it is natural to explore whether any aspects of this algorithmic technique can be improved. In this paper, we present several improvements, some of which are rather general while others are specifically tailored to the OPI problem. As a result, we demonstrate that, given a quantumly addressable quantum memory, these improvements yield a nearly linear-time DQI algorithm for the OPI problem.

\subsection{Decoded Quantum Interferometry}
\label{sec:introdqi}

\paragraph{The max-LINSAT problem.}

The computational problem addressed by DQI is a generalization of max-XORSAT to non-binary fields. Jordan et al.~\cite{jordan25:DQI-original} named this problem max-LINSAT, and it is defined as follows. Let $\F_p$ be the finite field of prime order $p$. Given a collection of vectors $b_1,\ldots,b_m\in\F_p^n$ and a collection of sets $S_1,\ldots,S_m\subset\F_p$, the max-LINSAT problem is to find a vector $x\in\F_p^n$ maximizing the number of satisfied constraints $b_i\cdot x\in S_i$.

The matrix $B$ formed by the vectors $b_1,\ldots,b_m$ as its rows plays an instrumental role in the DQI algorithm. In particular, DQI performs well when the minimum distance $d^\perp$ of the code $C^\perp:=\ker B^\top$ is large.

\paragraph{Stages of the DQI algorithm.}

The DQI algorithm, as presented by Jordan et al., has four primary stages that correspond to four main registers: the weight, mask, error, and syndrome registers. At the end of any given stage, only one of these registers is in a nontrivial state, while the others remain initialized in or have already been coherently returned to the all-zeros state $|00\ldots0\>$.

Let us now briefly present these stages and describe the transitions between them. This will not only familiarize us with the DQI framework, but it will also help to illustrate the contributions of the present work.

\begin{enumerate}
\item
Initially, the weight register gets set to a superposition
\beE
\label{eq:stage1}
\sum_{k=0}^{\ell}w_k|k\>,
\enE
 where $\ell\le d^\perp/2-1$ and $w_k$ is any collection of normalized amplitudes. 
\item
Next, for every $|k\>$ in the weight register, the mask register gets prepared as the Dicke state 
$|D_k^m\>:=\sum_{\substack{\mu\in\{0,1\}^m\\|\mu|=k}}|\mu\>/\sqrt{\binom{m}{k}}$.
Then, given $|\mu\>$ in the mask register, we uncompute%
\footnote{Jordan et al.~postpone the uncomputing procedure to a later stage of the algorithm, but that procedure commutes with the other steps, and thus can be done already here.}
 the weight register by computing the Hamming weight of $\mu$, hence obtaining the state
\beE
\label{eq:stage2}
\sum_{k=0}^\ell \sum_{\substack{\mu\in\{0,1\}^m\\|\mu|=k}}
\frac{w_k}{\sqrt{\binom{m}{k}}}|\mu\>.
\enE
\item
Next, for every $|\mu\>=|\mu_1\mu_2\ldots\mu_m\>$ in the mask register, the error register gets prepared in an input-dependent state $|\mathcal E_\mu\>=\bigotimes_{i=1}^{m}|\mathcal E_{\mu,i}\>$ where $|\mathcal E_{\mu,i}\>=|0\>$ if $\mu_i=0$ and, if $\mu_i=1$, $|\mathcal E_{\mu,i}\>$ is a certain superposition $|\widehat G_i\>$ over $\F_p\setminus\{0\}$ whose amplitudes are determined by $S_i$.%
\footnote{While Stages 1 and 2 of the algorithm introduce qubit registers, Stages 3 and 4 introduce $p$-dimensional qudit registers. Each such qudit can be encoded into $\lceil \log_2 p\rceil$ qubits.}
 The state $|\cE_\mu\>$ can be uniquely expressed as
$|\cE_\mu\>=\sum_{y\in\F_p^m}\beta_y|y\>$ and, because $\<0|\mathcal E_{\mu,i}\>=\delta_{\mu_i,0}$, we have that, for all $y$ with $\beta_y\ne 0$, the mask $\mu$ and the error $y$ are non-zero at exactly same indices $i$. This means that, given any such pair $|\mu,y\>$, we can easily uncompute $\mu$, which one does, obtaining
\beE
\label{eq:stage3}
\sum_{k=0}^\ell \sum_{\substack{y\in\F_p^m\\|y|=k}}
\frac{w_k}{\sqrt{\binom{m}{k}}}\beta_y|y\>.
\enE
\item
Next, for every $|y\>$ in the error register, we first compute $|B^\top y\>$ in the syndrome register. For any two distinct errors $y,y'$ of Hamming weights at most $\ell\le d^\perp/2-1$, we have $B^\top y\ne B^\top y'$. Therefore, it is possible to obtain $y$ from $B^\top y$, even if it might be computationally expensive, and this is indeed what one does to uncompute $y$ stored at the error register. Thus, one ends up with the state 
\[
\sum_{k=0}^\ell \sum_{\substack{y\in\F_p^m\\|y|=k}}
\frac{w_k}{\sqrt{\binom{m}{k}}}\beta_y|B^\top y\>.
\]
Finally, one applies the quantum Fourier transform $F_p$ on each of the $n$ qudits of the syndrome register and measure, obtaining a solution $x\in\F_p^n$ to the max-LINSAT problem.
\end{enumerate}

\paragraph{Runtime of DQI.}

Now that we know what are the four stages of DQI, let us briefly mention their computational costs when they are implemented as in~\cite{jordan25:DQI-original}, and thus the computational costs of the whole DQI algorithm. Let us assume here that $\ell$, $p$, and $n$ are all $\Theta(m)$---as is the case in the parameter regime that Ref.~\cite{jordan25:DQI-original} uses for the OPI problem---and thus let us express all complexities in $m$. Here $\widetilde\OO(m^\alpha)$ is a shorthand for $\OO(m^\alpha\polylog m)$.

\begin{enumerate}
\item
Given a classical description of amplitudes $\alpha_k$, the state~(\ref{eq:stage1}) can be prepared in time $\widetilde{\OO}(m)$~\cite{low24:Tgates}.
\item
The Dicke state $|D_k^m\>$ can be prepared in time $\OO(m^2)$~\cite{bartschi22:Dicke}, and, for each $|\mu\>$ constituting that Dicke state, the Hamming weight of $\mu$ can be computed in time $\OO(m)$.
\item
To implement a subroutine $\widehat G_i$ that maps $|0\>$ to the superposition $|\widehat G_i\>$, Jordan et al.~first compute the amplitudes $\widehat g_i(0),\widehat g_i(1),\ldots,\widehat g_i(p-1)$ of $|\widehat G_i\>$ explicitly, which can be done in advance classically from the input $S_i$ in time $\widetilde\OO(m)$, say, using the fast Fourier transform, and then they use the techniques of Ref.~\cite{low24:Tgates} to prepare the state corresponding to these amplitudes, namely, $|\widehat G_i\>$. Since $\widehat G_i$ controlled by $\mu_i\in\{0,1\}$ has to be performed for every $i=1,\ldots,m$, the total running time amounts to $\widetilde\OO(m^2)$. The uncomputing of $\mu$ from $y$ can be done in time $\widetilde\OO(m)$.
\item
Computing the syndrome $B^\top y$ from the error $y$ is done by the straight-forward matrix-vector multiplication in time $\widetilde\OO(m^2)$. The reverse---uncomputing the error $y$ given only the syndrome $B^\top y$---is however not at all straight-forward, and it is done in three steps.
\begin{enumerate}
\item Find any $y'\in\F_p^m$ such that $B^\top y'=B^\top y$, which is done using Gaussian elimination.
\item Error-correct $y'$ to the nearest codeword $c\in C^\perp$, that is, find $c$ such that $B^\top c=0$ and $|y'-c|<d^\perp/2$.
\item Return $y'-c$ as $y$.
\end{enumerate}
When analyzing the runtime of these steps, it is important to specify which implementation of Gaussian elimination is being considered. The standard cubic-time algorithm relies on pivoting, where the subsequent actions of the algorithm depend on values stored in memory. Such pivoting inherently assumes Random Access Memory (RAM). However, because in Step (a) Gaussian elimination must be executed in quantum superposition, achieving cubic runtime now would require \QRAMq. We elaborate on this quantum analogue of RAM in Section~\ref{sec:prelim}.

Fortunately, Gaussian elimination can also be implemented in a branch-free, non-pivoting form. Therefore, it can also be run in cubic time within the standard quantum circuit model.

The runtime of the error correction in Step (b) highly depends on the properties of the code $C^\perp$. For the OPI problem, the relevant code $C^\perp$  is the narrow-sense Reed--Solomon code \cite{reed60:RS-codes}, for which Jordan et al.~propose, as an example, to use the Berlekamp--Massey algorithm \cite{berlekamp1968:book,massey1969:BM-algo} for error correction.%
\footnote{Originally, in~\cite{jordan24:DQI-original}, Jordan et al.~proposed using the Berlekamp--Welch algorithm, which runs in time $\widetilde\OO(m^3)$.}
The Berlekamp--Massey algorithm runs in time $\widetilde\OO(m^2)$.
\end{enumerate}

\noindent
Collecting the costs of all the stages, the total runtime of DQI amounts to $\widetilde\OO(m^3+T)$, where $T$ is the time used for the error-correcting procedure.

\paragraph{The ratio of satisfied constraints.}

Having discussed how quickly the DQI algorithm produces a random solution $x\in\F_p^n$, we should now address the expected quality of that solution, namely, what fraction of constraints can we expect $x$ to satisfy. Jordan et al.~present a theorem  that quantifies this expectation for a rather broad class of max-LINSAT problems.

Suppose we have a family of max-LINSAT problems indexed by $m$ and each problem in the family is specified by $B\in\F_p^{m\times n}$ and $S_1,\ldots,S_m$ (here $n$ and $p$ are some functions of $m$). Additionally suppose that the distance $d^\perp$ of the code $\ker B^\top$ asymptotically grows linearly in $m$ and that there exists some $r$ such that all $S_i$ have the same cardinality $|S_i|=r$. Then \cite[Theorem 1.1]{jordan25:DQI-original} tells us that one can choose amplitudes $w_0,w_1,\ldots,w_\ell$ of the first stage of the DQI algorithm, where $\ell:=\lfloor d^\perp/2-1\rfloor$, so that, in the limit $m\rightarrow\infty$, the expected ratio $\<f\>/m$ of satisfied constraints is
\beE
\label{eq:soverm}
\frac{\langle f \rangle}{m} = 
\begin{cases}
\left( \sqrt{\frac{\ell}{m} \left( 1 - \frac{r}{p} \right)} + \sqrt{\frac{r}{p} \left( 1 - \frac{\ell}{m} \right)} \right)^2
&\text{if}\quad \frac{r}{p} < 1 - \frac{\ell}{m},
\\
 1
 &\text{if}\quad \frac{r}{p} \geq 1 - \frac{\ell}{m}.
\end{cases}
\enE

\paragraph{Our improvements to runtime.}

In the present work, we improve the runtime of various stages of the DQI algorithm in the broad parameter regime mentioned above, namely, where \cite[Theorem 1.1]{jordan25:DQI-original} holds.

Our first improvement is that we effectively combine Stages 1 and 2, bypassing the Dicke state preparation altogether. In particular, we show that there is a choice of amplitudes $w_k$ that give us the expected ratio of satisfied constraints as in (\ref{eq:soverm}), yet the state (\ref{eq:stage2}) at the end of Stage 2, which is always a superposition over Dicke states, is a simple product state and thus producible in linear time. Thus we give a quadratic speedup for combined Stages 1 and 2.

Our second improvement concerns Stage 3, in particular, the implementation of the procedure $\widehat G_i$ for every $i$, which maps $|0\>$ to $|\widehat G_i\>$. We show that, if we are given quantum random access to the classical input $S_i$ (we call this a \QRAMc-access; see Section~\ref{sec:prelim}), we can implement $\widehat G_i$ in time $\widetilde\OO(\sqrt{p/\min\{r,p-r\}})$. When $r\approx p/2$, which will be the case in the parameter regime relevant to the OPI problem, this runtime is only $\OO(\polylog m)$, and hence the whole Stage 3 can be done in time $\widetilde\OO(m)$. For our implementation of $\widehat G_i$, we use techniques reminiscent of the exact Grover's search (see \cite{long01:grover,hoyer00:exactGrover,brassard02:exactGrover}), after which we employ the quantum Fourier transform, contrasting~\cite{jordan25:DQI-original}, where the discrete Fourier transform is employed to implement $\widehat G_i$.

\subsection{Optimal Polynomial Intersection}

Currently, the best demonstration of the power of the DQI algorithm is its application to the Optimal Polynomial Intersection (OPI) problem introduced in~\cite{jordan25:DQI-original}.
This constraint-satisfaction problem is defined as follows.
Given a prime $p$, subsets $T_1,\ldots, T_{p-1}\subset\F_p$, and an integer $n<p$, the OPI problem is to find a degree $n-1$ polynomial $X$ in $\F_p[y]$ that maximizes 
\[
f_{\mathrm{OPI}}(X) = | \{ z \in \F_p^* : X(z) \in T_z \} |,
\]
where $\F_p^*:=\F_p\setminus\{0\}$ is the multiplicative group of the non-zero elements in $\F_p$.
In other words, the OPI problem is to find a polynomial $X(z)=\sum_{j=0}^{n-1}x_jz^j$ that intersects as many of subsets $T_1,\ldots,T_{p-1}$ as possible.

\paragraph{OPI as a special case of max-LINSAT.}

OPI can be easily reduced to a special case of max-LINSAT as follows, and, thus, it can be optimized using the DQI algorithm. 
Every finite field $\F_p$ has at least one primitive element, that is, an element $\gamma\in\F_p$ such that $\{\gamma^1,\gamma^2,\ldots,\gamma^{p-1}\}=\F_p^*$.
Take $\gamma$ to be any primitive element of $\F_p$ (it can be found in time $\widetilde\OO(p)$). Then
\[
f_{\mathrm{OPI}}(X) = | \{ i \in \{1,2, \dots, p-1\} : X(\gamma^i) \in T_{\gamma^i} \} |.
\]
For $x:=(x_0,x_1,x_2,\ldots,x_{n-1})$ the vector of the coefficients of $X$ and $b_i:=(\gamma^0,\gamma^i,\gamma^{2i},\ldots,$ $\gamma^{(n-1)i})$
 the vector of successive powers of $\gamma^i$, we have $X(\gamma^i)=b_i\cdot x$ for every $i\in\{1,\ldots,p-1\}$.
Therefore, if one chooses $m:=p-1$ and $S_i:=T_{\gamma^i}$, then, for every $i\in\{1,\ldots,m\}$, constraints $X(\gamma^i)\in T_{\gamma^i}$ and $b_i\cdot x\in S_i$ are equivalent. Hence, any solution $x=(x_0,x_1,\ldots,x_{n-1})$ to the max-LINSAT problem specified by the collection of vectors $b_1,\ldots,b_m$ and the collection of sets $S_1,\ldots,S_m$ as above yields the solution $X(z)=\sum_{j=0}^{n-1} x_j z^j$ to the corresponding OPI problem of the same objective value.

\paragraph{The parameter regime for OPI.}

Jordan et al.~show that a parameter regime of OPI where DQI achieves exponential speedups over the best known classical techniques is when  $n:=\lfloor p/10\rfloor+1$ and the sets $T_1,\ldots,T_{p-1}$ all have the same cardinality $|T_i|=\lfloor p/2\rfloor$. Recalling the vectors $b_1,\ldots,b_{p-1}$ given by the above reduction of OPI to max-LINSAT,  the matrix $B$ formed by these vectors is a Vandermonde matrix whose entries are
$B_{i,j}=\gamma^{ij}$ where $i=1,2,\ldots,p-1$ and $j=0,1,\ldots,n-1$.
The corresponding code $\ker B^\top$ is known as the narrow-sense Reed--Solomon code, and its minimum distance is $d^\perp=n+1$.

Because all the sets $T_i$ have the same cardinality $r:=\lfloor p/2\rfloor$ and because asymptotically 
 $\ell=\lfloor d^\perp/2-1\rfloor =\lfloor\frac{n-1}{2}\rfloor$ grows linearly with $m=p-1$, when reducing OPI to max-LINSAT with these parameters, one can apply \cite[Theorem 1.1]{jordan25:DQI-original}  and (\ref{eq:soverm}) to the resulting max-LINSAT problem. In particular, in the limit $p\rightarrow\infty$, because we have  $r/p =1/2$ and $\ell/m = (n/2)/p = 1/20$, from (\ref{eq:soverm}) we get that the expected ratio of satisfied constraints is
\[
\frac{\< f_{\mathrm{OPI}} \>}{p-1}
=
\left( \sqrt{\frac{1}{20} \left( 1 - \frac{1}{2} \right)} + \sqrt{\frac{1}{2} \left( 1 - \frac{1}{20} \right)} \right)^2
=
\frac12 + \frac{\sqrt{19}}{20},
\]
which is about $0.7179$.

To the best of our knowledge, the best classical polynomial-time algorithm for the OPI problem in this parameter regime is still the truncation heuristic, described in \cite{jordan25:DQI-original}. The truncation heuristic proceeds by first randomly choosing $n$ distinct elements $z\in\F_p^*$ and for each chosen element $z$ randomly choosing a value $t_z\in T_z$. There is a unique degree $n-1$ polynomial $X$ such that $X(z)=t_z$ for all chosen $z$, and this polynomial can be easily found using Gaussian elimination. Heuristically, for any element $z'\in\F_p^*$ that was not chosen, the probability that $X(z')\in T_{z'}$ is $r/p$, which means that the expected fraction of satisfied constraints, in the limit $p\rightarrow\infty$, is
\[
\frac{\< f_{\mathrm{OPI}} \>}{p-1} =\frac{n+(m-n)(r/p)}{p-1}=\frac{11}{20}=0.55.
\]

\paragraph{Runtime of DQI for OPI.}

As we already mentioned in Section~\ref{sec:introdqi}, Jordan et al.~propose a way to implement Stage 4 of the DQI framework in time $\widetilde\OO(p^3)$. First, they obtain the syndrome $B^\top y$ from the error $y$ by the standard matrix-vector multiplication, running in time $\widetilde\OO(p^2)$. Then, for Step (a) of Stage 4, which is, given the syndrome $B^\top y\in\F_p^n$, to find any $y'\in\F_p^m$ such that $B^\top y'=B^\top y$, they propose to use Gaussian elimination, whose runtime is $\widetilde\OO(p^3)$. And, for Step (b) of Stage 4, which is to error-correct $y'$ to the nearest codeword in $C^\perp$, they propose to use the Berlekamp--Massey algorithm, whose runtime is $\widetilde\OO(p^2)$.

Here we propose, first, to take the advantage of the fact that $B$ is a Vandermonde matrix and thereby to speed up the multiplication of $y$ by $B^\top$ using the fast number-theoretic transform (NTT), which requires only $\widetilde\OO(p)$ gates.
Second, for the task of uncomputing the error $y$ given the syndrome $B^\top y$, we propose not to divide this task into Steps (a), (b), (c) as described in Section~\ref{sec:introdqi}, but to do it directly using the $\widetilde\OO(p)$-time decoding algorithm for the narrow-sense Reed--Solomon codes based on continued fractions and the fast extended Euclidean algorithm~ \cite{reed78:GCDdecoding}\cite[Chapter 8]{aho74:algoDesign}. We provide more details about these techniques in Appendix~\ref{app:FastDecoding}.

As a result, now also having provided an $\widetilde\OO(p)$-time implementation of Stage 4---in addition to $\widetilde\OO(p)$-time implementations of earlier stages, as described before---we get the main result of this work.

\begin{thm}
\label{thm:main}
There is an $\OO(p\polylog p)$-time DQI-based algorithm for the OPI problem that, given a \QRAMc-access to input sets $T_1,\ldots,T_{p-1}$ of size $\lfloor p/2\rfloor$ each,  
finds a degree-$\lfloor p/10\rfloor$ polynomial $X$  that, in expectation, satisfies $(\frac12+\frac{\sqrt{19}}{20})p+o(p)$ constraints $X(z)\in T_z$, where  $z\in\F_p^*$.
The algorithm works in the \QRAMq model.
\end{thm}

\subsection{Concurrent work}

Concurrently and independently with this work, Khattar et al.~\cite{khattar25:optimized_DQI} also construct a nearly linear-time DQI algorithm for the OPI problem. The main difference between the two approaches is that our algorithm uses ``dense'' Dicke states $|D_k^m\>$, whereas Khattar et al.~employ ``sparse'' Dicke states $|SD_k^m\>$, the two being defined, respectively, as
\[
|D_k^m\>:=\frac{1}{\sqrt{\binom{m}{k}}} \sum_{\substack{\mu\in\{0,1\}^m\\|\mu|=k}}|\mu\>,
\qquad
|SD_k^m\>:=\frac{1}{\sqrt{\binom{m}{k}}}
\sum_{1\le c_1< c_2 < \ldots c_k \le m}|c_1,c_2,\ldots,c_k\>.
\]
The correspondence between the two representations is that, for a bit string $\mu\in\{0,1\}^m$ of Hamming weight $k$, the ordered tuple $(c_1,c_2,\ldots,c_k)$ encodes the positions of the $1$-bits in $\mu$.
In this work, we achieve linear runtime for Stages 1 and 2 of the DQI algorithm by effectively combining them and bypassing the quadratic-time preparation of dense Dicke states altogether. In contrast, Khattar et al.~achieve nearly-linear runtime for Stages 1 and 2 by constructing a procedure that prepares sparse Dicke states in nearly linear time.

Another notable difference is that the algorithm of Ref.~\cite{khattar25:optimized_DQI} never explicitly computes error vectors $y\in\F_p^m$ in their entirety.  Instead, it generates each error term $y_i\in\F_p$ sequentially and temporarily, using it immediately to update the syndrome register. This approach drastically reduces the memory required for the error register.

\paragraph{Organization.}

The rest of this paper is organized as follows. Section~\ref{sec:prelim} introduces the necessary preliminaries on computational models, Fourier transforms, and the Reed--Solomon code. In Section~\ref{sec:algo}, we describe our algorithm and analyze its runtime. Section~\ref{sec:performance} examines the number of constraints satisfied by the algorithm’s output. Finally, we leave various proofs to the appendix. 

\section{Preliminaries}
\label{sec:prelim}

\subsection{Computational model}
\label{sec:Comp_model}

We assume that the reader is familiar with basic concepts of quantum information and quantum computation (see \cite{NielsenChuang, watrous2018} for reference). Here we formalize the computational model that we are considering in this paper, that is, quantum circuits with quantum random access memory (QRAM). In fact, we will introduce three computational models: the standard circuit model (i.e., without random access memory), the model with quantum random read-only access to classical memory (\QRAMc), and the model with quantum random read-write access to quantum memory (\QRAMq).

For all these models, we can think of the total memory of the algorithm as consisting of two parts: the input registers, initialized with the algorithm's input, and the workspace registers, whose qubits are initially set to zero. Some of the workspace registers are also designated for returning the output and are measured at the end of the computation for this purpose.

\bigskip

In the \emph{standard quantum circuit model}, a quantum algorithm is represented as a quantum circuit composed of single- and two-qubit gates.%
\footnote{Or a universal subset thereof, with $\{H, S, T, \text{CNOT}\}$ being a common choice.}
The running time of the algorithm is measured by the total number of gates it employs to perform the computation. 
In this model, assuming that each of the $N$ input bits, $i_0, i_1, \ldots, i_{N-1}$, is relevant for producing the desired result, the computation must execute at least $\Omega(N)$ gates. To overcome this limitation, one may consider allowing random access to the input.

\bigskip

In the \emph{quantum random access to classical input} (\QRAMc) model, we assume that the single- and two-qubit quantum gates are restricted to the workspace registers, and the only way to access the input from the workspace is through the random-access gate $R_{\mathrm{C}}$, defined as follows. Among the workspace registers, we designate $\lceil\log_2 N\rceil$ qubits as the \emph{address register} and a single qubit as the \emph{interface register}, collectively referred to as the \emph{random access registers}. 
The gate $R_{\mathrm C}$ acts on the random access and the input registers as
\[
R_{\mathrm C}\colon |a,b\>|i_0,i_1,\ldots,i_{N-1}\> \mapsto |a,b\oplus i_a\>|i_0,i_1,\ldots,i_{N-1}\>
\]
where $a \in \{0,1,\ldots,N-1\}$ is an address encoded in $\lceil\log_2 N\rceil$ bits and $b \in \{0,1\}$.

The gate $R_{\mathrm C}$ can be applied when the random access registers are in a superposition over various pairs of $a$ and $b$, as well as when they are entangled with other workspace registers. However, note that in all cases the content of the input registers remains unchanged.
The running time of an algorithm in the \QRAMc model is measured by the total number of single-qubit gates, two-qubit gates, and $R_{\mathrm C}$-gates that it employs.

\bigskip

In the \emph{quantum random access to quantum memory} (\QRAMq) model, in contrast to the \QRAMc model, the single- and two-qubit quantum gates and the random access are not restricted to the workspace and the input registers, respectively. 
Moreover, the random access gate $R_{\mathrm Q}$ is now read-write rather than read-only.

In the \QRAMq model, we assume that the quantum memory consists of a total of $M + \lceil\log_2 M\rceil + 1$ qubits, where $\lceil\log_2 M\rceil + 1$ non-input qubits are designated as the random access registers, $\lceil\log_2 M\rceil$ qubits forming the address register and the remaining one qubit forming the interface register.

Unlike $R_{\mathrm C}$, the gate $R_{\mathrm Q}$ acts on all registers. Given that the address register contains $a \in \{0,1,\ldots,M-1\}$, the effect of $R_{\mathrm{Q}}$ is to swap the content of the interface register with the $a$-th qubit of the rest of the memory. That is, $R_{\mathrm Q}$ acts as
\begin{multline*}
R_{\mathrm Q}\colon|a,b\>|\ell_0,\ell_1,\ldots,\ell_{a-1},\ell_{a},\ell_{a+1},\ldots,\ell_{M-1}\>
\\\mapsto
|a,\ell_a\>|\ell_0,\ell_1,\ldots,\ell_{a-1},b,\ell_{a+1},\ldots,\ell_{M-1}\>,
\end{multline*}
where $b\in\{0,1\}$ is the content of the interface register and $|\ell_0,\ell_1,\ldots,\ell_{M-1}\>$ is the state of the rest of the memory before the application of $R_{\mathrm Q}$  (i.e., the memory excluding the random access registers).

\bigskip

The computational model employed for DQI-based algorithms in this paper, including the algorithm addressed in Theorem~\ref{thm:main}, is a \emph{hybrid} of the \QRAMc and \QRAMq models. Specifically, we assume that the input is stored in a read-only memory accessible in the \QRAMc fashion, while the remaining registers—the algorithm’s workspace—are accessible in the \QRAMq fashion.

\paragraph{Simulating QRAM in the standard model.}

The random access gate $R_{\mathrm Q}$ can be simulated within the standard quantum circuit model using $\OO(M)$ single- and two-qubit gates. Perhaps the simplest approach is to employ an address decoder that converts the address $a$ from binary to one-hot representation.

Recall that the Fredkin gate is a three-qubit gate that performs controlled-swap. It can be written as a product of three Toffoli gates and implemented using a constant number of single- and two-qubit gates. An address decoder requires $M$ ancilla qubits and can be implemented using $M-1$ Fredkin gates (see Appendix~\ref{app:QRAM} for details).

Once the one-hot representation of $a$ is obtained, it is straightforward to swap $b$ and $\ell_a$ using an additional $M$ Fredkin gates. Finally, by using another $M-1$ Fredkin gates, we apply the address decoder in reverse and restore the ancilla qubits to their original state.

Similarly, the random access gate $R_{\mathrm C}$ can be simulated using $\OO(N)$ single- and two-qubit gates. In this case, once we have the one-hot encoding of $a$, we replace the Fredkin gates with Toffoli gates.

A drawback of this approach is that it requires a linear number of ancilla qubits. For \QRAMq, this effectively doubles the total size of memory. Fortunately, both $R_{\mathrm Q}$ and $R_{\mathrm C}$ can also be implemented using only a logarithmic number of ancilla qubits—the same the number of qubits used to store $a$—while still maintaining a linear gate complexity (see Appendix~\ref{app:QRAM}).

\bigskip

Regarding the OPI problem, specifying the input $T_1, \ldots, T_{p-1}$ requires approximately $p^2$ bits in total. Consequently, a single \QRAMc-access to the input can be simulated in the standard model using $\OO(p^2)$ gates. However, the DQI-based algorithm in Theorem~\ref{thm:main} accesses the sets $T_1, \ldots, T_{p-1}$ sequentially, never using \QRAMc to access more than one at a time. Therefore, each \QRAMc-access performed by the algorithm can be simulated using only $\OO(p)$  single- and two-qubit gates.

The workspace used by the algorithm in Theorem~\ref{thm:main} has the same size as its running time in the quantum random access model. More precisely, it is $\OO(p\polylog p)$ qubits. Thus, each \QRAMq-access can be simulated in the standard model using $\OO(p \polylog p)$ gates.
As a result, the full algorithm can be simulated using $\OO(p^2 \polylog p)$ single- and two-qubit gates in total.

\subsection{Fourier and number-theoretic transforms}
\label{sec:NTT}

The algorithm we propose employs the quantum Fourier transform and the fast number-theoretic transform, the latter being a nearly linear-time algorithm for computing the discrete Fourier transform over $\F_p$. Here we briefly introduce these operations and describe their runtime.

\paragraph{Quantum Fourier Transform (QFT).}

For a positive integer $n$, the quantum Fourier transform $F_n$ is a unitary operation acting on $\lceil\log_2n\rceil$ qubits whose action on the computational basis states $|j\>$ is defined as 
\begin{alignat*}{2}
& F_n\colon |j\>\mapsto \frac{1}{\sqrt{n}}\sum_{k=0}^{n-1}\omega_n^{jk}|k\>
&\qquad& \text{when }0\le j\le n-1,
\\
& F_n\colon |j\>=|j\> &&
\text{when }n\le j \le 2^{\lceil\log_2n\rceil}-1,
\end{alignat*}
where $|j\>$ denotes the binary encoding of the integer $j$, and
$\omega_n:=\ee^{2\pi\mathsf{i}/n}$ is the $n$-th root of unity.

When $n$ is a power of $2$, $F_n$ can be implemented exactly using $\OO(\polylog n)$ single- and two-qubit quantum gates (see, e.g., \cite[Chapter 5]{NielsenChuang}). This exact implementation enables the computational primitive of quantum phase estimation, which, in turn, can be used to approximately implement $F_n$ for arbitrary $n$. In particular, given any $n$ and any $\epsilon>0$, one can implement a unitary operation $U$ such that $\|U - F_n\| \le \epsilon$ using $\OO(\polylog n \cdot \log 1/\epsilon)$ gates. For our application, an approximation error $\epsilon$ that scales inverse-polynomially with $n$ is sufficient, therefore we assume that $F_n$ can be implemented \emph{exactly} using $\OO(\polylog n)$ gates for any $n$.

\paragraph{Number-Theoretic Transform (NTT).}

Let $p$ be a prime and $\gamma$ a primitive element in the finite field $\F_p$. The number-theoretic transform $\NTT_{p,\gamma}$ is the linear map on $\F_p^{p-1}$ that maps  $x=(x_0,x_1,\ldots,x_{p-2})$ to $X=(X_0,X_1,\ldots,X_{p-2})$ where $X_j=\sum_{i=0}^{p-2}\gamma^{ij}x_i$ for all $j$. We write $X = \NTT_{p,\gamma}(x)$, or simply $X = \NTT(x)$ when $p$ and $\gamma$ are implicit or clear from context. The number-theoretic transform is invertible; specifically, $x_i=(p-1)^{-1}\sum_{j=0}^{p-2}\gamma^{-ij}X_j$ for all $i$, and we write $x=\NTT^{-1}_{p,\gamma}(X)$.

Given an input vector $x \in \F_p^{p-1}$ where each entry is represented using $\lceil\log_2 p\rceil$ bits, the transform $X = \NTT_{p,\gamma}(x)$ can be computed exactly using $\OO(p \polylog p)$ logical gates in the classical circuit model. This can be done using the Cooley--Tukey algorithm~\cite{cooley65:dft,heideman84:gauss}, possibly in combination with Rader's algorithm~\cite{rader68:dft} and fast Fourier transform over the complex numbers.%
\footnote{If all prime divisors of $p-1$ are of size $\OO(\polylog p)$, then the Cooley--Tukey algorithm alone suffices.}
We provide further implementation details in Appendix~\ref{app:FastNTT}.

\subsection{Narrow-Sense Reed--Solomon codes}
\label{sec:reed-solomon}

For a prime $p$, a primitive element $\gamma\in\F_p$, and an integer $n\le p-2$, the narrow-sense Reed--Solomon code~\cite{reed60:RS-codes} is defined as $\ker B^\top\subseteq \F_p^{p-1}$ where $B$
is given as  $B_{i,j}=\gamma^{ij}$ with $i=1,2,\ldots,p-1$ and $j=0,1,\ldots,n-1$. Alternatively, the code can be described by its generator matrix $G$ given as $G_{i,j}:=\gamma^{ij}$ with $i=1,2,\ldots,p-1$ and $j=1,2,\ldots,p-n-1$, so that $\im G = \ker B^\top$.
The minimum distance of this code is $d^\perp=n+1$.

\paragraph{Computing the syndrome from the error.}

Given an error $y\in\F_p^{p-1}$, one can compute the syndrome $B^\top y$ efficiently using the number-theoretic transform $\NTT:=\NTT_{p,\gamma}$. The matrix form of $\NTT$ has entries  $(\NTT)_{i,j}=\gamma^{ij}$ for $i,j=0,1,\ldots,p-2$. For $y=(y_1,\ldots,y_{p-2},y_{p-1})$, let $y^\rightarrow:=(y_{p-1},y_1,\ldots,y_{p-2})$ denote its circular shift to the right. 
Then $B^\top y = (\NTT \cdot y^\rightarrow)_{0..n-1}$,
where the subscript $0..n-1$ indicates the first $n$ entries of the resulting sequence. 
Therefore, the syndrome $B^\top y$ can be computed in time $\OO(p \polylog p)$.

\paragraph{Computing the error from the syndrome.}

The error $y$ can be recovered from the syndrome $B^\top y$ using a decoding algorithm based on continued fractions and the fast extended Euclidean algorithm~\cite{reed78:GCDdecoding,aho74:algoDesign}. 
 Specifically, given the syndrome $B^\top y$ for an unknown error $y$ of Hamming weight less than $d^\perp/2$, the error $y$ can be found in time $\OO(p\polylog p)$. 
Additional details on this decoding procedure for the narrow-sense Reed--Solomon codes are provided in Appendix~\ref{app:FastDecoding}.

\section{Algorithm}
\label{sec:algo}

In this section we show how to prepare the state
\beE
\label{eq:psi4}
|\algostate_4\> := (F_p^{-1})^{\otimes n}\sum_{k=0}^\ell \sum_{\substack{y\in\F_p^m\\|y|=k}}
\frac{w_k}{\sqrt{\binom{m}{k}}}\beta_y|B^\top y\>
\enE
with probability almost $1$,
where the coefficients $w_k$ and $\beta_y$ are given below, in
\eqref{eq:wk_def} and \eqref{eq:betay_def}, respectively. 
Then in Section~\ref{sec:performance} we will show that, upon measuring this state in the standard basis, we obtain $x\in\F_p^{n}$ such that, in expectation, the fraction of constraints satisfied by $x$ is asymptotically as in \eqref{eq:soverm} .

The running time of the algorithm is
\beE
\label{eq:general_runtime}
\OO\Big(m\frac{\sqrt{p}}{\min\{\sqrt{r},\sqrt{p-r}\}}\polylog (p+m) + T_B\Big)
\enE
where $T_B$ is the time necessary both to multiply vectors in $\F_p^m$ by $B^\top$ and to perform the error correction for the code $\ker B^\top$.%
\footnote{In this context, by ``performing the error correction'' we mean, given the syndrome $B^\top y\in\F_p^n$ corresponding to an error $y\in\F_p^m$ of Hamming weight less than $d^\perp/2$, obtaining $y$ itself.}
 The former term of the time complexity, namely, $\OO(m\sqrt{p/\min\{r,p-r\}}\polylog p)$, corresponds to Stages 1, 2, and 3 of the DQI algorithm, and it requires \QRAMc access to the classical input, but it does not require \QRAMq.%
\footnote{Without \QRAMc, the former term would become $\OO(mp\polylog p)$, still benefiting from our fast preparation of the superposition over Dicke states. In particular, in the binary case, when $p=2$, it would become $\OO(m)$.}

In the parameter regime relevant to the OPI problem, we have $m=p-1$, $r=\lfloor p/2\rfloor$, and $B$ is a Vandermonde matrix with entries $B_{i,j}=\gamma^{ij}$. Multiplication by $B^\top$ can be performed in time $\OO(p\polylog p)$ in the standard model using the fast number-theoretic transform. And $\ker B^\top$ is a narrow-sense Reed--Solomon code, for which we can perform error correction in time $\OO(p\polylog p)$ using \QRAMq. Hence, if we have quantum random access memory, then the preparation time of the state $|\algostate_4\>$ corresponding to the OPI problem is $\OO(p\polylog p)$.

\subsection{Final state}
\label{sec:finalState}

We choose $\ell:=\min\{\lfloor d^\perp/2\rfloor-1, m(1-r/p)\}$. When $r$ is large enough so that $\ell/m=1-r/p$, as described  in Section~\ref{sec:performance}, we will asymptotically be able to satisfy almost all the constraints.
Let us complete the definition of the state $|\algostate_4\>$ by specifying the coefficients $w_k$ and $\beta_y$ appearing in \eqref{eq:psi4}. 
First, define
\begin{alignat}{3}
\label{eq:q_def}
& \nmean &\,\,:=\,\,&\ell/m-m^{-1/2+c}
&& \text{for a constant }c>0,
\\ & \label{eq:wprime_k_def}
w'_k&\,\,:=\,\,& \sqrt{\nmean^k(1-\nmean)^{m-k}\binom{m}{k}} &\qquad&\text{for } k=0,1,\ldots,m,
\\  & \epsilon &\,\,:=\,\,& \sum_{k=\ell+1}^m ( w'_k)^2
 =\sum_{k=\ell+1}^m \nmean^k(1-\nmean)^{m-k}\binom{m}{k},
\\ & 
\label{eq:wk_def}
w_k&\,\,:=\,\,& w_k'/\sqrt{1-\epsilon} &\qquad&\text{for } k=0, 1,\ldots, \ell.
\end{alignat}
Note that $(w'_k)^2$ is the probability mass function of the binomial distribution, and $\epsilon$ is essentially a truncation error.
In our applications $\ell$ grows linearly with $m$, and therefore the value of $\epsilon$ is minuscule: in Claim~\ref{clm:tail_bound} we show that, assuming $\ell\ge 4m^{1/2+c}$, we have
$\epsilon \le  1/\ee^{m^{2c}\!/2}$.
Next, define 
\beE
\label{eq:g_i_def}
g_i(z):=\frac{f_i(z)-r/p}{\sqrt{r(p-r)/p}}
\qquad\text{where}\qquad
f_i(z):=\begin{cases}
1 & \text{if }z\in S_i, \\
0 & \text{if }z\notin S_i,
\end{cases}
\enE
and then
\beE
\label{eq:hatg_i_def}
\widehat{g}_i(z) := \frac{1}{\sqrt{p}}\sum_{z'\in\F_p}\omega_p^{zz'} g_i(z')
\enE
where $\omega_p:=\ee^{2\pi\ii/p}$, as before. 
Note that $\sum_{z\in\F_p}g_i(z)=0$, therefore $\widehat{g}_i(0)=0$. 
Finally, the coefficients $\beta_y$ are defined as
\beE
\label{eq:betay_def}
\beta_y := \prod_{\substack{i=1\\y_i\ne 0}}^m \widehat{g}_i(y_i).
\enE

\subsection{Stages One and Two: Quick superposition over Dicke states}
\label{sec:linear_Dicke}

Here we address joint implementation of Stages 1 and 2 of the DQI algorithm, in which we prepare a superposition over Dicke states $|D_k^m\>=\sum_{\substack{\mu\in\{0,1\}^m\\|\mu|=k}}|\mu\>/\sqrt{\binom{m}{k}}$.
 In particular, we show that it is easy to prepare the mask register in state 
\beE
\label{eq:psi2prime}
|\algostate_2'\>:=\sum_{k=0}^m w'_k|D_k^m\>
=\sum_{k=0}^m \sum_{\substack{\mu\in\{0,1\}^m\\|\mu|=k}} \frac{w'_k}{\sqrt{\binom{m}{k}}}|\mu\>
\enE
in time $\OO(m)$.
Recall from \eqref{eq:wprime_k_def} that $w'_k=\sqrt{\nmean^k(1-\nmean)^{m-k}\binom{m}{k}}$,
therefore
\[
|\algostate_2'\>
= \sum_{k=0}^m \sum_{\substack{\mu\in\{0,1\}^m\\|\mu|=k}} \sqrt{\nmean^k(1-\nmean)^{m-k}} |\mu\>
= (\sqrt{1-\nmean}|0\>+\sqrt{\nmean}|1\>)^{\otimes m}.
\]
This is a simple product state without any entanglement between the qubits and can be prepared by initializing all $m$ qubits to $|0\>$ and then applying the same Pauli-$Y$ rotation to each of them.
We have thus prepared a superposition over Dicke state bypassing explicit Dicke state preparation algorithms, such as in~\cite{bartschi22:Dicke}.

\paragraph{On the choice of $\nmean$.}

The state $|\algostate_2'\>$ that we prepare is similar to the general form \eqref{eq:stage2} used in~\cite{jordan25:DQI-original}, except that $|\algostate_2'\>$ also has a support on $\mu$ with Hamming weights larger than $\ell$. The probability weight of the total support on such $\mu$, namely, $\sum_{\mu\colon|\mu|>\ell}|\<\mu|\algostate'_2\>|^2$, is $\epsilon$, which is minuscule because we have chosen $\nmean$ to be such that the mean value $\nmean m$ of the corresponding binomial distribution over Hamming weights $k$ is $\ell-m^{1/2+c}$, and $m^{1/2+c}$ is substantially larger than the standard deviation $\sqrt{m\nmean(1-\nmean)}$ of that distribution.

More generally, in order to achieve the expected ratio $\<f\>/m$ of satisfied constraints to be as in \eqref{eq:soverm}, we will implicitly see in Section~\ref{ssec:WeightsToExp} that it is necessary and sufficient that the probability distribution over $k$, given by the probabilities $(w'_k)^2$, satisfy the following two informal constraints. First, we need that probability distribution places almost all probability weight on the values of $k$ between $\ell-o(\ell)$ and $\ell$, that is, that $\ell-o(\ell)\le k\le \ell$ with probability $1-o(1)$. And, second, we need that $w'_k$ does not change by much from one $k$ to the next one for almost all values of $k$, or, more precisely, we need that $\sum_k w'_kw'_{k-1}=1-o(1)$.

There are multiple ways to chose the probability distribution $\{(w_k')^2\}_k$ to satisfy these constraints: we have chosen the binomial distribution while \cite[Section~6.3]{jordan25:DQI-original} chooses the uniform distribution over the values of $k$ between $\ell-\sqrt{\ell}$ and $\ell$. We illustratively compare these two distributions in Figure~\ref{fig:weights_wk}.

\begin{figure}[!h]
\centering
\includegraphics[width=14cm]{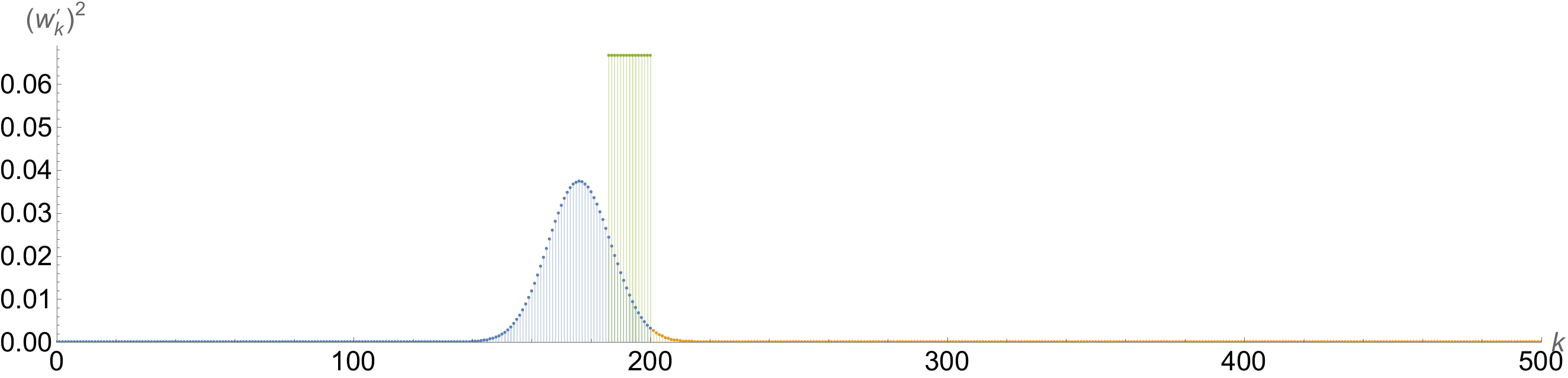}
\caption{
\footnotesize
Let $m=500$, $\ell=200$, and $c=0.01$. The binomial distribution $B(m,q)$ is shown in blue and orange, the orange part representing the part where we ``overshoot'' $\ell$. The uniform distribution used in \cite{jordan25:DQI-original}---they choose $(w'_k)^2=1/\lceil\sqrt{\ell}\rceil$ for $k$ with $\ell-\lceil\sqrt{\ell}\rceil < k \le \ell$ and $(w'_k)^2=0$ for other $k$---is shown in green.
}
\label{fig:weights_wk}
\end{figure}

\subsection{Stage Three: Exact Grover's algorithm}
\label{sec:Exact_Grover_main}

Here we will  present a procedure $G_i$ that, given \QRAMc access to the input, maps $|\widehat 0\>$ to $|G_i\>:=\sum_{z\in\F_p}g_i(z)|z\>$, where the amplitudes $g_i(z)$ are given in \eqref{eq:g_i_def} and can be expressed as
\beE
\label{eq:g_i_alt}
g_i(z)=\begin{cases}
\sqrt{p-r}/\sqrt{pr} & \text{if }z\in S_i, \\
-\sqrt{r}/\sqrt{p(p-r)} & \text{if }z\notin S_i.
\end{cases}
\enE
Note that $\sum_{z\in\F_p}g_i^2(z)=1$, therefore the state $|G_i\>$ has unit norm.

Given this procedure, our DQI algorithm proceeds as follows. First, for every qubit of the mask register, we introduce a $p$-dimensional qudit in the uniform superposition $|\widehat 0\>:=\sum_{x\in \F_p}|x\>/\sqrt{p}$, these qudits jointly forming the error register. 
Then, for each $i$ in succession, conditioned on $i$-th mask qubit being in state $|1\>$, we apply $G_i$ to the $i$-th error qudit. As a result, these two steps jointly map $|\algostate'_2\>$ to
\[
|\algostate_3'\> := \bigotimes_{i=1}^{m}\left(\sqrt{1-\nmean}|0\>|\widehat 0\>+\sqrt{\nmean}|1\>|G_i\>\right).
\]
Each instance of the state $|\widehat 0\>$ can be prepared in time $\OO(\polylog p)$, and it is left to describe an efficient implementation of $G_i$.

\bigskip

The proposed procedure for implementing $G_i$ is very reminiscent of the exact Grover's search algorithm~\cite{long01:grover,hoyer00:exactGrover,brassard02:exactGrover}.
The procedure alternates between two types of parametrized rotations. One can be thought of as a partial diffusion operator
\beE
\label{eq:diffusion_phi}
 D^\phi := |\widehat 0\>\<\widehat 0|+\ee^{\ii\phi} (I_p-|\widehat 0\>\<\widehat 0|),
\enE
performing the rotation around the uniform superposition $|\widehat 0\>$,
the Grover's diffusion operator being $D^\pi$.
The other one is 
\beE
\label{eq:oracle_psi}
\Xi^{\psi}_{i} := \ee^{\ii\psi}\sum_{x\in S_i}|x\>\<x| + \sum_{x\in \F_p\setminus S_i} |x\>\<x|,
\enE
which is analogous to a (partial) membership oracle in search algorithms.

Let us denote $\rho:= r/p$ for brevity, and for the sake of simplicity let as assume $\rho\le 1/2$, the case $\rho\ge 1/2$ being handled analogously. Let us define angles $\theta_\rho$, $\beta_\rho$, $\phi_\rho$, $\psi_\rho$ as
\beE
\label{eq:Grover_angles}
\begin{alignedat}{2}
&\theta_\rho := \arcsin\sqrt{\rho}\in(0,\pi/4),
&\qquad
& \beta_\rho:=\pi/2-2 \lfloor \pi/(4\theta_\rho)\rfloor \theta_\rho \in [0,2\theta_\rho),
\\
& \phi_\rho := \arccos\left(-\frac{\tan\beta_\rho}{\tan2\theta_\rho}\right)\in [\pi/2, \pi),
&&
\psi_\rho := \arccos \frac{\sin\beta_\rho}{\sin2\theta_\rho} \in (0, \pi/2].
\end{alignedat}
\enE
Also let us define a positive integer $\tau_\rho := \lfloor \pi/(4\theta_\rho)\rfloor$.
In Appendix~\ref{app:Gorver_angles} we prove the following claim, and give the interpretation of the above angles and $\tau_\rho$ in its context.

\begin{clm}
\label{clm:exact_Grover}
We have
$
\ee^{\ii(\pi-\psi_\rho)}\Xi_i^{\pi+2\psi_\rho} D^{\phi_\rho} (\Xi_i^\pi D^\pi)^{\tau_\rho} |\widehat 0\>= |G_i\>.
$%
\footnote{The right-most instance of $D^\pi$ in the expansion of $(\Xi_i^\pi D^\pi)^{\tau_\rho}$ is redundant, because $D^\pi|\widehat 0\> = |\widehat 0\>$.}
\end{clm}

This claim tells us how to implement $G_i$ given implementations of $D^\phi$ and $\Xi_i^\psi$. Since we are interested in implementing $G_i$ that is controlled by $i$-th mask qubit, we need efficient implementations of controlled-$D^\phi$ and controlled-$\Xi_i^\psi$. We describe implementation of both in Appendix~\ref{app:efficient_reflections}. In particular, controlled-$D^\phi$ can be implemented using $\OO(\polylog p)$ gates in the standard model.
As for the controlled-$\Xi_i^\psi$, for a general $\psi$, we show how to implement it using two calls to \QRAMc and 
$\OO(\polylog (p+m))$
 of single- and two-qubit gates. When $\psi=\pi$, a single call to \QRAMc suffices.

Note that $\theta_\rho=\OO(\sqrt{\rho})$, and thus the iteration count $\tau_\rho$ satisfies $\tau_\rho = \OO(\sqrt{1/\rho}) = \OO(\sqrt{p/r})$.
Hence, we can implement controlled-$G_i$ in time $\OO(\sqrt{p/r}\polylog p)$, or, more generally, if we do not assume $r/p\le 1/2$, in time
$\OO(\sqrt{p/\min\{r,p-r\}}\polylog p)$.
Because we need to apply controlled-$G_i$ for every $i$, this shows that we can map $|\algostate_2'\>$ to $|\algostate_3'\>$ in time $\OO(m\sqrt{p/\min\{r,p-r\}}\polylog p)$.

\paragraph{Approximate implementation of $G_i$ for OPI.}

In the parameter regime relevant for the OPI problem, we have $\rho = 1/2 - 1/(2p)$, for which $\tau_\rho=1$, as we detail in Appendix~\ref{sec:Grover_Single}.
Hence, Claim~\ref{clm:exact_Grover} implies  that
$-\ee^{-\ii\psi_\rho}\Xi_i^{\pi+2\psi_\rho} D^{\phi_\rho} \Xi_i^\pi |\widehat 0\>= |G_i\>$.
In particular, for this value of $\rho$, we have
$\psi_\rho = \arccos\frac{1}{\sqrt{p^2-1}}\approx \pi/2$, therefore, aproximately, 
$\ii D^{\phi_\rho} \Xi_i^\pi |\widehat 0\>\approx |G_i\>$. Moving $D^{\phi_\rho}$ to the other side, we have
$\ii \Xi_i^\pi |\widehat 0\>\approx D^{-\phi_\rho}|G_i\> = \ee^{-\ii\phi_\rho}|G_i\>$, where the latter equality is from \eqref{eq:diffusion_phi}, as we detail in Appendix~\ref{app:Gorver_angles}.
Furthermore, we have $\phi_\rho = \arccos\frac{-1}{p^2-1}\approx\pi/2$, therefore
$\ii \Xi_i^\pi |\widehat 0\>\approx  -\ii |G_i\>$, which means $-\Xi_i^\pi |\widehat 0\>\approx |G_i\>$.

We show in Appendix~\ref{sec:Grover_Single} that, for the OPI problem,  approximately computing $|G_i\>$ as $|\widetilde G_i\>:=-\Xi_i^\pi |\widehat 0\>$ is good enough, that is, we show that, for 
\[
|\widetilde\algostate_3'\> := \bigotimes_{i=1}^{m}\left(\sqrt{1-\nmean}|0\>|\widehat 0\>+\sqrt{q}|1\>|\widetilde{G}_i\>\right),
\]
we have $\||\widetilde\algostate_3'\>-|\algostate_3'\>\|=\OO(1/\sqrt{p})$.
Therefore, Theorem~\ref{thm:main} is still satisfied if we implement the procedure $G_i$ approximately as $-\Xi_i^\pi$.
While the precise implementation of each $G_i$ requires three calls to \QRAMc, the approximate one requires only a single call. Hence, this approximation reduces the total number of calls to \QRAMc threefold.

\subsection{Stage Three: Quantum Fourier Transform}

As the next step of the algorithm, we apply the quantum Fourier transform $F_p$ to each of the $m$ error qudits of the state $|\algostate'_3\>$. 
Similarly to the definition of $|G_i\>$, define $|\widehat G_i\>:=\sum_{z\in\F_p}\widehat{g}_i(z)|z\>$, where $\widehat{g}_i(z)$ is given in \eqref{eq:hatg_i_def}.
Note that
\[
F_p|G_i\>
= \sum_{z'\in\F_p}g_i(z')F_p|z'\>
= \frac{1}{\sqrt{p}}\sum_{z,z'\in\F_p}g_i(z')\omega_p^{zz'}|z\>
= \sum_{z\in\F_p}\widehat{g}_i(z)|z\>
=|\widehat{G}_i\>,
\]
and also $F_p|\widehat 0\> = |0\>$.
Thus, by applying $F_p^{\otimes m}$ to the error register of $|\algostate'_3\>$, we obtain
\[
|\algostate_3''\>:=\bigotimes_{i=1}^{m}(\sqrt{1-\nmean}|0\>|0\>+\sqrt{\nmean}|1\>|\widehat{G}_i\>).
\]

Next, recall that $\widehat{g}_i(0)=0$, and therefore $\<0|\widehat G_i\> = 0$.
Hence, for each $i$, we can flip the $i$-th mask qubit conditioned on the $i$-th error qudit being in non-zero state, and, by doing so, we disentangle the mask register by returning it to its initial all-zeros state. Then we discard the mask register, and we are left with the error register in the state
\[
|\algostate_3'''\>:=\bigotimes_{i=1}^{m}(\sqrt{1-\nmean}|0\>+\sqrt{\nmean}|\widehat{G}_i\>)
= \sum_{k=0}^m \sum_{\substack{y\in\F_p^m\\|y|=k}}
\frac{w'_k}{\sqrt{\binom{m}{k}}}\beta_y|y\>.
\]

\paragraph{Trimming the heavy errors.}

From practical point of view, we can directly proceed to Stage 4, where we apply the error correction to each $|y\>$ in the superposition. The support of $|\algostate_3'''\>$ on $y$ with the Hamming weight more that $\ell$, for which we cannot guarantee the correct performance of the error correction, is minuscule. In particular, it has the probability weight $\epsilon$. Thus, we can just ignore those $y$ and carry on.%
\footnote{Doing so also lets us remove the term $\polylog m$ from \eqref{eq:general_runtime}, though in most cases we expect this term to be dominated by other terms anyways.}

To keep the analysis clear, however, let us add an ancilla register in which compute the Hamming weight of $y$, then add an additional ancilla qubit, in which we compute if this Hamming weight is at most $\ell$, and finally uncompute the Hamming weight and discard that ancilla register. This all can be done in time $\OO(m\polylog (p+m))$. Then we measure the remaining ancilla qubit, the measurement with probability $1-\epsilon$ saying that the Hamming weight was at most $\ell$ and collapsing the error register to
\[
|\algostate_3\>:=\sum_{k=0}^\ell \sum_{\substack{y\in\F_p^m\\|y|=k}}
\frac{w_k}{\sqrt{\binom{m}{k}}}\beta_y|y\>.
\]
This expression for $|\algostate_3\>$ exactly matches \eqref{eq:stage3},
and we can proceed to Stage 4.

\subsection{Stage Four and Improvements for OPI}
\label{ssec:algo_final_stage}

In the final stage of the algorithm, we essentially have to replace every $|y\>$ in the support of $|\algostate_3\>$ by $|B^\top y\>$, and, as in \eqref{eq:general_runtime}, we denote by $T_B$ the time necessary to do this. This can be achieved as follows. First, we map $|y\>$ to $|y\>|B^\top y\>$ using the matrix-vector multiplication. And then, using the error correction for the code $C^\perp = \ker B^\top$, we can obtain the error $y$ from the syndrome $B^\top y$ alone, thus uncomputing the error register. In particular, the procedure for obtaining $y$ from $B^\top y$, as proposed by Jordan et al.~and already described n Section~\ref{sec:introdqi}, can be decomposed into three steps:
\begin{enumerate}
\item[(a)] Using Gaussian elimination, find any $y'\in\F_p^m$ such that $B^\top y'=B^\top y$;
\item[(b)] Using the error correction, find $c\in C^\perp$ such that $|y'-c|<d^\perp/2$;
\item[(c)] Return $y'-c$ as $y$.
\end{enumerate}
Thus, we obtain the state 
\beE
\label{eq:psi4prime}
|\algostate'_4\>:=\sum_{k=0}^\ell \sum_{\substack{y\in\F_p^m\\|y|=k}}
\frac{w_k}{\sqrt{\binom{m}{k}}}\beta_y|B^\top y\>
\enE
from $|\algostate_3\>$ in time $\OO(T_B)$.
Finally, we apply the inverse quantum Fourier transform $F_p^{-1}$  to each of $n$ qudits, obtaining $|\algostate_4\>$, given in \eqref{eq:psi4}. Since $n\le m$, the complexity of this final task is $\OO(m\polylog p)$.
Accumulating the time complexities of all the stages, we can see that the total running time is as given in \eqref{eq:general_runtime}.

\paragraph{Optimal Polynomial Intersection.}

Unlike for the earlier stages, for Stage 4 of the DQI framework we do not provide general improvements over the work by Jordan et al. However, let us consider here the scenario when $C^\perp$ is the narrow-sense Reed--Solomon code, which is the relevant code for the OPI problem. For this scenario, we propose a faster way to transform $|\algostate_3\>$ into $|\algostate_4'\>$ than Jordan et al.~do.

Jordan et al.~implement Stage 4 by following the general framework described above and, within it, using the Berlekamp--Welch algorithm for the error correction. The most expensive steps of that implementation are Gaussian elimination and the Berlekamp--Welch algorithm, both running in time $\OO(p^3\polylog p)$ given random access memory.

In this paper, we first take the advantage of the fact  the multiplication of $y$ by $B^\top$ can be performed using the fast number-theoretic transform. Indeed, as stated in Section~\ref{sec:reed-solomon} and further elaborated in Appendix~\ref{app:FastNTT}, we can use NTT to map $|y\>$ to $|y\>|B^\top y\>$ using $\widetilde\OO(p)$ gates in the standard quantum circuit model.

Second, for the task of uncomputing the error $y$ given the syndrome $B^\top y$, instead of dividing this task into Steps (a), (b), (c) as in the general setting described above, we use the fast decoding based on continued fractions and the fast extended Euclidean algorithm~\cite{reed78:GCDdecoding,aho74:algoDesign}. As stated in Section~\ref{sec:reed-solomon}, this decoding can be performed in time $\OO(p\polylog p)$ using \QRAMq. The further details of this decoding algorithm are given in Appendix~\ref{app:FastDecoding}.

\section{Performance analysis}
\label{sec:performance}

As given in \eqref{eq:psi4}, the final state of the DQI algorithm is
\[
|\algostate_4\> = (F_p^{-1})^{\otimes n}\sum_{k=0}^\ell \sum_{\substack{y\in\F_p^m\\|y|=k}}
\frac{w_k}{\sqrt{\binom{m}{k}}}\beta_y|B^\top y\>,
\]
where the coefficients $\beta_y$ depend on the input sets $S_1,\ldots,S_m$. The amplitudes $w_k$, in principle, can be chosen arbitrarily when designing the algorithm, as long as $\sum_k|w_k|^2=1$, and in Section~\ref{sec:algo} we show how for a certain choice of $w_k$ we can prepare the state $|\algostate_4\>$ more efficiently than before. 

In this section, we first state how these amplitudes relate to expected number of satisfied constraints. This relation was already given by Jordan et al., and, for completeness, we present its proof in Appendix~\ref{app:expected_ratio_proof}.

Then we focus on our choice of amplitudes $w_k$, for which we can can prepare the superposition over Dicke states $\sum_{k=0}^\ell w_k|D_k^m\>\approx |\algostate_2'\>$ in nearly linear time, with $|\algostate_2'\>$ being an intermediary state on our way to preparing $|\algostate_4\>$.%
\footnote{The state $|\algostate_2'\>=\sum_{k=0}^m w'_k|D_k^m\>$ defined in \eqref{eq:psi2prime} with $w_k'$ as in \eqref{eq:wprime_k_def} can be prepared in linear time.}
We show that our choice of amplitudes $w_k$ yields asymptotically as good a result as the best possible choice of $w_k$.

\subsection{Relation between amplitudes and expected number of constraints}
\label{ssec:WeightsToExp}

For a given max-LINSAT problem, let $\overline f(x)$ be the number of constraints $b_i\cdot x\in S_i$ satisfied by an assignment $x\in\F_p^n$ of optimization variables, where $i=1,\ldots,m$. That is, $\overline f(x)=|\{i\colon b_i\cdot x\in S_i\}|$. The DQI algorithm produces the state $|\algostate_4\>$, which is a superposition over such assignments.

Let $\Phi_f$ be the observable that for any superposition over $|x\>$ yields the expected value of $\overline f$ when this superposition is measured. That is,  
\[
\Phi_f: = \sum_{x\in\F_p^n} \overline f(x)|x\>\<x|.
\]
The following theorem from \cite{jordan25:DQI-original} (see Appendix~\ref{app:expected_ratio_proof} for a proof) gives the expected number of satisfied constraints by $x$ resulting form measuring the state $|\algostate_4\>$.

\begin{thm}
\label{thm:WeightsToExp}
Let $|\algostate_4\>$ be the final state of the DQI algorithm, as in \eqref{eq:psi4}.
We have 
\[
\<\algostate_4|\Phi_f|\algostate_4\> = \frac{mr}{p} + \frac{p-2r}{p}
\sum_{k=1}^\ell |w_k^2| k
+   \frac{2\sqrt{r(p-r)}}{p}
 \sum_{k=0}^{\ell-1} \Re(w_k^* w_{k+1})
\sqrt{(k+1)(m-k)}.
\]
\end{thm}

We remark that the theorem is true even when amplitudes $w_k$ are not chosen as in \eqref{eq:wk_def}, as long as the vector $w:=(w_0,w_1,\ldots,w_\ell)\in\C^{\ell+1}$ formed by these amplitudes has a unit norm. In particular, if $|\algostate_4\>$ corresponds to such a vector $w$, then Theorem~\ref{thm:WeightsToExp} states that
\[
\<\algostate_4|\Phi_f|\algostate_4\>  = \rho m+
\sqrt{\rho(1-\rho)}\, w^* A^{(m,\ell,d)}w 
\]
where $\rho:=r/p$ and $A^{(m,\ell,d)}$ is a symmetric real $(\ell+1)\times(\ell+1)$ tridiagonal matrix 
\[
A^{(m,\ell,d)}:=
\left[
\begin{array}{ccccc}
0 & a_1 & & & \\
a_1 & d & a_2 & & \\
& a_2 & 2d & \ddots & \\
& & \ddots & & a_\ell \\
& & & a_\ell & \ell d \\
\end{array}
\right]
\]
with $a_k:=\sqrt{k(m-k-1)}$ and $d:=(1-2\rho)\big/\sqrt{\rho(1-\rho)}$. Hence, the choice of $w$ that maximizes the expected number of satisfied constraints $\<\algostate_4|\Phi_f|\algostate_4\>$ is an eigenvector of $A^{(m,\ell,d)}$ corresponding to its largest eigenvalue.%
\footnote{Note that the largest eigenvalue of $A^{(m,\ell,d)}$ equals $\|A^{(m,\ell,d)}\|$ if and only if $d\ge 0$.}

Let $\lambda:=\ell/m\le 1/2$. Because of our definition
$\ell=\min\{\lfloor d^\perp/2\rfloor-1, m(1-\rho)\}$, we are guaranteed that $\lambda+\rho\le 1$, and Jordan et al.~\cite[Lemma 6.3]{jordan25:DQI-original} show that, given this assumption,%
\footnote{They assume $d\ge-(1-2\lambda)/\sqrt{\lambda(1-\lambda)}$, which is equivalent to $\lambda+\rho\le 1$.} 
 the largest eigenvalue of $A^{(m,\ell,d)}$ is at most $2\sqrt{m}+\ell d+2\sqrt{\ell(m-\ell)}$, and thus one can see that the expected fraction of satisfied constraints can be upper bounded as 
\begin{align*}
\frac{\<\algostate_4|\Phi_f|\algostate_4\>}{m}
& \le 
\big(\sqrt{\lambda(1-\rho)}+\sqrt{\rho(1-\lambda)}\big)^2
+ \frac{1}{2\sqrt{m}}.
\end{align*}
Asymptotically, as $m$ goes to infinity and $\ell$ grows linearly with $m$, this upper bound goes to $(\sqrt{\lambda(1-\rho)}+\sqrt{\rho(1-\lambda)})^2$. 

We next show that $w$ with the amplitudes $w_k$ corresponding our construction---given in \eqref{eq:wk_def}---asymptotically matches the upper bound. In particular, we show that assuming that $4m^{1/2+c}\le\ell\le m/2$ and that there exists a constant $\lambda_\star$ such that $\ell\ge \lambda_\star m$, the expected fraction of satisfied constraints by our algorithm is
\begin{align*}
\frac{\<\algostate_4|\Phi_f|\algostate_4\>}{m}
 \ge 
\big(\sqrt{\lambda(1-\rho)}+\sqrt{\rho(1-\lambda)}\big)^2
- \frac{6+2/\sqrt{\lambda_\star}}{m^{1/2-c}}
- \frac{4}{\ee^{m^{2c}/2}}.
\end{align*}
Note that, by assuming $4m^{1/2+c}\le\ell\le m/2$, we implicitly assume $m\ge 64$.
We prove this inequality in the next section.

\subsection{Optimality of our amplitudes}

In our construction, $w_k\ge 0$ for all $k$, and thus $\Re(w_k^*w_{k-1})=w_kw_{k-1}$ and $|w_k^2|=w_k^2$ for all $k$.
For convenience, let us denote $\Delta:=\ell-\nmean m=m^{1/2+c}$. In Appendix~\ref{app:binomial} we show that almost all the probability mass of the probability distribution over $k$ with the probability mass function $w_k^2$ is concentrated between $\ell-2\Delta$ and $\ell$, with its mode and mean both being approximately $\ell-\Delta$.

Let's see how diagonal-adjacent elements $a_1$, $a_2$, $a_3$, $\ldots$ and diagonal elements $d$, $2d$, $3d$, $\ldots$ of $A^{(m,\ell,d)}$ contribute to $w^* A^{(m,\ell,d)}w$.

\paragraph{Off-diagonal contribution.}

We lower bound $\sum_{k=1}^{\ell} a_k w_k w_{k-1}$ as follows.

First, because $\ell< m/2$, we have $a_0<a_1<\ldots < a_\ell$  and, in turn, we have
\[
\sum_{k=1}^{\ell} a_k w_k w_{k-1}
 \ge \sum_{k=1}^{\ell} a_k \min\{w_k^2, w_{k-1}^2\}   
 \ge a_{\ell-2\Delta} \sum_{k=\ell-2\Delta}^{\ell} \min\{w_k^2, w_{k-1}^2\}.
\]
In Claim~\ref{clm:mode_k} of Appendix~\ref{ssec:mode} we show that $\kappa:=\lceil \nmean(m+1)\rceil - 1$ is the mode of the probability distribution $(w_k^2)_k$ and that we have $w_{k-1}^2 \le w_{k}^2$ for $k\le \kappa$ and $w_{k-1}^2 \ge w_{k}^2$ for $k > \kappa$. Hence,
\[
\sum_{k=\ell-2\Delta}^{\ell} \min\{w_k^2, w_{k-1}^2\}
=
\sum_{k=\ell-2\Delta}^{\kappa} w_{k-1}^2
+ \sum_{k=\kappa+1}^{\ell} w_k^2
= \sum_{k=\ell-2\Delta-1}^{\ell} w_{k}^2 - w_\kappa^2.
\]
We lower bound $\sum_{k=\ell-2\Delta-1}^{\ell} w_{k}^2$ in Claim~\ref{clm:tail_bound} using the Chernoff bound and we
 upper bound the probability weight of mode in $w_\kappa^2$ in Claim~\ref{clm:mode_mass} using Stirling's formula.

\bigskip

Recall from \eqref{eq:wprime_k_def}--\eqref{eq:wk_def} that
$(w'_k)^2 = \nmean^k(1-\nmean)^{m-k}\binom{m}{k}$
for $k=0,1,\ldots,m$, that
$\epsilon = \sum_{k=\ell+1}^m ( w'_k)^2$,
and
$w_k^2= (w_k')^2/(1-\epsilon)$ for  $k=0, 1,\ldots, \ell$.
For convenience, let us define $\epsilon':= \ee^{-m^{2c}/2}$.
Claim~\ref{clm:tail_bound} states that both $\epsilon\le \epsilon'$ and $\sum_{k=0}^{\ell-2\Delta-1} ( w'_k)^2 \le \epsilon'$.
Therefore, first, we get that
\beE
\label{eq:most_w_mass}
\sum_{k=\ell-2\Delta-1}^{\ell} w_{k}^2
\ge \sum_{k=\ell-2\Delta-1}^{\ell} (w'_{k})^2
= 1-\sum_{k=0}^{\ell-2\Delta-2} ( w'_k)^2-\sum_{k=\ell+1}^m ( w'_k)^2
\ge 1-2\epsilon'.
\enE

Regarding the probability mass of the mode $\kappa$, 
Claim~\ref{clm:mode_mass} says that (assuming $m\ge 7$), we have $(w'_\kappa)^2\le 3/\sqrt{m\nmean(1-\nmean)}$. 
Note that, because $m^{2c}\ge 1$, we have 
$\epsilon'\le 1/\sqrt{\ee}$, and therefore $1/(1-\epsilon')\le 1+3\epsilon'$ and, in turn, $w_\kappa^2\le 3(1+3\epsilon')/\sqrt{m\nmean(1-\nmean)}$.

\bigskip

Now we need to lower bound $a_{\ell-2\Delta}$.
Let us recall that $\lambda=\ell/m$, and 
thus we have
\begin{multline*}
a_{\ell-2\Delta}
 = \sqrt{(\ell-2\Delta)(m-\ell+2\Delta-1)}
 \ge \sqrt{(\ell-2\Delta)(m-\ell)}
\\
 = m\sqrt{\lambda(1-\lambda)} \cdot \sqrt{1-\frac2{\lambda m^{1/2-c}}}.
\ge m\sqrt{\lambda(1-\lambda)} \cdot \left(1-\frac2{\lambda m^{1/2-c}}\right).
\end{multline*}
We have thus shown that 
\begin{align*}
\sum_{k=1}^{\ell} \frac{a_k w_k w_{k-1}}{m}
& \ge
\sqrt{\lambda(1-\lambda)}
\left(1-\frac2{\lambda m^{1/2-c}}-2\epsilon'-\frac{3(1+3\epsilon')}{\sqrt{m\nmean(1-\nmean)}}\right).
\end{align*}
Recall that $\nmean=\lambda-1/m^{1/2-c}$, from which we get that
\[
\frac{\sqrt{\lambda(1-\lambda)}}{\sqrt{\nmean(1-\nmean)}} 
\le
\frac{\sqrt{\lambda}}{\sqrt{\nmean}} 
= \frac1{\sqrt{1-1/(\lambda m^{1/2-c})}}
\le \frac1{1-1/(\lambda m^{1/2-c})}.
\]
The assumption $\ell\ge 4m^{1/2+c}$ is equivalent to $\lambda m^{1/2-c} \ge 4$, and therefore we have
$\sqrt{\lambda(1-\lambda)}/\sqrt{\nmean(1-\nmean)} \le 4/3$.
We also have $\sqrt{\lambda(1-\lambda)}\le 1/2$ and $\sqrt{1-\lambda}\le 1$, therefore
\begin{align}
\notag
\sum_{k=1}^{\ell} \frac{a_k w_k w_{k-1}}{m}
& \ge
\sqrt{\lambda(1-\lambda)}
-\frac2{\sqrt{\lambda} m^{1/2-c}}-\epsilon'-\frac{4(1+3\epsilon')}{\sqrt{m}}
\\ \label{eq:off-diag-bound}& \ge
\sqrt{\lambda(1-\lambda)}
-\frac{4+2/\sqrt{\lambda_\star}}{ m^{1/2-c}}-3\epsilon',
\end{align}
where for the last inequality we have used the implicit assumption that $m\ge 64$ and the assumption that $\lambda \ge \lambda_\star$.

\paragraph{Diagonal contribution.}

Now let us address the diagonal of $A^{(m,\ell,d)}$. 
We want to lower bound $d \sum_{k=0}^\ell k w_k^2$.
When $d\le 0$, we require the upper bound
\[
\sum_{k=0}^\ell kw_k^2
\le \ell \sum_{k=0}^\ell w_k^2 
= \ell,
\]
and, when $d\ge 0$, we require the lower bound 
\[
\sum_{k=0}^\ell kw_k^2
\ge (\ell-2\Delta)\sum_{k=\ell-2\Delta}^\ell w_k^2
\ge \ell\sum_{k=\ell-2\Delta}^\ell w_k^2
-2\Delta \sum_{k=0}^\ell w_k^2
\ge
\ell-m\epsilon'-2m^{1/2+c},
\]
where we have used inequalities \eqref{eq:most_w_mass} and $2\ell\le m$ and the definition $\Delta=m^{1/2+c}$.
Regardless whether $d$ is positive or negative, we have
\beE
\label{eq:diag-bound}
d\sum_{k=0}^\ell \frac{k w_k^2}{m}
\ge d\lambda - |d|(\epsilon'+2m^{-1/2+c})
\ge \frac{(1-2\rho)\lambda - \epsilon'-2m^{-1/2+c}}{\sqrt{\rho(1-\rho)}},
\enE
where for the last inequality we used that $d=(1-2\rho)/\sqrt{\rho(1-\rho)}$ and, thus, $|d|\le 1/\sqrt{\rho(1-\rho)}$.

\paragraph{Combining the contributions.}

By combining the contributions from the diagonal and the diagonal-adjacent entries of $A^{(m,\ell,d)}$---Eqns.~\eqref{eq:diag-bound} and \eqref{eq:off-diag-bound}, respectively---we get
\begin{align*}
\frac{\<\algostate_4|\Phi_f|\algostate_4\>}{m} 
&= \rho +
\sqrt{\rho(1-\rho)}\frac{w^* A^{(m,\ell,d)}w}{m}
\\ & \ge 
\rho+
(1-2\rho)\lambda - \epsilon'-\frac{2}{m^{1/2-c}}
\\ & \qquad 
+2
\sqrt{\rho(1-\rho)}\left(
\sqrt{\lambda(1-\lambda)}
-\frac{4+2/\sqrt{\lambda_\star}}{m^{1/2-c}}
- 3\epsilon'
\right)
\\ & \ge 
\big(\sqrt{\lambda(1-\rho)}+\sqrt{\rho(1-\lambda)}\big)^2
- \frac{6+2/\sqrt{\lambda_\star}}{m^{1/2-c}}
- 4\epsilon',
\end{align*}
where we have used $2\sqrt{\rho(1-\rho)}\le 1$ for the last inequality.

\section*{Acknowledgements}

I thank Asmae Benhemou, Harry Buhrman, Simon Burton, Ben Criger, Niklas Galke, Alexandre Krajenbrink, and Matthias Rosenkranz for valuable discussions and feedback.

\newpage

\addcontentsline{toc}{section}{References}
{
\small
%
\newcommand{\etalchar}[1]{$^{#1}$}

}

\appendix

\section[Implementation of Gi]{Implementation of $G_i$}
\label{app:Gorver}

Here we present an efficient implementation of the procedure $G_i$ that maps $|\widehat 0\>$ to $|G_i\>$. First, in Appendix~\ref{app:Gorver_angles}, we prove Claim~\ref{clm:exact_Grover}, which shows how $G_i$ can be performed by alternating between parametrized rotations $D^\phi$ and $\Xi_i^\psi$. In Appendix~\ref{sec:Grover_Single}, we briefly describe the parameter regime relevant to the OPI problem. Finally, in Appendix~\ref{app:efficient_reflections}, we provide efficient implementations of $D^\phi$ and $\Xi_i^\psi$.

\subsection{Proof of Claim~\ref{clm:exact_Grover}}
\label{app:Gorver_angles}

Recall that $\rho=r/p$, which we assume to be at most $1/2$, and recall from \eqref{eq:Grover_angles} the angles $\theta_\rho$, $\beta_\rho$, $\phi_\rho$, $\psi_\rho$ and the iteration count $\tau_\rho$, all
parametrized by $\rho$. Also recall the rotations $D^\phi$ and $\Xi_i^\psi$ from \eqref{eq:diffusion_phi} and \eqref{eq:oracle_psi}, respectively. In this section we prove the following.

\bigskip

\noindent
{\bf{}Claim~\ref{clm:exact_Grover}} (restated){\bf{}.}
\emph{
We have
$
\ee^{\ii(\pi-\psi_\rho)}\Xi_i^{\pi+2\psi_\rho} D^{\phi_\rho} (\Xi_i^\pi D^\pi)^{\tau_\rho} |\widehat 0\>= |G_i\>.
$
}

\bigskip

In the proof, for simplicity, let us drop the subscript $\rho$ from the angles $\theta$, $\beta$, $\phi$, $\psi$. We will show how the values of these angles, given in \eqref{eq:Grover_angles}, emerge when we are attempting to map $|\widehat 0\>$ to $|G_i\>$. The same way, we drop the subscript $\rho$ from the iteration count $\tau$.

Let us define unit vectors
\[
|\Yes_i\>:=
\frac1{\sqrt{r}}\sum_{z\in S_i}|z\>
\qqAnd
|\No_i\>:=
\frac1{\sqrt{p-r}} \sum_{z\in \F_p\setminus S_i}|z\>.
\]
Then $\theta=\arcsin\sqrt\rho$ of \eqref{eq:Grover_angles} is simply the angle between $|\widehat 0\>$ and $|\No_i\>$ and, equivalently, the angle between $|G_i\>$ and $|\Yes_i\>$. We have
\begin{equation}
\label{eq:0andSi}
\begin{cases}
|\widehat 0\> = \sin\theta|\Yes_i\>+\cos\theta|\No_i\>, \\
|G_i\> = \cos\theta|\Yes_i\>-\sin\theta|\No_i\>,
\end{cases}
\qqAnd
\begin{cases}
|\Yes_i\> = \sin\theta|\widehat 0\>+\cos\theta|G_i\>, \\
|\No_i\> = \cos\theta|\widehat 0\>-\sin\theta|G_i\>,
\end{cases}
\end{equation}
therefore $\{|\widehat 0\>,|G_i\>\}$ and $\{|\Yes_i\>,|\No_i\>\}$ are two orthonormal basis of the same two-dimensional subspace of $\C^p$, which we denote by $\cH_i$. Effectively, we will only have to consider states in this subspace.

The reflections $D^\pi$ and $\Xi_i^\pi$ keep the space $\cH_i$ invariant, and their restrictions to $\cH_i$ are, respectively,
\[
D^\pi |_{\cH_i} = |\widehat 0\>\<\widehat 0| - |G_i\>\<G_i|
\qqAnd
\Xi_i^\pi |_{\cH_i} = |\No_i\>\<\No_i| - |\Yes_i\>\<\Yes_i|.
\]
Their product in the basis  $\{|\widehat 0\>,|G_i\>\}$ is 
\begin{align*}
\Xi_i^\pi D^\pi |_{\cH_i}
&= \big( |\No_i\>\<\No_i| - |\Yes_i\>\<\Yes_i|\big)
\big(|\widehat 0\>\<\widehat 0| - |G_i\>\<G_i|\big)
\\ & =
\cos 2\theta \big( |\widehat 0\>\<\widehat 0| + |G_i\>\<G_i|\big)
+ \sin 2\theta \big( |\widehat 0\>\<G_i| - |G_i\>\<\widehat 0|\big),
\end{align*}
where the second equality comes from basis switching according to \eqref{eq:0andSi}.
We can see that, in the basis  $\{|\widehat 0\>,|G_i\>\}$, this operator performs a rotation by angle $2\theta$. Thus, its $t$-th power is
\begin{align*}
(\Xi_i^\pi D^\pi)^t |_{\cH_i}
= 
\cos 2t\theta \big( |\widehat 0\>\<\widehat 0| + |G_i\>\<G_i|\big)
+ \sin 2t\theta \big( |\widehat 0\>\<G_i| - |G_i\>\<\widehat 0|\big),
\end{align*}
and therefore
\begin{align*}
|\zeta_t\> := 
(\Xi_i^\pi D^\pi)^t |\widehat 0\>
= 
\cos 2t\theta  |\widehat 0\> - \sin 2t\theta  |G_i\>.
\end{align*}

In Claim~\ref{clm:exact_Grover}, we effectively choose $t$, the number of Grover's iterations, to be $\tau=\lfloor \pi/(4\theta)\rfloor$, the largest integer value such that $2\tau\theta\le \pi/2$. This value is chosen so that Grover's iterations rotate the state from $|\widehat 0\>$ towards $|G_i\>$ as close as possible without overshooting it.
Then $\beta=\pi/2-2\tau\theta\in [0,2\theta)$, as given in \eqref{eq:Grover_angles}, is the remaining angle by which we still need to rotate the state $|\zeta_\tau\>$ towards $|G_i\>$.
 Indeed, we have
\beE
 \label{eq:penultimateGrover}
|\zeta_\tau\>
=\cos(2 \tau \theta)|\widehat 0\> - \sin(2\tau\theta)|G_i\>
=\sin\beta|\widehat 0\>-\cos\beta|G_i\>.
\enE

To perform this remaining ``real rotation'' by angle $\beta$, we once again employ a pair of ``complex rotations'' $D^\phi$ and $\Xi_i^{\pi+2\psi}$, except that now the rotation phases $\phi$ and $\pi+2\psi$ are not $\pi$.
See Figure~\ref{fig:Bloch_sphere} for an illustration.

\def\rB{3} 
\def\rP{0.02} 

\tikzset{
  Bloch2D/.pic={
  \draw (0,0) circle (1);
  \coordinate (P1) at (0, -1);
  \coordinate (P2) at (-0.96,-0.28);
  \coordinate (P3) at (0.96,-0.28);
  \coordinate (P4) at (-0.5376,0.8432);
  \coordinate (P5) at (0.209066666,0.8432);
  \coordinate (P6) at (0,1);
  \coordinate (A1) at (-0.6, -0.8);
  \coordinate (A2) at (0.6, 0.8);
 \draw[red,dashed,thick] (P1) -- (P2); 
 \draw[blue,dashed,thick] (P2) -- (P3);
 \draw[red,dashed,thick] (P3) -- (P4); 
 \draw[blue,dashed,thick] (P4) -- (P5);
 \draw[red,dashed,thick] (P5) -- (P6); 
 \draw[red,opacity = 0.15] (P1) -- (P2); 
 \draw[blue,opacity = 0.15] (P2) -- (P3);
 \draw[red,opacity = 0.15] (P3) -- (P4); 
 \draw[blue,opacity = 0.15] (P4) -- (0.5376,0.8432);
 \draw[red,opacity = 0.15] (0.96,0.28) -- (P6); 
\draw[blue, thick] (P1) -- (P6);
\draw[red, thick] (A1) -- (A2);
\filldraw[blue] (P1) circle (\rP);
\filldraw[blue] (P6) circle (\rP);
\filldraw[red]  (A1) circle (\rP);
\filldraw[red]  (A2) circle (\rP);
}}

\tikzset{
  BlueEllipse/.pic={
  \draw[blue, opacity = 0.15] (0,0) ellipse (1 and 0.479426);
}}
\tikzset{
  RedEllipse/.pic={
  \draw[red, opacity = 0.15, rotate = -40.5179] (0,0) ellipse (1 and 0.38354);
}}
\tikzset{
  BlueArcFull/.pic={
  \draw[blue,dashed,thick] (0,0) arc[start angle=-180, end angle=0, x radius=1, y radius=0.479426];   
}}
\tikzset{
  BlueArcPart/.pic={
  \draw[blue,dashed,thick] (0,0) arc[start angle=-180, end angle=-67.1146, x radius=1, y radius=0.479426];   
}}
\tikzset{
  RedArcFull/.pic={
  \draw[red,thick,dashed, rotate = -40.5179] (0,0) arc[start angle=-161.852, end angle=18.1482, x radius=1, y radius=0.38354];   
}}
\tikzset{ 
  RedArcPart/.pic={
  \draw[red,thick,dashed, rotate = -40.5179] (0,0) arc[start angle=-161.852, end angle=-106.216, x radius=1, y radius=0.38354];   
}}
\tikzset{
  Bloch3D/.pic={
  \draw (0,0) circle (1);
  \coordinate (P1) at (0., -0.877583);
  \coordinate (P2) at (-0.96, -0.245723);
  \coordinate (P3) at (0.96, -0.245723);
  \coordinate (P4) at (-0.5376, 0.739978);
  \coordinate (P5) at (0.209067, 0.502526);
  \coordinate (P6) at (0., 0.877583);
  \coordinate (A1) at (-0.6, -0.702066);
  \coordinate (A2) at (0.6, 0.702066);
  \coordinate (C1) at (-0.48, -0.561653);
  \coordinate (C2) at (0., -0.245723);
  \coordinate (C3) at (0.2112, 0.247127);
  \coordinate (C4) at (0., 0.739978);
  \coordinate (C5) at (0.48, 0.561653);

\draw[blue, thick] (P1) -- (P6);
\draw[red, thick] (A1) -- (A2);

\pic[scale=0.6*\rB] at (C1) {RedEllipse};
\pic[scale=0.96*\rB] at (C2) {BlueEllipse};
\pic[scale=0.936*\rB] at (C3) {RedEllipse};
\pic[scale=0.5376*\rB] at (C4) {BlueEllipse};
\pic[scale=0.6*\rB] at (C5) {RedEllipse};

\pic[scale=0.6*\rB] at (P2) {RedArcFull};
\pic[scale=0.96*\rB] at (P2) {BlueArcFull};
\pic[scale=0.936*\rB] at (P4) {RedArcFull};
\pic[scale=0.5376*\rB] at (P4) {BlueArcPart};
\pic[scale=0.6*\rB] at (P6) {RedArcPart};

\filldraw[blue] (P1) circle (\rP);
\filldraw[blue] (P6) circle (\rP);
\filldraw[red]  (A1) circle (\rP);
\filldraw[red]  (A2) circle (\rP);
}}

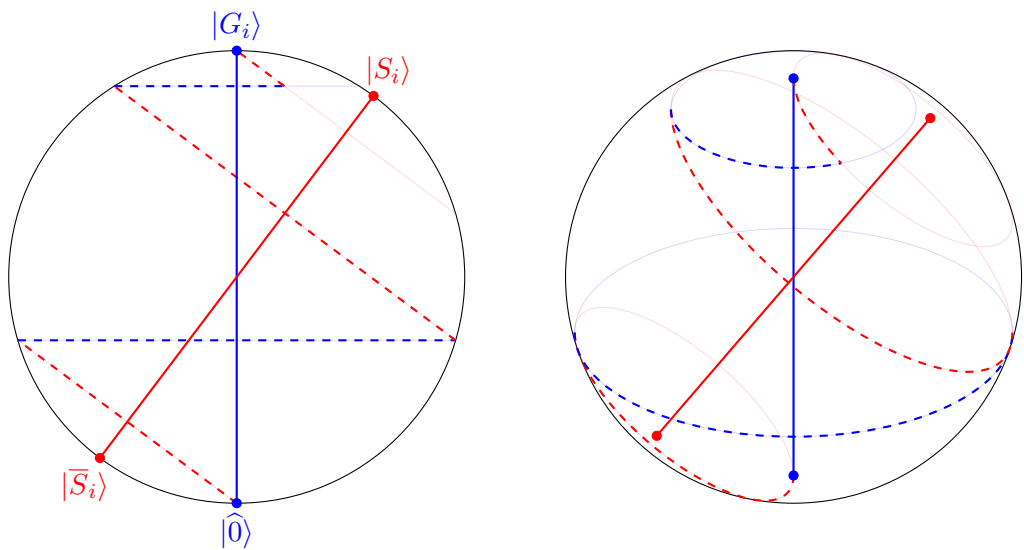
\begin{figure}[!h]
\centering
\begin{tikzpicture}
\pic[scale=\rB] at (0,0) {Bloch2D};
\node[anchor = south,blue] at (0,\rB) {$|G_i\>$};
\node[anchor = north,blue] at (0,-\rB) {$|\widehat 0\>$};
\node[anchor = south, red] at (0.6*\rB+0.2,0.8*\rB) {$|\Yes_i\>$};
\node[anchor = north, red] at (-0.6*\rB-0.2,-0.8*\rB) {$|\No_i\>$};
\filldraw[white, opacity=0.2] (0,3.7) circle (2pt);
\filldraw[white, opacity=0.2] (0,-3.7) circle (2pt);
\end{tikzpicture}
\hspace{1cm}
\begin{tikzpicture}
\pic[scale=\rB] at (0,0) {Bloch3D};
\filldraw[white, opacity=0.2] (0,3.7) circle (2pt);
\filldraw[white, opacity=0.2] (0,-3.7) circle (2pt);
\end{tikzpicture}
\caption{
\footnotesize
The Bloch sphere corresponding to basis
 $\{|\widehat 0\>,|G_i\>\}$ from two angles.
The blue axis is between $|\widehat 0\>$ and $|G_i\>$, and the red axis is between $|\Yes_i\> = \sin\theta|\widehat 0\>+\cos\theta|G_i\>$ and $|\No_i\> = \cos\theta|\widehat 0\>-\sin\theta|G_i\>$. 
The transformation of the initial state  $|\widehat 0\>$ and into the final state $|G_i\>$ follows the dashed lines.
}
\label{fig:Bloch_sphere}
\end{figure}

First, note that
\[
D^\phi |_{\cH_i} = |\widehat 0\>\<\widehat 0|+\ee^{\ii\phi} |G_i\>\<G_i|,
\]
and therefore we have
\begin{align*}
D^\phi |\zeta_t\>
&=
\sin\beta|\widehat 0\> - \ee^{\ii\phi} \cos\beta |G_i\>
\\ &=
\big(\sin\beta \sin\theta - \ee^{\ii\phi}\cos\beta \cos\theta\big) |\Yes_i\>
+ \big(\sin\beta \cos\theta + \ee^{\ii\phi}\cos\beta \sin\theta\big)|\No_i\>
\end{align*}
The phase $\phi$ is chosen so that \emph{absolute value} of the fraction of the amplitudes of $|\Yes_i\>$ and $|\No_i\>$ is $\cot\theta$. As we will see, this is achieved by choosing
$
\phi = \arccos\left(-\frac{\tan\beta}{\tan2\theta}\right)\in [\pi/2, \pi), 
$
for which
\[
\ee^{\ii\phi} = \frac{-\cos2\theta \sin\beta+\ii\sqrt{\sin^22\theta -\sin^2\beta}}{\sin2\theta \cos\beta}.
\]
Inserting this expression for $\ee^{\ii\phi}$ in the above expression for $D^\phi |\zeta_t\>$, after some straight forward yet tedious derivations, we get
\begin{align*}
D^\phi |\zeta_t\>
&=
\frac{\sin\beta-\ii\sqrt{\sin^22\theta -\sin^2\beta}}{\sin2\theta}\cdot\cos\theta |\Yes_i\>
+\frac{\sin\beta+\ii\sqrt{\sin^22\theta -\sin^2\beta}}{\sin2\theta}\cdot\sin\theta |\No_i\>
\\ & = 
\ee^{-\ii\psi}\cos\theta |\Yes_i\>
+\ee^{\ii\psi}\sin\theta |\No_i\>,
\end{align*}
where 
$\psi = \arccos \frac{\sin\beta}{\sin2\theta} \in (0,\pi/2]$.
Note that these choices of $\phi$ and $\psi$ are the same as in \eqref{eq:Grover_angles}.

Now recall $\Xi_i^{\psi}$ from \eqref{eq:diffusion_phi}. The restriction of $\Xi_i^{\pi+2\psi}$ to the subspace $\cH_i$ is 
\[
\Xi_i^{\pi+2\psi} |_{\cH_i} = 
|\No_i\>\<\No_i| - \ee^{\ii2\psi} |\Yes_i\>\<\Yes_i|.
\]
Therefore, we can see that
\[
\Xi_i^{\pi+2\psi}D^\phi |\zeta_t\> =
-\ee^{\ii\psi}\cos\theta |\Yes_i\>
+\ee^{\ii\psi}\sin\theta |\No_i\>
= \ee^{\ii(\pi+\psi)} |G_i\>,
\]
which concludes the proof.

 \subsection{Balanced case}
 \label{sec:Grover_Single}
 
Let us here consider $r=(p-1)/2$, which is the parameter regime relevant for the OPI problem. We have $\rho = 1/2 - 1/(2p)$, and, since $\rho > 1/2 -1/(2\sqrt 2) = \sin^2(\pi/8)$, we have  $\theta_\rho=\arcsin\sqrt\rho\in(\pi/8,\pi/4)$. Hence, in turn, $\tau_\rho=\lfloor \pi/(4\theta_\rho)\rfloor=1$ and $\beta_\rho=\pi/2-2\theta_\rho$.

We can see that
\[
\sin\theta_\rho=\sqrt{1/2 - 1/(2p)},
\qquad
\cos\theta_\rho=\sqrt{1/2 + 1/(2p)},
\]
and, because $\beta_\rho=\pi/2-2\theta_\rho$, we also have
\[
\sin\beta_\rho=\cos 2\theta_\rho = 1/p,
\qquad
\cos\beta_\rho=\sin 2\theta_\rho = \sqrt{p^2-1}/p.
\]
From these, in turn, we get
\[
\phi_\rho = \arccos\frac{-1}{p^2-1},
\qquad
\psi_\rho = \arccos\frac{1}{\sqrt{p^2-1}}.
\]
These are the values of $\phi_\rho$ and $\psi_\rho$ given for the OPI problem in Section~\ref{sec:Exact_Grover_main} to intuitively justify the approximation $|\widetilde G_i\>=-\Xi_i^\pi |\widehat 0\>$ for $|G_i\>$. From now on, however, we can ignore them.

Recall that, in the parameter regime relevant for OPI, we have $m=p-1$ and $\ell=\lfloor p/20\rfloor$. 
Therefore, $q = 1/20 + \OO(p^{-1/3})$.
We want to show that the states
\[
|\algostate_3'\> = \bigotimes_{i=1}^{p-1}\left(\sqrt{1-\nmean}|0\>|\widehat 0\>+\sqrt{q}|1\>|G_i\>\right),
\quad
|\widetilde\algostate_3'\> = \bigotimes_{i=1}^{p-1}\left(\sqrt{1-\nmean}|0\>|\widehat 0\>+\sqrt{q}|1\>|\widetilde{G}_i\>\right)
\]
are almost equal, in particular, that $\||\algostate_3'\>-|\widetilde\algostate_3'\>\|=\OO(1/\sqrt{p})$.

First, let us examine the inner product of $|G_i\>$ and $|\widetilde{G}_i\>$. From \eqref{eq:0andSi} and the expression for $\Xi_i^\pi|_{\cH_i}$, we have $ |G_i\> = \cos\theta_\rho|\Yes_i\>-\sin\theta_\rho|\No_i\>$ and
\[
|\widetilde G_i\> =
-\Xi_i^\pi |\widehat 0\>
 =  - \sin\theta_\rho \Xi_i^\pi |\Yes_i\> -\cos\theta_\rho \Xi_i^\pi |\No_i\>
 = \sin\theta_\rho|\Yes_i\>-\cos\theta_\rho|\No_i\>.
\]
Hence
\[
\<G_i|\widetilde G_i\> = 2\sin\theta_\rho\cos\theta_\rho = 
\sin 2\theta_\rho = \sqrt{1-1/p^2}.
\]
Note that this value is real, and thus so is
\begin{multline*}
\<\algostate_3'|\widetilde\algostate_3'\>
= \prod_{i=1}^{p-1}
\left(1-\nmean+\nmean \<G_i|\widetilde G_i\> \right)
 \\ =
\left(1-\nmean+\nmean \sqrt{1 - 1/p^2}\right)^{p-1}
\ge
1-\nmean p\left(1-\sqrt{1 - 1/p^2}\right).
\end{multline*}
Therefore, we have
\[
\big\||\algostate_3'\>-|\widetilde\algostate_3'\>\big\|^2
= 2-2 \<\algostate_3'|\widetilde\algostate_3'\>
\le 2 qp\left(1-\sqrt{1 - 1/p^2}\right) .
\]
Recall that $q = 1/20 + \OO(p^{-1/3})$, and we can also see that
$1-\sqrt{1 - 1/p^2}=1/(2p^2)+o(1/p^2)$. Hence, we have
$\||\algostate_3'\>-|\widetilde\algostate_3'\>\|^2 \le 1/(20 p) + o(1/p)$,
or, more loosely, $\||\algostate_3'\>-|\widetilde\algostate_3'\>\| = \OO(1/\sqrt{p})$.

 \subsection{Efficient implementations of complex rotations}
 \label{app:efficient_reflections}

In this section we present efficient implementations of controlled-$D^\phi$ and controlled-$\Xi_i^\psi$.

\paragraph{Controlled-$D^\phi$.}

Recall that $D^\phi = |\widehat 0\>\<\widehat 0|+\ee^{\ii\phi} (I_p-|\widehat 0\>\<\widehat 0|)$, which acts on a $p$-dimensional qudit. When we talk about implementing $D^\phi$ (and, later, controlled-$D^\phi$), we have to discuss how exactly each qudit gets represented in a binary memory. The most natural approach, which is also what we consider here, is to use $\lceil \log_2 p\rceil$ qubits for that purpose.

For integers $z\ge 0$ and $\ell > \log_2 z$, let $z_\ell$ be the binary representation of $z$ using $\ell$ bits, and let $|z_\ell\>$ be the corresponding $\ell$-qubit state. Let $\ell^\star:=\lceil \log_2 p\rceil$, for short. When implementing $D^\phi$, we are particularly interested in the state $|\widehat 0\>$, which is represented in the memory as
$|\widehat 0\> = \frac{1}{\sqrt{p}} \sum_{z=0}^{p-1}|z_{\ell^\star}\!\>$.

Let $U$ be any unitary acting on $\ell^\star$ qubits that maps $|0_{\ell^\star}\!\>=|0\>^{\otimes \ell^\star}$ to $|\widehat 0\>$ and let
\[
\mathsf{Diag}^\phi:= |0_{\ell^\star}\>\<0_{\ell^\star}|+\ee^{\ii\phi} (I_p-|0_{\ell^\star}\>\<0_{\ell^\star}|).
\]
Then we can express $D^\phi$ as $D^\phi = U \mathsf{Diag}^\phi U^*$. It is not hard to see that $\mathsf{Diag}^\phi$ can be implemented using $\OO(\log p)$ ancilla qubits and $\OO(\log p)$ gates, and it is left to see how to implement $U$.

\bigskip

Consider the following recursive expression relating uniform superpositions.  Let $q\ge 1$ and $\ell\ge \log_2 q$ be positive integers. Note that $\ell$ bits suffice to express in binary any integer $z$ such that $0\le z\le q-1$. We have
\[
\frac{1}{\sqrt{q}}\sum_{z=0}^{q-1}|z_{\ell}\> =
\begin{cases}
|0\>
\frac{1}{\sqrt{q}}\sum_{z=0}^{q-1}|z_{\ell-1}\>
& \text{if }q\le 2^{\ell-1},
\\
\sqrt\frac{2^{\ell-1}}{q}|0\>|+\>^{\otimes (\ell-1)}
+ \sqrt{1-\frac{2^{\ell-1}}{q}}|1\>
\frac{1}{\sqrt{q'}}\sum_{z=0}^{q'-1}|z_{\ell-1}\>
& \text{if }q\ge 2^{\ell-1},
\end{cases}
\]
where $q':= q-2^{\ell - 1}\le 2^{\ell - 1}$, for short. Having this recursive expression at our disposal, it is then not hard to see how to map $|0_{\ell^\star}\!\>$ to $|\widehat 0\>=\frac{1}{\sqrt{p}} \sum_{z=0}^{p-1}|z_{\ell^\star}\!\>$ using $\OO(\log^2 p)$ gates and $\OO(\log p)$ ancilla bits. This implements $U$, and $U^*$ can be simply implemented by reversing the circuit for $U$.

To implement controlled-$D^\phi$, we it suffices to implement controlled-$U$ and controlled-$\mathsf{Diag}^\phi$. This can be simply achieved by controlling every single- and two-qubit gate in implementations of $U$ and $\mathsf{Diag}^\phi$.

\paragraph{Controlled-$\Xi_i^\psi$.}

Recall that, for the max-LINSAT problem, the input is a collection of sets $S_1,\ldots,S_m\subset\F_p$ and a collection of vectors $b_1,\ldots,b_m\in\F_p^n$. We assume that this input is encoded in binary as follows. First, for every $i\in\{1,\ldots,m\}$, the set $S_i$ is encoded as a $p$-bit string $(f_i(0),f_i(1),\ldots,f_i(p-1))$ where, as in \eqref{eq:g_i_def}, $f_i(z)=1$ if and only if $z\in S_i$. And then, for every $i\in\{1,\ldots,m\}$, the vector $b_i$ is encoded component-by-component as an $n\lceil\log_2p\rceil$-bit string. Hence, in total, the input is an $N:=mp+mn\lceil\log_2p\rceil$-bit string $(k_0,k_1, \ldots, k_{N-1})$.

As described in Section~\ref{sec:Comp_model}, we assume that we have a read-only \QRAMc access to the input via the gate
$R_{\mathrm C}$ acting on the random access and the input registers as
\[
R_{\mathrm C}\colon |a,b\>|k_0,k_1,\ldots,k_{N-1}\> \mapsto |a,b\oplus k_a\>|k_0,k_1,\ldots,k_{N-1}\>
\]
where $a \in \{0,1,\ldots,N-1\}$ is an address encoded in $\lceil\log_2 N\rceil$ bits and $b \in \{0,1\}$ is an interface bit.
For the purposes of implementing $\Xi_i^\psi$, we will only need to access input bits $(k_{(i-1)p},k_{(i-1)p+1},k_{ip-1})=(f_i(0),f_i(1),\ldots,f_i(p-1))$ of the input, which encode the set $S_i$.

\bigskip

Let us introduce three more registers: the \emph{control qubit}, which will control whether to apply $\Xi_i^\psi$, the \emph{set number} register, corresponding to $i\in\{1,\ldots,m\}$, and the \emph{set element} register, corresponding to $x\in\F_p$. We will implement controlled-$\Xi_i^\psi$ on the control and the set element registers given that the set number register is in state $|i\>$.

Suppose we have an ``address translator'' unitary $T$ that acts on the set number, the set element, and the address registers as
\[
T\colon |i,x,0\> \mapsto |i,x, (i-1)p+x \>.
\]
where $i\in\{1,\ldots,m\}$ and $x\in\{0,1,\ldots,p-1\}$.
The translator $T$ performs simple arithmetic, and can be implemented using $\OO(\polylog(p+m))$ gates.%
\footnote{We could also consider that the input is organized in a memory in such a way that an address translator is not necessary at all. This would be especially reasonable when reducing the OPI problem to max-LINSAT, because for the OPI problem the input consists of the sets $S_1,\ldots,S_m$ alone.}
We have
\[
T^* R_{\mathrm C} T \colon |i,x,0,b\>|k_0,k_1,\ldots,k_{N-1}\> \mapsto |i,x,0,b\oplus f_i(x)\>|k_0,k_1,\ldots,k_{N-1}\>
\]
where the four registers corresponding to the state $|i,x,0,b\>$ are, respectively,  the set number, the set element, the address, and the interface registers.

Let $M_i$ be the restriction of $T^* R_{\mathrm C} T$ to the subspace where the set number, the address, and the input registers are, respectively, in states $|i\>$, $|0\>$, and $|k_0,k_1,\ldots,k_{N-1}\>$.
The unitary $M_i$ is essentially what is known as the membership oracle in search algorithms, it acts on the set element and the interface registers, and it can be expressed as
\[
M_i
:= \sum_{z\in\F_p}{} |z\>\<z|\otimes \big(|f_{i}(z)\>\<0| + |\overline{f_{i}(z)}\>\<1|\big) 
= \sum_{z\in\F_p}{} |z\>\<z|\otimes \big(|+\>\<+| + (-1)^{f_{i}(z)}|-\>\<-|\big).
\]
Having $M_i$ at our disposal, the implementation of controlled-$\Xi_i^\psi$ and, its special case, controlled-$\Xi_i^\pi$ are given in Figure~\ref{fig:circuits_for_Xi}.

\def\ywA{0.13} 
\def\ywB{0.48} 
\def\xwA{2.6} 
\def\xwB{2.25} 

\tikzset{
  wires/.pic={
    \draw (-\xwA,\ywA+\ywB)--(\xwA,\ywA+\ywB);
    \foreach \i in {-1,...,1}      \draw (-\xwA,\ywA*\i)--(\xwA,\ywA*\i);
    \draw (-\xwB,-\ywA-\ywB)--(\xwB,-\ywA-\ywB);
}}

\def\ypA{0.2} 
\def\ypB{0.3} 
\def\xgw{0.35} 
\def\xgc{1.1} 

\def\rCXctrl{0.05} 

\tikzset{
  ctrlGate/.pic={
\draw [fill = black] (0,\ywA+\ywB) circle (\rCXctrl);
\draw (0,\ywA+\ywB)--(0,-\ywA-\ywB+\ypB);
\draw [fill=white] (-\xgw,-\ywA-\ywB-\ypB) rectangle (\xgw,-\ywA-\ywB+\ypB);
}}

\tikzset{
  MiGate/.pic={
    \draw [fill=white] (-\xgw,-\ywA-\ywB-\ypB) rectangle (\xgw,\ywA+\ypA);
    \node at (0,-0.5*\ywB-0.5*\ypB+0.5*\ypA) {$M_i$};
}}

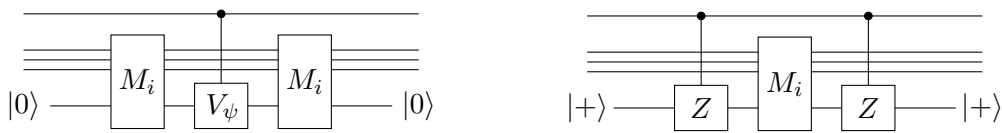
\begin{figure}[!h]
\centering
\centering
\begin{tikzpicture}
\pic at (0,0) {wires};
\pic at (-\xgc,0) {MiGate};
\pic at (\xgc,0) {MiGate};
\pic at (0,0) {ctrlGate};
\node at (0.02,-\ywA-\ywB-0.03) {$V_\psi$};
\node at (-\xwA,-\ywA-\ywB) {$|0\>$};
\node at (\xwA,-\ywA-\ywB) {$|0\>$};
\end{tikzpicture}
$\hspace{1.2cm}$
\begin{tikzpicture}
\pic at (0,0) {wires};
\pic at (0,0) {MiGate};
\pic at (-\xgc,0) {ctrlGate};
\node at (-\xgc,-\ywA-\ywB) {$Z$};
\pic at (\xgc,0) {ctrlGate};
\node at (\xgc,-\ywA-\ywB) {$Z$};
\node at (-\xwA,-\ywA-\ywB) {$|+\>$};
\node at (\xwA,-\ywA-\ywB) {$|+\>$};
\end{tikzpicture}
\caption{
\footnotesize
Implementation of controlled-$\Xi_i^\psi$ (left) and controlled-$\Xi_i^\pi$ (right), requiring, respectively, two and a single call to \QRAMc. 
Here $V_\psi:=|0\>\<0|+\ee^{\ii\psi}|1\>\<1|$ is a Pauli-$Z$ rotation and $Z=V_\pi$ is the Pauli-$Z$ gate.
The top wire represents the control qubit, the center wires represent the set element register, and the bottom wire is the interface register. Both circuits return the interface register to its initial state.
}
\label{fig:circuits_for_Xi}
\end{figure}

Note that, when implementing controlled-$\Xi_i^\psi$, we use \QRAMc to access only $p$ out of $N$ input bits. Thus, if we wanted to simulate this \QRAMc access to the input in the standard model, as described in Section~\ref{sec:Comp_model}, simulating each access would require only $\OO(p)$ gates, not $\OO(N)$.

\section{Fast implementation of the number-theoretic transform}
\label{app:FastNTT}

In Section~\ref{sec:NTT}, we introduced the number-theoretic transform $\NTT_{p,\gamma}$ for a primitive element $\gamma \in \F_p$. In order to see that this transform can be computed in time $\OO(p \polylog p)$, it is helpful to generalize the definition and consider $\NTT_{p,\beta}$ for an arbitrary nonzero element $\beta \in \F_p^*$.

Given $\beta \in \F_p^*$, let $o_p(\beta)$ denote the multiplicative order of $\beta$ in $\F_p^*$, that is, the smallest positive integer $n$ such that $\beta^n = 1$.  The number-theoretic transform $\NTT_{p,\beta}$ is the linear map on $\F_p^{o_p(\beta)}$ that maps  $x=(x_0,x_1,\ldots,x_{o_p(\beta)-1})$ to $X=(X_0,X_1,\ldots,X_{o_p(\beta)-1})$ where $X_j=\sum_{i=0}^{o_p(\beta)-1}\beta^{ij}x_i$ for all $j$. 
The inverse map $\NTT_{p,\beta}^{-1}$ is given by $x_i=(o_p(\beta))^{-1}\sum_{j=0}^{o_p(\beta)-1}\beta^{-ij}X_j$ for all $i$, where $o_p(\beta)^{-1}$ is the multiplicative inverse of $o_p(\beta)$ in $\F_p$.

\bigskip

In this section, we briefly sketch the main ideas behind fast algorithms that compute $\NTT_{p,\beta}$ using $\OO(o_p(\beta) \polylog p)$ logical gates. 
The choice of algorithm depends on $o_p(\beta)$: when $o_p(\beta)$ is a composite number, the Cooley--Tukey algorithm can be used; when $o_p(\beta)$ is prime, Rader's algorithm applies.

As a special case, when $\beta = \gamma$ is a primitive element of $\F_p$, we have $o_p(\gamma) = p - 1$, so the runtime becomes $\OO(p \polylog p)$, as claimed. The same complexity applies to computing the inverse transform $\NTT_{p,\beta}^{-1}$.

\paragraph{The Cooley--Tukey algorithm.}

The Cooley--Tukey algorithm provides a recursive method for computing $\NTT_{p,\beta}$ when the order $o_p(\beta)$ is a composite integer. This method was originally discovered by Gauss and later independently rediscovered by Cooley and Tukey~\cite{cooley65:dft,heideman84:gauss}.

Suppose $o_p(\beta) = n'n''$ for integers $n', n'' \ge 2$, and define
Define $\beta':=\beta^{n''}$ and $\beta'':=\beta^{n'}$.
Then $\beta'$ and $\beta''$ have orders $o_p(\beta') = n'$ and $o_p(\beta'') = n''$, respectively. 
The Cooley--Tukey algorithm recursively computes $\NTT_{p,\beta}$ by performing the following steps:
\begin{enumerate}
  \item Compute $n'$ independent instances of $\NTT_{p,\beta''}$.
  \item Multiply the resulting values by powers of $\beta$ (known as twiddle factors).
  \item Compute $n''$ independent instances of $\NTT_{p,\beta'}$.
\end{enumerate}
Assuming the existence of algorithms for computing $\NTT_{p,\beta'}$ and $\NTT_{p,\beta''}$ in time $\OO(o_p(\beta') \, \polylog p)$ and $\OO(o_p(\beta'') \, \polylog p)$, respectively, this recursive approach yields an overall runtime of $\OO(o_p(\beta) \, \polylog p)$ for computing $\NTT_{p,\beta}$.

\paragraph{Rader's algorithm.} 
 
When $q:=o_p(\beta)$ is a prime, Rader's algorithm reduces the computation of $\NTT_{p,\beta}$ to a computation of a (cyclic) convolution over $\F_p$ of two sequences of length linear in $q$. The reduction proceeds as follows.

The goal is to compute $X_j = \sum_{i\in\F_q} \beta^{ij}x_i$ for all $j\in\F_q$. We begin by separately computing $X_0 = \sum_{i\in\F_q} x_i$, which can be done in time $\OO(q\log p)$. The remaining task is to jointly compute
\begin{equation}
\label{eq:Xj_Rader1}
X_j = x_0 + \sum_{i\in\F_q^*}x_i \beta^{ij}\qquad\text{for all}\qquad j\in\F_q^*.
\end{equation}

To proceed, we choose a primitive element $\zeta$ of $\F_q$, so that $\{\zeta^0,\zeta^1,\ldots,\zeta^{q-2}\}=\F_q^*$.
We can reindex  (\ref{eq:Xj_Rader1}) by letting $j = \zeta^{-k}$ and $i = \zeta^{k - \ell}$, therefore computing (\ref{eq:Xj_Rader1}) is equivalent to computing
\[
X_{\zeta^{-k}}
= x_0+\sum_{\ell=0}^{q-2} x_{\zeta^{k-\ell}} \beta^{\zeta^{-\ell}}
\qquad\text{for all}\qquad k\in\{0,1,\ldots,q-2\}.
\]
The above sum over $\ell$ is a convolution of two sequences of length $q-1$.

To simplify the computation of this convolution, particularly in cases where $q - 1$ has large prime divisors, we instead embed it into a convolution of two longer sequences whose lengths are powers of $2$.
Let $n:=2^{\lceil \log_2(2q-3)\rceil}$ to be the smallest power of $2$ that is at least $2q-3$, and
define two sequences $a = (a_0,a_1, \ldots, a_{n-1})$ and $b = (b_0,b_1, \ldots, b_{n-1})$ over $\F_p$ as
\begin{equation}
\label{eq:ab_for_convolution}
\begin{split}
&a := (x_{\zeta^0},x_{\zeta^1},\ldots,x_{\zeta^{q-2}},0,\ldots,0,x_{\zeta^1},\ldots,x_{g^{\zeta-2}}),
\\&b :=(\beta^{\zeta^{-0}},\beta^{\zeta^{-1}},\ldots,\beta^{\zeta^{-(q-2)}},0,\ldots,0,0,\ldots,0),
\end{split}
\end{equation}
where the zero-padding ensures that both sequences have length $n$.
We can observe that
\[
 X_{\zeta^{-k}} 
= x_0+ \sum_{\ell=0}^{n-1} a_{k-\ell} b_\ell 
\qquad\text{for all}\quad
k=0,1,\ldots,q-2
\]
where $k-\ell$ is taken modulo $n$.
Thereby, if the convolution $a * b$ over $\F_p$ can be evaluated in time $\OO(q\polylog p)$, then so can $\NTT_{p,\beta}$.

\bigskip

For the sequences $a, b \in \F_p^n$ from (\ref{eq:ab_for_convolution}), let $a_\C, b_\C \in \C^n$ denote the corresponding sequences where each entry is interpreted as an integer in $\{0, 1, \ldots, p-1\} \subset \C$, rather than as an element of the finite field $\F_p$. Let $a_\C * b_\C$ denote their (cyclic) convolution over $\C$, each entry of which is an integer in the range $\{0, 1, \ldots, n(p - 1)^2\}$. Note that
$a * b = a_\C * b_\C \mod p$,
where the modulus is applied entrywise and interpreted as mapping integers to elements of $\F_p$.
Therefore, to compute $a * b$, it suffices to approximate $a_\C * b_\C$ to within an additive error of at most $1/4$ in each entry and then round each value to the nearest integer. 
The convolution $a_\C * b_\C$ can be computed efficiently using the discrete Fourier transform over $\C$ in conjunction with the convolution theorem~(\ref{eq:convolution_theorem}).

\paragraph{Discrete Fourier Transform (DFT) over $\C$.}

Recall that, for a positive integer $n$, $\omega_n=\ee^{2\pi\mathsf{i}/n}$ denotes the $n$-th root of unity. Given a sequence $x=(x_0,x_1,\ldots,x_{n-1})\in\C^n$, the discrete Fourier transform of $x$ is the sequence 
$\DFT_n(x):=X=(X_0,X_1,\ldots,X_{n-1})\in\C^n$ such that $X_j=\sum_{i=0}^{n-1}\omega_n^{-ij}x_i$ for all $j$. The inverse tranform is given by $x_i=\frac{1}{n}\sum_{j=0}^{n-1}\omega_n^{ij}X_j$ for all $i$.

The convolution theorem states that, for any two sequences  $x,y\in\C^n$,
\begin{equation}
\label{eq:convolution_theorem}
x*y=\DFT_n^{-1}(\DFT_n(x)\circ\DFT_n(y)),
\end{equation}
where $\circ$ denotes the entrywise product of two sequences. 
For fast computation of $\DFT_n$ and its inverse, one can again use the Cooley--Tukey algorithm, now applied over $\C$. 
This algorithm is particularly well-suited to the case when when $n$ is a power of 2, and  both $\DFT_n$ and $\DFT_n^{-1}$ can be evaluated using $\OO(n \log n)$ arithmetic operations.

In practice, one can approximately carry out these arithmetic operations using fixed-point arithmetic. Consequently, the number of logical gates required to implement them depends on the bit-length of the input values and the desired level of precision.

It can be shown that, to achieve sufficient accuracy when evaluating $a_\C * b_\C$ via the convolution theorem, it suffices to represent all intermediate values using $\OO(\log p)$ bits. Consequently, each approximate arithmetic operation involved in computing $\DFT_n$ and $\DFT_n^{-1}$ can be implemented using $\OO(\polylog p)$ logical gates. This yields an overall gate complexity of $\OO(q \, \polylog p)$ for computing $a_\C * b_\C$, and thus also for computing $a * b$ over $\F_p$.

\section{Fast decoding of Reed--Solomon codes}
\label{app:FastDecoding}

The goal of this appendix is to guide a reader through the main steps and ideas of the nearly linear-time decoding algorithm of the narrow sense Reed--Solomon codes. The techniques described here are either directly taken of adapted from~\cite{aho74:algoDesign,reed78:GCDdecoding}.

\bigskip

Let us recall the goal of a decoder. For an unknown error $y=(y_0,y_1,\ldots,y_{p-2})$ of Hamming weight at most $t$, we are given the syndrome $s=(s_0,s_1,\ldots,s_{2t-1})$, where
\beE
\label{eq:synd_def}
s_{j}:=\sum_{i=0}^{p-2}\gamma^{ij}y_{i},
\qquad \forall j\in\Z.
\enE
While $s_j$ for $j=0,1,\ldots,2t-1$ are of particular interest to us, we define $s_j$ for all integers $j$. Note that, because $\gamma^{p-1}=1$, we have $s_{j+p-1}=s_j$ for all $j\in\Z$. The goal is to find $y$ in time $\OO(p\polylog p)$.

Here we show how to achieve an equivalent goal: to find all $s_0,s_1,\ldots,s_{p-2}$. These two goals are equivalent, because in time $\OO(p\polylog p)$ one can apply the inverse number-theoretic transform $\NTT_{p,\gamma}^{-1}$ to $(s_0,s_1,\ldots,s_{p-2})$, which yields $y$.

\subsection{Field of formal series}

Let \( \F_p(\!(x)\!) \) denote the field of \emph{formal Laurent series}, consisting of expressions of the form
$\fa = \sum_{i=0}^{\infty} a_{d-i} x^{d-i}$,
where each coefficient $a_{d-i} \in \F_p$ and where the integer $d = \deg \fa $ is called the \emph{degree} of $\fa$ and satisfies $a_d \ne 0$.%
\footnote{Strictly speaking, the standard notation for this field is \( \F_p(\!(x^{-1})\!) \), but we adopt \( \F_p(\!(x)\!) \) for convenience.}
For brevity, we refer to elements of \( \F_p(\!(x)\!) \) simply as \emph{formal series}.
Let $\lfloor\fa\rceil$ denote $\sum_{i=0}^{d} a_{d-i}x^{d-i}$ if $d\ge 0$ and let $\lfloor\fa\rceil=0$ if $d<0$. Essentially, $\lfloor\fa\rceil$ is the polynomial part of $\fa$.
Similarly, for an integer $k$, define $\polyr\fa_{k}:=x^k\polyr{x^{-k}\fa}$. This, informally speaking, is the rounding of $\fa$ to its most significant terms, down to until the term of degree $x^k$. For any $\fa$, the statements $\polyr{\fa}_k=0$ and $\deg\fa<k$ are equivalent.
 We denote the zero-element of the field simply by $0$, whose degree we define as $\deg0:=-\infty$.

Given a non-zero element $\fa \in \F_p(\!(x)\!)$ of degree $d$, we can obtain the coefficients of its multiplicative inverse  $\fb := \fa^{-1}$ recursively as follows.
Since $\fb$ must satisfy $\fa\fb = 1$, it has the form $\fb=\sum_{j=0}^\infty b_{-d-j}x^{-d-j}$, where the coefficients $b_{-d-j}$ are given recursively as
\[
b_{-d}=a_d^{-1}
\qqAnd
b_{-d-j}  = -a_d^{-1}\sum_{i=1}^{j}a_{d-i}b_{-d-j+i} 
 \quad \text{for }j\ge 1.
\]
Observe that, for any $k \ge 0$, the $k+1$ most significant coefficients of $\fa^{-1}$ depend only on the $k+1$ most significant coefficients of $\fa$.
More formally, we have
\begin{align}
\label{eq:inverse_leading_terms}
\polyr{\fa^{-1}}_{-\deg\fa-k}=\polyr{\polyr{\fa}^{-1}_{\deg\fa-k}}_{-\deg\fa-k}.
\end{align}
Similarly, the $k+1$ most significant coefficients of the product $\fa \fa'$ of two non-zero formal series $\fa, \fa' \in \F_p(\!(x)\!)$ depend only on the $k+1$ most significant coefficients of each operand. Combined with~\eqref{eq:inverse_leading_terms}, this implies that the same holds for their division $\fa / \fa' := \fa (\fa')^{-1}$.

\paragraph{Fast multiplication and division of polynomials.}

Polynomials over $\F_p$ can be multiplied and divided efficiently, as the following theorem demonstrates. However, note that the division of two polynomials $\fa$ and $\fa'$ is a formal series in $\F_p(\!(x)\!)$ and may, in general, have infinitely many non-zero coefficients. For this reason, we compute a ``rounded division''.

\begin{thm}
\label{thm:fast_mult_and_div}
Let $\fa = \sum_{i=0}^{d} a_i x^i$ and $\fa' = \sum_{i=0}^{d'} a'_i x^i$ be polynomials of degrees $d$ and $d'$, respectively, represented by their coefficient lists $a_0, a_1, \ldots, a_d$ and $a'_0, a'_1, \ldots, a'_{d'}$. Suppose we are also given a non-negative integer $k$, and assume that $d$, $d'$, and $k$ are all in $\OO(p)$.
Then the coefficients $b_0, b_1, \ldots, b_{d+d'}$ of the product
$\fb:=\fa\fa'=\sum_{i=0}^{d+d'}b_ix^i$
and the coefficients $c_{-k}, c_{-k+1}, \ldots, c_{d-d'}$ of the rounded division
$\fc := \polyr{\fa/\fa'}_{-k} = \sum_{i=-k}^{d-d'} c_i x^i$
can both be computed in time $\OO(p \polylog p)$.
\end{thm}

To see that the multiplication can indeed be done efficiently, as claimed by Theorem~\ref{thm:fast_mult_and_div}, consider first the special case when $d+d'\le p-2$. Extend the coefficient lists by setting $a_i = 0$ for $d < i \le p - 2$ and $a'_i = 0$ for $d' < i \le p - 2$, so that $a = (a_0, a_1, \ldots, a_{p-2})$ and $a' = (a'_0, a'_1, \ldots, a'_{p-2})$ are both vectors of length $p-1$. 
 We can then observe that the product $\fb = \fa \fa'$ has coefficient vector $b = (b_0, b_1, \ldots, b_{p-2})$, where each entry is given by the cyclic convolution $b_j = \sum_{i=0}^{p-2} a_{j - i} a'_i$ for all $j = 0, 1, \ldots, p - 2$ where $j-i$ is taken modulo $p=1$. The convolution theorem, analogously to (\ref{eq:convolution_theorem}), tells that
\[
b=a*a'=\NTT_{p,\gamma}^{-1}(\NTT_{p,\gamma}(a)\circ\NTT_{p,\gamma}(a')),
\]
which we can compute using the fast number-theoretic transform described in Section~\ref{sec:NTT} and Appendix~\ref{app:FastNTT}.

When $d + d' > p - 2$ but still $d, d' \in \OO(p)$, we can reduce the multiplication to a constant number of smaller multiplications by partitioning the input polynomials into blocks of degree at most $(p-3)/2$. Specifically, write
\[
\fa = \sum_{\ell=0}^{\lfloor 2d / (p-1) \rfloor} \fa^{(\ell)} x^{\ell(p-1)/2}
\qquad \text{and} \qquad
\fa' = \sum_{\ell'=0}^{\lfloor 2d' / (p-1) \rfloor} \fa^{\prime(\ell')} x^{\ell'(p-1)/2},
\]
where each $\fa^{(\ell)}$ and $\fa^{\prime(\ell')}$ is a polynomial of degree at most $(p - 3)/2$. Then, the product $\fb = \fa \fa'$ can be computed as
\[
\fb = \sum_{\ell, \ell'} \fa^{(\ell)} \fa^{\prime(\ell')}  x^{(\ell + \ell')(p-1)/2},
\]
which involves a constant number of multiplications of polynomials of degree at most $(p - 3)/2$. Each of these smaller products falls into the efficient case handled above. Therefore, the full product $\fa \fa'$ can still be computed in time $\OO(p \polylog p)$.

\bigskip

Regarding the efficient division, first observe that the degree of $\fa/\fa'$ is $d - d'$. So, as per the discussion after (\ref{eq:inverse_leading_terms}), to compute $\polyr{\fa/\fa'}_{-k}$, we only care about the $d-d'+1+k$ most significant coefficients of $\fa'$. More precisely, we have
$\polyr{\fa/\fa'}_{-k} = \polyr{\fa \polyr{(\fa')^{-1}}_{2d'-d-k}}_{-k}$.
Hence, it suffices to compute the rounded reciprocal $\polyr{(\fa')^{-1}}_{2d'-d-k}$.
An efficient method for doing that, which uses fast polynomial multiplication as a subroutine, is described in \cite[Algorithm~8.3]{aho74:algoDesign}.

\subsection{Fast extended Euclidean algorithm}
\label{sec:fast_Euclidean}

Building on fast $\OO(p \polylog p)$-time algorithms for polynomial multiplication and division, we can obtain a fast implementation of the extended Euclidean algorithm, also running in $\OO(p \polylog p)$ time. We implicitly specify the precise goal of the extended Euclidean algorithm below in Theorem~\ref{thm:fast_Euclidean}.

Let $\fR_0 \in \F_p[x]$ be a polynomial of degree $2t \le p - 1$ and $\fR_1 \in \F_p[x]$ a polynomial of degree strictly less than $2t$. Consider a finite sequence of polynomials 
\beE
\label{eq:R_q_sequence}
\fR_0,\fR_1,\fq_1,\fR_2,\fq_2,\fR_3,
\ldots,\fq_{j-2},\fR_{j-1},\fq_{j-1},\fR_{j}
\enE
recursively defined as
\begin{alignat*}{2}
& \fq_i := \polyr{\fR_{i-1}/\fR_i},
&\qquad& \fR_{i+1} := \fR_{i-1}-\fq_i\fR_i
\end{alignat*}
where $j$ is the unique index such that
\[
\deg\fR_{j-1}\ge t > \deg\fR_{j}.
\]
We refer to the polynomials $\fR_i$ as \emph{remainders} and the $\fq_i$ as \emph{quotients}. The sequence terminates when the degree of the remainder becomes less than $t$.
Given the quotients $\fq_1,\fq_2,
\ldots,\fq_{j-2},\fq_{j-1}$, recursively define polynomials 
\beE
\label{eq:polyPL_defs}
\begin{alignedat}{3}
    & \fP_{0} := 1, &\qquad& \fP_{1} := 0, &\qquad& 
    \fP_{i+1} := \fP_{i-1}-\fq_{i}\fP_{i}, \\
    & \fL_{0} := 0, && \fL_{1} := 1, &&
        \fL_{i+1} := \fL_{i-1}-\fq_{i}\fL_{i},
\end{alignedat}
\enE
for $i=1,2,\ldots,j-1$.

\begin{thm}
\label{thm:fast_Euclidean}
Given $\fR_0$ and $\fR_1$, one can compute $\fP_j$ and $\fL_j$ in time $\OO(p\polylog p)$.
\end{thm}

\begin{proof}
By induction on $i$, one can verify that
\beE
\label{eq:Ri_via_PiLi}
  \fP_{i}\fR_{0} + \fL_{i}\fR_{1} = \fR_{i}.
\enE
According to \cite[Theorems~8.17 and 8.18]{aho74:algoDesign}, in time $\OO(p\polylog p)$ we can compute
$\fP_{j^*-1}$,  $\fL_{j^*-1}$, $\fP_{j^*}$, $\fL_{j^*}$
where $j^*$ is the unique index such that 
\[
\deg\fR_{j^*-1}>t\ge \deg\fR_{j^*}.
\]
Using (\ref{eq:Ri_via_PiLi}), we can then obtain $\fR_{j^*-1}$ and $\fR_{j^*}$ via the fast polynomial multiplication.

If $\deg\fR_{j^*}<t$, then $j=j^*$ and we are done.
Otherwise, if $\deg\fR_{j^*}=t$, then $j=j^*+1$, and we require one additional iteration of the extended Euclidean algorithm. Using the fast polynomial division, we compute $\fq_{j^*} = \polyr{\fR_{j^*-1}/\fR_{j^*}}$, and then, in turn, we compute the sought after $\fP_{j^*+1} = \fP_{j^*-1}-\fq_{j^*}\fP_{j^*}$ and $\fL_{j^*+1} = \fL_{j^*-1}-\fq_{j^*}\fL_{j^*}$.
\end{proof}

Let us bound the degrees of $\fP_j$ and $\fL_j$, starting with the latter. Observe that for all $i\ge 1$, we have $\deg\fq_i = \deg\fR_{i-1} - \deg\fR_i$ and, thus, $\deg\fq_i \ge 1$.

From (\ref{eq:polyPL_defs}), we see that $\fL_{2} = -\fq_{1}$. Moreover, for $i\ge 2$, because $\deg\fL_{i-1}\ne\deg(\fq_{i}\fL_{i})$, we also have
\[
\deg\fL_{i+1} = \max\left\{\deg\fL_{i-1},\deg(\fq_{i}\fL_{i})\right\}
 = \deg\fq_{i}+\deg\fL_{i}.
\]
By induction, we can thus see that $\deg\fL_i=\sum_{k=1}^{i-1}\deg\fq_{k}$ for all $i\ge 2$.
Now, since $\deg\fR_0=2t$ and $\deg\fR_{j-1}\ge t$, we get
\[
\deg\fL_j=\sum_{i=1}^{j-1}\deg\fq_{i}
= \sum_{i=1}^{j-1}\left(\deg\fR_{i-1} - \deg\fR_i\right)
= \deg\fR_0 - \deg\fR_{j-1} \le t.
\]

Similarly, by induction we get that for $i\ge 2$,  $\deg\fP_i=\sum_{k=2}^{i-1}\deg\fq_{k} =\deg\fL_i - \deg\fq_1$. Since $\deg\fq_1\ge 1$, this implies $\deg\fP_j\le t-1$.

\subsection{Syndrome series}
\label{sec:syndrome_series}

From this point onward, the description and analysis of the fast decoding algorithm are adapted from the proof of~\cite[Theorem~2]{reed78:GCDdecoding}.

In this section, we instantiate the polynomials $\fR_0$ and $\fR_1$ that serve as inputs to the extended Euclidean algorithm described in Section~\ref{sec:fast_Euclidean}. Specifically, we choose $\fR_0 := x^{2t}$ and $\fR_1 := \sum_{k=0}^{2t-1}s_k x^{2t-1-k}$, where the values $s_{k}=\sum_{i=0}^{p-2}\gamma^{ik}y_{i}$ were introduced in (\ref{eq:synd_def}) and they are know to us for  $k=0,1,\ldots,2t-1$. 

The \emph{syndrome series} is defined as $\fS:=\sum_{k=0}^\infty s_k x^{-k}$. 
Note that $\fR_1 = \polyr{x^{2t-1}\fS}$.
The following theorem lies at the heart of the fast decoding algorithm. The remainder of this appendix is devoted to guiding the reader through its proof.

\begin{thm}
\label{thm:S_via_Pj_over_Lj}
Let $\fR_0 = x^{2t}$ and $\fR_1 = \polyr{x^{2t-1}\fS}$, and let $\fP_j$ and $\fL_j$ be the polynomials produced by the extended Euclidean algorithm applied to $\fR_0$ and $\fR_1$, as described in Section~\ref{sec:fast_Euclidean}. Then $-x\fP_j/\fL_j=\fS$.
\end{thm}

Recall that Theorem~\ref{thm:fast_Euclidean} states that the polynomials $\fP_j$ and $\fL_j$ can be computed in time $\OO(p  \polylog p)$. Because $\deg \fP_j, \deg \fL_j \in \OO(p)$, so Theorem~\ref{thm:fast_mult_and_div} ensures that 
$\polyr{-x \fP_j / \fL_j}_{-p+2}$ can also be computed in time $\OO(p \polylog p)$.
By Theorem~\ref{thm:S_via_Pj_over_Lj}, this implies that we can compute
\[
 \polyr{\fS}_{-p+2} = \sum_{k=0}^{p-2} s_k x^{-k}.
\]
in time $\OO(p \, \polylog p)$. This solves our problem, because—as discussed at the beginning of this appendix—the coefficients \( s_0, s_1, \ldots, s_{p-2} \) are sufficient to recover the vector \( y \).

\bigskip

The following lemma, which is instrumental towards proving Theorem~\ref{thm:S_via_Pj_over_Lj}, establishes that $\fS$ is a rational function, namely, that $\fS$ can be expressed as the ratio of two polynomials.

\begin{lem}
\label{lem:S_via_the_locator_poly}
There are polynomials $\fP,\fL\in\F_p[x]$ such that $\fS=-x\fP/\fL$ and $\deg\fL\le t$.
\end{lem}

\begin{proof}
Recall that the Hamming weight of the error vector $y=(y_0,y_1,\ldots,y_{p-2})$ is at most $t$. 
Let $L := \{ \ell \in \{0, 1, \ldots, p - 2\} \colon y_\ell \ne 0 \}$ be the set of error locations.
Define polynomials
\[
\fL:=\prod_{\ell\in L}(x-\gamma^\ell)
\qquad\text{and}\qquad
\fP:=-\sum_{\ell\in L}y_\ell \prod_{\ell'\in L\setminus\{\ell\}}(x-\gamma^{\ell'}),
\]
the former of which is sometimes called the \emph{error-locator polynomial}.
Clearly, $\deg \fL = |L| \le t $, as required.

Note that 
$s_k = \sum_{\ell\in L}\gamma^{k\ell}y_\ell$
and that the inverse of formal series $1-\gamma^\ell x^{-1}\in \F_p(\!(x)\!)$ is  formal series $\sum_{k=0}^\infty \gamma^{k\ell} x^{-k}\in \F_p(\!(x)\!)$.
Hence, we have
\[
\fS=\sum_{k=0}^\infty s_k x^{-k}
= \sum_{\ell\in L}\sum_{k=0}^\infty \gamma^{k\ell}y_\ell x^{-k}
= \sum_{\ell\in L}\frac{y_\ell}{1-\gamma^\ell x^{-1}} = x\sum_{\ell\in L}\frac{y_\ell}{x-\gamma^\ell}
= -x\frac{\fP}{\fL}.
\qedhere
\]
\end{proof}

\subsection{Parallel remainder--quotient sequence for formal series}

Recall the sequence of remainders and quotients from~\eqref{eq:R_q_sequence} in Section~\ref{sec:fast_Euclidean}, which we consider when the initial polynomials of the sequence are $\fR_0=x^{2t}$ and $\fR_1=\polyr{x^{2t-1}\fS}$, as in Section~\ref{sec:syndrome_series}. We now define a closely related sequence of remainders and quotients, this time allowing the remainders to be formal series rather than just polynomials.

We begin by defining the initial terms of the new sequence:
\beE
\label{eq:fRprime01_defs}
\fR'_0 := x^{2t}, \qquad \fR'_1 := x^{2t-1}\fS.
\enE
Using these, define the following finite sequence of formal series
\beE
\label{eq:R_q_prime_sequence}
\fR_0',\fR_1',\fq'_1,\fR_2',\fq'_2,\fR_3',
\ldots,\fq'_{j'-2},\fR'_{j'-1},\fq'_{j'-1},\fR'_{j'}
\enE
where the quotients (which are still polynomials) and the remainders (which are now formal series) are recursively defined as
\beE
\label{eq:R_q_prime_recursion}
 \fq'_i := \polyr{\fR'_{i-1}/\fR'_i},
\qquad \fR'_{i+1} := \fR'_{i-1}-\fq'_i\fR'_i,
\enE
and where $j'$ is the unique index such that
\beE
\label{eq:def_of_jprime}
\deg\fR'_{j'-1}\ge t > \deg\fR'_{j'}.
\enE
The two sequences, (\ref{eq:R_q_sequence}) and (\ref{eq:R_q_prime_sequence}), are very closely related, as shown by the following lemma, whose proof is deferred to Section~\ref{sec:quotient_equivalence}.

\begin{lem}[Quotient equivalence lemma]
\label{lem:series_equiv}
We have $j'=j$ and $\fq'_i=\fq_i$ for all $1\le i\le j-1$. 
\end{lem}

We now show how Lemma~\ref{lem:series_equiv}, together with Lemma~\ref{lem:S_via_the_locator_poly}, implies Theorem~\ref{thm:S_via_Pj_over_Lj}. 
The following is a key identity, analogous to~\eqref{eq:Ri_via_PiLi}.

\begin{clm}
\label{clm:Riprime_via_PiLi}
We have $\fP_i+\fS\fL_i/x = \fR'_i/x^{2t}$ for all $0\le i \le j$. 
\end{clm}

\begin{proof}
We prove the claim by induction on $ i $. For the base cases $ i = 0 $ and $ i = 1 $, recall from~\eqref{eq:polyPL_defs} that $\fP_0 = \fL_1 = 1 $ and $ \fP_1 = \fL_0 = 0 $.
 Therefore the claim holds due to the definitions (\ref{eq:fRprime01_defs}) of $\fR'_0$ and $\fR'_1$.

For the inductive step, assume the claim holds for all $k \le i$. Then,
\begin{multline*}
\frac{\fR'_{i+1}}{x^{2t}}
= \frac{\fR'_{i-1}}{x^{2t}}-\fq'_i\frac{\fR'_i}{x^{2t}}
= \frac{\fR'_{i-1}}{x^{2t}}-\fq_i\frac{\fR'_i}{x^{2t}}
= \fP_{i-1}+\frac{\fS\fL_{i-1}}{x}-\fq_i
\left(\fP_{i}+\frac{\fS\fL_{i}}{x}\right)
\\ = \left(\fP_{i-1}-\fq_i\fP_{i}\right)+\frac{\fS}{x}\left(\fL_{i-1}-\fq_i\fL_{i}\right)
= \fP_{i+1}+\frac{\fS}{x}\fL_{i+1},
\end{multline*}
where the first equality is by the recursive definition (\ref{eq:R_q_prime_recursion}) of $\fR'_{i+1}$, the second uses Lemma~\ref{lem:series_equiv} to replace $\fq'_i$ with $\fq_i$, the third is by the inductive hypothesis, the forth rearranges terms, and, finally, the last is by the recursive definitions (\ref{eq:polyPL_defs}) of $\fP_{i+1}$ and $\fL_{i+1}$.
\end{proof}

We are now ready to complete the proof of Theorem~\ref{thm:S_via_Pj_over_Lj}.
Recall from Lemma~\ref{lem:S_via_the_locator_poly} that $\fS=-x\fP/\fL$, where $\fP$ and $\fL$ are polynomials with $\deg\fL\le t$. 
  Substituting this into Claim~\ref{clm:Riprime_via_PiLi} with $i = j$, and multiplying both sides of the equality by $\fL$, we obtain
\[
\fP_j\fL - \fP\fL_j = \fR'_j\fL/x^{2t}.
\]
Since $j = j'$ (by Lemma~\ref{lem:series_equiv}) and $\deg \fR'_{j'} \le t-1 $ (by~\eqref{eq:def_of_jprime}), we have
\[
\deg\big(\fP_j\fL - \fP\fL_j\big)
=\deg\big(\fR'_{j'}\fL/x^{2t}\big) \le (t-1) + t - 2t < 0.
\]
But since $\fP_j$, $\fL$, $\fP$, $\fL_j$ are all polynomials, namely, they do not contain negative powers of $x$, $\deg(\fP_j\fL - \fP\fL_j)<0$ implies $\fP_j\fL - \fP\fL_j=0$. As a result, we have
$-x\fP_j/\fL_j=-x\fP/\fL=\fS$, which concludes the proof of Theorem~\ref{thm:S_via_Pj_over_Lj}.

\subsection{Proof of the quotient equivalence lemma}
\label{sec:quotient_equivalence}

Here we prove here the following lemma, which subsumes Lemma~\ref{lem:series_equiv}.

\begin{lem}
\label{lem:series_equiv_Plus}
We have $j=j'$,  and  
\begin{alignat*}{2}
&\fq_i=\fq'_i
&\qquad&\text{for all}\quad 1\le i\le j-1,
\\
&\polyr{\fR_{i}}_{\sum_{k=1}^{i-1}\deg\fq_{k}} = \polyr{\fR'_{i}}_{\sum_{k=1}^{i-1}\deg\fq_{k}}
&\qquad&\text{for all}\quad 0\le i\le j.
\end{alignat*}
\end{lem}

To start with, first note that the last statement of the lemma holds for $i=0$ and $i=1$, because 
$\sum_{k=1}^{-1}\deg\fq_{k}=\sum_{k=1}^{0}\deg\fq_{k}=0$ and we have both $\fR_{0}=\fR'_{0}$ and $\fR_{1}=\polyr{\fR'_{1}}$. This will, effectively, serve as inductive basis, with the induction running over $i$.

Note that it might already be the case that $\deg\fR'_{1}<t$ and $\deg\fR_{1}<t$, which happens if (and only if) $y=0$. In such a case, $j=j'=1$, and we are done. 

\bigskip

Below we present Claims~\ref{clm:next_quotient_equality}, \ref{clm:next_rounding_of_Ri}, and \ref{clm:Ri_degrees}, which jointly prove Lemma~\ref{lem:series_equiv_Plus}. We can think of these three claims as forming a while loop that incrementally appends elements to sequences (\ref{eq:R_q_sequence}) and (\ref{eq:R_q_prime_sequence}).

We can think of Claims~\ref{clm:next_quotient_equality} and \ref{clm:next_rounding_of_Ri} forming the body of the loop. If we have not terminated thus far---we have both $\deg\fR_i\ge t$ and $\deg\fR'_i\ge t$---then Claim~\ref{clm:next_quotient_equality} appends the same quotient $\fq_i=\fq'_i$. Using this quotient, we can obtain remainders $\fR_{i+1}$ and $\fR'_{i+1}$ to be appended to sequences (\ref{eq:R_q_sequence}) and (\ref{eq:R_q_prime_sequence}), respectively. Claim~\ref{clm:next_rounding_of_Ri} then ensures that certain number of initial coefficients of $\fR_{i+1}$ and $\fR'_{i+1}$ all match.
Finally, Claim~\ref{clm:Ri_degrees} governs the termination of this iterative process, and tells us that the halting conditions for both sequences are achieved after the same number of iterations.

\begin{clm}
\label{clm:next_quotient_equality}
Suppose $i\ge 1$ is such that
\begin{alignat*}{2}
&\fq_k=\fq'_k
&\qquad&\text{for all}\quad 1\le k\le i-1,
\\
&\polyr{\fR_{h}}_{\sum_{k=1}^{h-1}\deg\fq_{k}} = \polyr{\fR'_{h}}_{\sum_{k=1}^{h-1}\deg\fq_{k}}
&\qquad&\text{for all}\quad 0\le h\le i,
\end{alignat*}
and that both $\deg\fR_i\ge t$ and $\deg\fR'_i\ge t$.
Then we also have $\fq_i=\fq'_i$.
\end{clm}

\begin{clm}
\label{clm:next_rounding_of_Ri}
Suppose $i\ge 1$ is such that
\begin{alignat*}{2}
&\fq_k=\fq'_k
&\qquad&\text{for all}\quad 1\le k\le i,
\\
&\polyr{\fR_{h}}_{\sum_{k=1}^{h-1}\deg\fq_{k}} = \polyr{\fR'_{h}}_{\sum_{k=1}^{h-1}\deg\fq_{k}}
&\qquad&\text{for all}\quad 0\le h\le i.
\end{alignat*}
Then we also have
$\polyr{\fR_{i+1}}_{\sum_{k=1}^{i}\deg\fq_{k}}=\polyr{\fR'_{i+1}}_{\sum_{k=1}^{i}\deg\fq_{k}}$.
\end{clm}

\begin{clm}
\label{clm:Ri_degrees}
Suppose $i\ge 1$ is such that 
\begin{alignat*}{2}
&\fq_k=\fq'_k
&\qquad&\text{for all}\quad 1\le k\le i,
\\
&\polyr{\fR_{h}}_{\sum_{k=1}^{h-1}\deg\fq_{k}} = \polyr{\fR'_{h}}_{\sum_{k=1}^{h-1}\deg\fq_{k}}
&\qquad&\text{for all}\quad 0\le h\le i+1,
\end{alignat*}
and that both $\deg\fR_{i}\ge t$ and  $\deg\fR'_{i}\ge t$.
Then we have $\deg\fR_{i+1}\le t-1$ if and only if $\deg\fR'_{i+1}\le t-1$.
\end{clm}

It remains to prove these three claims.  Before turning to their individual proofs, we begin with some general observations about rounding of formal series.

\begin{fact}
\label{fact:series_rounding}
Let $\fa,\fb,\fc\in \F_p(\!(x)\!)$ and $d,d'\in\Z$. The following statements hold:
\begin{enumerate}
\item[(a)] If $\fa$ and $\fb$ are both non-zero with $\deg\fa\ge\deg\fb$, and
\begin{align*}
& d \le \deg\fa-(\deg\fa-\deg\fb)=\deg\fb,
\\ &d' \le \deg\fb-(\deg\fa-\deg\fb)=2\deg\fb-\deg\fa,
\end{align*}
then $\polyr{\fa/\fb} = \polyr{\polyr{\fa}_{d}/\polyr{\fb}_{d'}}$.
\item[(b)] If $\polyr{\fa}_{d}=\polyr{\fb}_{d}$ and $d'\ge d$, then $\polyr{\fa}_{d'}=\polyr{\fb}_{d'}$.
\item[(c)] If $\polyr{\fa}_d=\polyr{\fb}_d$, then $\polyr{\fa\fc}_{d+\deg\fc}=\polyr{\fb\fc}_{d+\deg\fc}$.
\end{enumerate}
\end{fact}

\begin{proof}[Proof sketch.]
Point (c) holds because, for any $k\ge 1$, the $k$ most significant coefficients of the product $\fa\fc$ depend only on the $k$ most significant coefficients of both $\fa$ and $\fc$. Point (a) then holds because the $k:=\deg\fa-\deg\fb+1$ most significant coefficients of the inverse $\fb^{-1}$ depend only on the $k$ most significant coefficients of $\fb$.
For Point (b), observe that if $d'\ge d$, then $\polyr{\fa}_{d'}=\polyr{\polyr{\fa}_{d}}_{d'}$.
\end{proof}

\noindent
{\bf{}Claim~\ref{clm:next_quotient_equality}} (restated){\bf{}.}
\emph{Suppose $i\ge 1$ is such that
\begin{alignat}{2}
\notag
&\fq_k=\fq'_k
&\qquad&\text{for all}\quad 1\le k\le i-1,
\\
\label{eq:next_quotient_equality_B}
&\polyr{\fR_{h}}_{\sum_{k=1}^{h-1}\deg\fq_{k}} = \polyr{\fR'_{h}}_{\sum_{k=1}^{h-1}\deg\fq_{k}}
&\qquad&\text{for all}\quad 0\le h\le i,
\end{alignat}
and that both $\deg\fR_i\ge t$ and $\deg\fR'_i\ge t$.
Then we also have $\fq_i=\fq'_i$.
}

\begin{proof}
We apply Fact~\ref{fact:series_rounding}(a) with $d=\sum_{k=1}^{i-2} \deg \fq_k$ and $d'=\sum_{k=1}^{i-1} \deg \fq_k$.
Recall that $\deg\fq_{k}=\deg\fR_{k-1}-\deg\fR_{k}\ge 1$ for all $k$. Hence,
\[
d'=\sum_{k=1}^{i-1} \deg \fq_k
= \deg \fR_0 - \deg \fR_{i-1}
= 2t - \deg \fR_{i-1}
\le 2 \deg \fR_i - \deg \fR_{i-1},
\]
where we used the assumption that $\deg\fR_i\ge t$.
Building upon this inequality, we get
\[
d=
\sum_{k=1}^{i-2} \deg \fq_k
\le \sum_{k=1}^{i-1} \deg \fq_k
\le 2 \deg \fR_i - \deg \fR_{i-1}
\le \deg \fR_i,
\]
where the final inequality follows from $\deg\fR_{i}\le\deg\fR_{i-1}$.
Similarly, because $\fq'_k=\fq_k$ for all $k\le i-1$, we also have
\[
d'=\sum_{k=1}^{i-1} \deg \fq_k
\le 2 \deg \fR'_i - \deg \fR'_{i-1}
\qqAnd
d = \sum_{k=1}^{i-2}\deg\fq_{k} \le \deg\fR'_i.
\]
Hence, we obtain
\[
\fq_i
 =\polyr{\fR_{i-1}/\fR_i}
 = \polyr{\polyr{\fR_{i-1}}_{d}/\polyr{\fR_i}_{d'}}
 = \polyr{\polyr{\fR'_{i-1}}_{d}/\polyr{\fR'_i}_{d'}}
 =\polyr{\fR'_{i-1}/\fR'_i} = \fq'_i,
\]
where the second and fourth equalities follow from Fact~\ref{fact:series_rounding}(a), and the third equality uses the assumption~\eqref{eq:next_quotient_equality_B} with $h = i - 1$ and $h = i$.
\end{proof}


\noindent
{\bf{}Claim~\ref{clm:next_rounding_of_Ri}} (restated){\bf{}.}
\emph{
Suppose $i\ge 1$ is such that
\begin{alignat}{2}
\label{eq:next_rounding_of_Ri_A}
&\fq_k=\fq'_k
&\qquad&\text{for all}\quad 1\le k\le i,
\\
\label{eq:next_rounding_of_Ri_B}
&\polyr{\fR_{h}}_{\sum_{k=1}^{h-1}\deg\fq_{k}} = \polyr{\fR'_{h}}_{\sum_{k=1}^{h-1}\deg\fq_{k}}
&\qquad&\text{for all}\quad 0\le h\le i.
\end{alignat}
Then we also have
$\polyr{\fR_{i+1}}_{\sum_{k=1}^{i}\deg\fq_{k}}=\polyr{\fR'_{i+1}}_{\sum_{k=1}^{i}\deg\fq_{k}}$.
}

\begin{proof}
We only require assumption (\ref{eq:next_rounding_of_Ri_A}) for $k=i$, and assumption (\ref{eq:next_rounding_of_Ri_B}) for $h=i-1$ and $h=i$.
By definition, we have $\fR_{i+1}=\fR_{i-1}-\fq_i\fR_{i}$ and $\fR'_{i+1}=\fR'_{i-1}-\fq_i\fR'_{i}$, where in the second identity we used that $\fq_i = \fq'_i$.
It therefore suffices to show that both
\begin{alignat*}{2}
&\polyr{\fR_{i-1}}_{\sum_{k=1}^{i}\deg\fq_{k}}&=&\polyr{\fR'_{i-1}}_{\sum_{k=1}^{i}\deg\fq_{k}},
\\ &
\polyr{\fq_i\fR_{i}}_{\deg\fq_{i}+\sum_{k=1}^{i-1}\deg\fq_{k}}&=&\polyr{\fq_i\fR'_{i}}_{\deg\fq_{i}+\sum_{k=1}^{i-1}\deg\fq_{k}}.
\end{alignat*}
The first equality follows from~\eqref{eq:next_rounding_of_Ri_B} with $h = i - 1$, together with Fact~\ref{fact:series_rounding}(b), for which we use
$\sum_{k=1}^{i-2} \deg \fq_k \le \sum_{k=1}^{i} \deg \fq_k$.
The second follows from~\eqref{eq:next_rounding_of_Ri_B} with $h = i$, along with Fact~\ref{fact:series_rounding}(c).
\end{proof}


\noindent
{\bf{}Claim~\ref{clm:Ri_degrees}} (restated){\bf{}.}
\emph{
Suppose $i\ge 1$ is such that 
\begin{alignat}{2}
\notag
&\fq_k=\fq'_k
&\qquad&\text{for all}\quad 1\le k\le i,
\\
\label{eq:Ri_degrees_B}
&\polyr{\fR_{h}}_{\sum_{k=1}^{h-1}\deg\fq_{k}} = \polyr{\fR'_{h}}_{\sum_{k=1}^{h-1}\deg\fq_{k}}
&\qquad&\text{for all}\quad 0\le h\le i+1,
\end{alignat}
and that both $\deg\fR_{i}\ge t$ and  $\deg\fR'_{i}\ge t$.
Then we have $\deg\fR_{i+1}\le t-1$ if and only if $\deg\fR'_{i+1}\le t-1$.
}

\begin{proof}
Let us apply Fact~\ref{fact:series_rounding}(b) with $d = \sum_{k=1}^{i} \deg \fq_k$ and $d' = t$. This choice is valid because, from the identity $\deg \fq_k = \deg \fR_{k-1} - \deg \fR_k$, we obtain
\[
d = \sum_{k=1}^{i} \deg \fq_k
= \deg \fR_0 - \deg \fR_i
= 2t - \deg \fR_i
\le 2t - t = t = d',
\]
where the sole inequality uses the assumption that $\deg \fR_i \ge t$.
Assumption (\ref{eq:Ri_degrees_B}) with $h={i+1}$ states that $\polyr{\fR_{i+1}}_{d}=\polyr{\fR'_{i+1}}_{d}$, from which Fact~\ref{fact:series_rounding}(b) implies 
$\polyr{\fR_{i+1}}_{t}=\polyr{\fR'_{i+1}}_{t}$.
We conclude by observing that $\deg\fR_{i+1}< t$ if and only if $\polyr{\fR_{i+1}}_{t}=0$, and
$\deg\fR'_{i+1}< t$ if and only if $\polyr{\fR'_{i+1}}_{t}=0$.
\end{proof}

\section{Simulation of QRAM in the standard model}
\label{app:QRAM}

Here we illustrate two quantum circuits that simulate the $R_{\mathrm Q}$ gate of the \QRAMq model in the standard circuit model for the case $M = 8$.  The first circuit (Figure~\ref{fig:RQ_linear_ancilla}) uses $M$ ancilla qubits, while the second (Figure~\ref{fig:RQ_log_ancilla}) uses only $\lceil \log_2 M \rceil$ ancilla qubits. In both circuit diagrams, $|\overline{\ell}\>$ denotes the state $|\ell_0, \ell_1, \ldots, \ell_{M-1}\>$, and $|\overline{0}\>$ denotes the all-zero state $|0, 0, \ldots, 0\>$, which is both the initial and final state of the ancilla qubits.

\def\ygap{0.2} 
\def\reggap{0.5} 
\def\xgap{0.25} 
\def\rCXctrl{0.05} 
\def\rCXnot{0.1} 
\def\rSwap{0.08} 

\def\groupPad{0.14} 
\def\groupRound{3} 

\def\yRecordTop{4.7}
\def\yFlagRec{3.7}
\def\yTruthTop{3.2}
\def\yIndexTop{2}
\def\yTarget{1}
\def\yFlagAlgo{0.41}

\def\xWireLeft{0}
\def\xWireMid{3.4}
\def\xWireRight{6.42}

\def\xOracle{1}
\def\yOraclePad{0.2} 
\def\xFlagCopy{4.2}
\def\xIndexCopy{5.0}
\def\xKetEnd{3.35} 

\def\xGateGap{0.35} 

\def\xpad{0.95} 


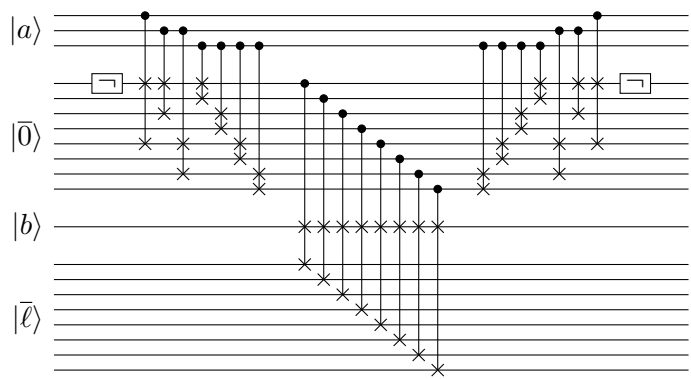
\begin{figure}[!h]
\centering
\begin{tikzpicture}

\foreach \i in {0,...,2} {
     \draw (-\xpad,-\ygap*\i)--(\xpad+23*\xgap+2*\xGateGap,-\ygap*\i);
}
\draw (-\xpad,-2*\reggap-9*\ygap)--(\xpad+23*\xgap+2*\xGateGap,-2*\reggap-9*\ygap);
\foreach \i in {0,...,7} {
     \draw (-\xpad,-\reggap-\ygap*2-\ygap*\i)--(\xpad+23*\xgap+2*\xGateGap,-\reggap-\ygap*2-\ygap*\i);
     \draw (-\xpad,-3*\reggap-9*\ygap-\ygap*\i)--(\xpad+23*\xgap+2*\xGateGap,-3*\reggap-9*\ygap-\ygap*\i);
}

\node [left] at (-\xpad,-0*\reggap-\ygap*1) {$|a\>$};
\node [left] at (-\xpad,-1*\reggap-\ygap*5.5) {$|\overline{0}\>$};
\node [left] at (-\xpad,-2*\reggap-\ygap*9) {$|b\>$};
\node [left] at (-\xpad,-3*\reggap-\ygap*12.5) {$|\overline\ell\>$};

\draw (1*\xgap,-0*\ygap)--(1*\xgap,-\reggap-6*\ygap);
\draw [fill = black] (1*\xgap,-0*\ygap) circle (\rCXctrl);
\draw (1*\xgap-\rSwap,-\reggap-2*\ygap-\rSwap)--(1*\xgap+\rSwap,-\reggap-2*\ygap+\rSwap);
\draw (1*\xgap-\rSwap,-\reggap-2*\ygap+\rSwap)--(1*\xgap+\rSwap,-\reggap-2*\ygap-\rSwap);
\draw (1*\xgap-\rSwap,-\reggap-6*\ygap-\rSwap)--(1*\xgap+\rSwap,-\reggap-6*\ygap+\rSwap);
\draw (1*\xgap-\rSwap,-\reggap-6*\ygap+\rSwap)--(1*\xgap+\rSwap,-\reggap-6*\ygap-\rSwap);

\draw (2*\xgap,-1*\ygap)--(2*\xgap,-\reggap-4*\ygap);
\draw [fill = black] (2*\xgap,-1*\ygap) circle (\rCXctrl);
\draw (2*\xgap-\rSwap,-\reggap-2*\ygap-\rSwap)--(2*\xgap+\rSwap,-\reggap-2*\ygap+\rSwap);
\draw (2*\xgap-\rSwap,-\reggap-2*\ygap+\rSwap)--(2*\xgap+\rSwap,-\reggap-2*\ygap-\rSwap);
\draw (2*\xgap-\rSwap,-\reggap-4*\ygap-\rSwap)--(2*\xgap+\rSwap,-\reggap-4*\ygap+\rSwap);
\draw (2*\xgap-\rSwap,-\reggap-4*\ygap+\rSwap)--(2*\xgap+\rSwap,-\reggap-4*\ygap-\rSwap);

\draw (3*\xgap,-1*\ygap)--(3*\xgap,-\reggap-8*\ygap);
\draw [fill = black] (3*\xgap,-1*\ygap) circle (\rCXctrl);
\draw (3*\xgap-\rSwap,-\reggap-6*\ygap-\rSwap)--(3*\xgap+\rSwap,-\reggap-6*\ygap+\rSwap);
\draw (3*\xgap-\rSwap,-\reggap-6*\ygap+\rSwap)--(3*\xgap+\rSwap,-\reggap-6*\ygap-\rSwap);
\draw (3*\xgap-\rSwap,-\reggap-8*\ygap-\rSwap)--(3*\xgap+\rSwap,-\reggap-8*\ygap+\rSwap);
\draw (3*\xgap-\rSwap,-\reggap-8*\ygap+\rSwap)--(3*\xgap+\rSwap,-\reggap-8*\ygap-\rSwap);

\draw (4*\xgap,-2*\ygap)--(4*\xgap,-\reggap-3*\ygap);
\draw [fill = black] (4*\xgap,-2*\ygap) circle (\rCXctrl);
\draw (4*\xgap-\rSwap,-\reggap-2*\ygap-\rSwap)--(4*\xgap+\rSwap,-\reggap-2*\ygap+\rSwap);
\draw (4*\xgap-\rSwap,-\reggap-2*\ygap+\rSwap)--(4*\xgap+\rSwap,-\reggap-2*\ygap-\rSwap);
\draw (4*\xgap-\rSwap,-\reggap-3*\ygap-\rSwap)--(4*\xgap+\rSwap,-\reggap-3*\ygap+\rSwap);
\draw (4*\xgap-\rSwap,-\reggap-3*\ygap+\rSwap)--(4*\xgap+\rSwap,-\reggap-3*\ygap-\rSwap);

\draw (5*\xgap,-2*\ygap)--(5*\xgap,-\reggap-5*\ygap);
\draw [fill = black] (5*\xgap,-2*\ygap) circle (\rCXctrl);
\draw (5*\xgap-\rSwap,-\reggap-4*\ygap-\rSwap)--(5*\xgap+\rSwap,-\reggap-4*\ygap+\rSwap);
\draw (5*\xgap-\rSwap,-\reggap-4*\ygap+\rSwap)--(5*\xgap+\rSwap,-\reggap-4*\ygap-\rSwap);
\draw (5*\xgap-\rSwap,-\reggap-5*\ygap-\rSwap)--(5*\xgap+\rSwap,-\reggap-5*\ygap+\rSwap);
\draw (5*\xgap-\rSwap,-\reggap-5*\ygap+\rSwap)--(5*\xgap+\rSwap,-\reggap-5*\ygap-\rSwap);

\draw (6*\xgap,-2*\ygap)--(6*\xgap,-\reggap-7*\ygap);
\draw [fill = black] (6*\xgap,-2*\ygap) circle (\rCXctrl);
\draw (6*\xgap-\rSwap,-\reggap-6*\ygap-\rSwap)--(6*\xgap+\rSwap,-\reggap-6*\ygap+\rSwap);
\draw (6*\xgap-\rSwap,-\reggap-6*\ygap+\rSwap)--(6*\xgap+\rSwap,-\reggap-6*\ygap-\rSwap);
\draw (6*\xgap-\rSwap,-\reggap-7*\ygap-\rSwap)--(6*\xgap+\rSwap,-\reggap-7*\ygap+\rSwap);
\draw (6*\xgap-\rSwap,-\reggap-7*\ygap+\rSwap)--(6*\xgap+\rSwap,-\reggap-7*\ygap-\rSwap);

\draw (7*\xgap,-2*\ygap)--(7*\xgap,-\reggap-9*\ygap);
\draw [fill = black] (7*\xgap,-2*\ygap) circle (\rCXctrl);
\draw (7*\xgap-\rSwap,-\reggap-8*\ygap-\rSwap)--(7*\xgap+\rSwap,-\reggap-8*\ygap+\rSwap);
\draw (7*\xgap-\rSwap,-\reggap-8*\ygap+\rSwap)--(7*\xgap+\rSwap,-\reggap-8*\ygap-\rSwap);
\draw (7*\xgap-\rSwap,-\reggap-9*\ygap-\rSwap)--(7*\xgap+\rSwap,-\reggap-9*\ygap+\rSwap);
\draw (7*\xgap-\rSwap,-\reggap-9*\ygap+\rSwap)--(7*\xgap+\rSwap,-\reggap-9*\ygap-\rSwap);

\foreach \i in {0,...,7} {
\draw (8*\xgap+\i*\xgap+1*\xGateGap,-\reggap-2*\ygap-\i*\ygap)--(8*\xgap+\i*\xgap+1*\xGateGap,-3*\reggap-9*\ygap-\i*\ygap);
\draw [fill = black] (8*\xgap+\i*\xgap+1*\xGateGap,-\reggap-2*\ygap-\i*\ygap) circle (\rCXctrl);
\draw (8*\xgap+\i*\xgap+1*\xGateGap-\rSwap,-2*\reggap-9*\ygap-\rSwap)--(8*\xgap+\i*\xgap+1*\xGateGap+\rSwap,-2*\reggap-9*\ygap+\rSwap);
\draw (8*\xgap+\i*\xgap+1*\xGateGap-\rSwap,-2*\reggap-9*\ygap+\rSwap)--(8*\xgap+\i*\xgap+1*\xGateGap+\rSwap,-2*\reggap-9*\ygap-\rSwap);
\draw (8*\xgap+\i*\xgap+1*\xGateGap-\rSwap,-3*\reggap-9*\ygap-\i*\ygap-\rSwap)--(8*\xgap+\i*\xgap+1*\xGateGap+\rSwap,-3*\reggap-9*\ygap-\i*\ygap+\rSwap);
\draw (8*\xgap+\i*\xgap+1*\xGateGap-\rSwap,-3*\reggap-9*\ygap-\i*\ygap+\rSwap)--(8*\xgap+\i*\xgap+1*\xGateGap+\rSwap,-3*\reggap-9*\ygap-\i*\ygap-\rSwap);
}

\draw (16*\xgap+2*\xGateGap,-2*\ygap)--(16*\xgap+2*\xGateGap,-\reggap-9*\ygap);
\draw [fill = black] (16*\xgap+2*\xGateGap,-2*\ygap) circle (\rCXctrl);
\draw (16*\xgap+2*\xGateGap-\rSwap,-\reggap-8*\ygap-\rSwap)--(16*\xgap+2*\xGateGap+\rSwap,-\reggap-8*\ygap+\rSwap);
\draw (16*\xgap+2*\xGateGap-\rSwap,-\reggap-8*\ygap+\rSwap)--(16*\xgap+2*\xGateGap+\rSwap,-\reggap-8*\ygap-\rSwap);
\draw (16*\xgap+2*\xGateGap-\rSwap,-\reggap-9*\ygap-\rSwap)--(16*\xgap+2*\xGateGap+\rSwap,-\reggap-9*\ygap+\rSwap);
\draw (16*\xgap+2*\xGateGap-\rSwap,-\reggap-9*\ygap+\rSwap)--(16*\xgap+2*\xGateGap+\rSwap,-\reggap-9*\ygap-\rSwap);

\draw (17*\xgap+2*\xGateGap,-2*\ygap)--(17*\xgap+2*\xGateGap,-\reggap-7*\ygap);
\draw [fill = black] (17 *\xgap+2*\xGateGap,-2*\ygap) circle (\rCXctrl);
\draw (17*\xgap+2*\xGateGap-\rSwap,-\reggap-6*\ygap-\rSwap)--(17*\xgap+2*\xGateGap+\rSwap,-\reggap-6*\ygap+\rSwap);
\draw (17*\xgap+2*\xGateGap-\rSwap,-\reggap-6*\ygap+\rSwap)--(17*\xgap+2*\xGateGap+\rSwap,-\reggap-6*\ygap-\rSwap);
\draw (17*\xgap+2*\xGateGap-\rSwap,-\reggap-7*\ygap-\rSwap)--(17*\xgap+2*\xGateGap+\rSwap,-\reggap-7*\ygap+\rSwap);
\draw (17*\xgap+2*\xGateGap-\rSwap,-\reggap-7*\ygap+\rSwap)--(17*\xgap+2*\xGateGap+\rSwap,-\reggap-7*\ygap-\rSwap);

\draw (18*\xgap+2*\xGateGap,-2*\ygap)--(18*\xgap+2*\xGateGap,-\reggap-5*\ygap);
\draw [fill = black] (18 *\xgap+2*\xGateGap,-2*\ygap) circle (\rCXctrl);
\draw (18*\xgap+2*\xGateGap-\rSwap,-\reggap-4*\ygap-\rSwap)--(18*\xgap+2*\xGateGap+\rSwap,-\reggap-4*\ygap+\rSwap);
\draw (18*\xgap+2*\xGateGap-\rSwap,-\reggap-4*\ygap+\rSwap)--(18*\xgap+2*\xGateGap+\rSwap,-\reggap-4*\ygap-\rSwap);
\draw (18*\xgap+2*\xGateGap-\rSwap,-\reggap-5*\ygap-\rSwap)--(18*\xgap+2*\xGateGap+\rSwap,-\reggap-5*\ygap+\rSwap);
\draw (18*\xgap+2*\xGateGap-\rSwap,-\reggap-5*\ygap+\rSwap)--(18*\xgap+2*\xGateGap+\rSwap,-\reggap-5*\ygap-\rSwap);

\draw (19*\xgap+2*\xGateGap,-2*\ygap)--(19*\xgap+2*\xGateGap,-\reggap-3*\ygap);
\draw [fill = black] (19*\xgap+2*\xGateGap,-2*\ygap) circle (\rCXctrl);
\draw (19*\xgap+2*\xGateGap-\rSwap,-\reggap-2*\ygap-\rSwap)--(19*\xgap+2*\xGateGap+\rSwap,-\reggap-2*\ygap+\rSwap);
\draw (19*\xgap+2*\xGateGap-\rSwap,-\reggap-2*\ygap+\rSwap)--(19*\xgap+2*\xGateGap+\rSwap,-\reggap-2*\ygap-\rSwap);
\draw (19*\xgap+2*\xGateGap-\rSwap,-\reggap-3*\ygap-\rSwap)--(19*\xgap+2*\xGateGap+\rSwap,-\reggap-3*\ygap+\rSwap);
\draw (19*\xgap+2*\xGateGap-\rSwap,-\reggap-3*\ygap+\rSwap)--(19*\xgap+2*\xGateGap+\rSwap,-\reggap-3*\ygap-\rSwap);

\draw (20*\xgap+2*\xGateGap,-1*\ygap)--(20*\xgap+2*\xGateGap,-\reggap-8*\ygap);
\draw [fill = black] (20*\xgap+2*\xGateGap,-1*\ygap) circle (\rCXctrl);
\draw (20*\xgap+2*\xGateGap-\rSwap,-\reggap-6*\ygap-\rSwap)--(20*\xgap+2*\xGateGap+\rSwap,-\reggap-6*\ygap+\rSwap);
\draw (20*\xgap+2*\xGateGap-\rSwap,-\reggap-6*\ygap+\rSwap)--(20*\xgap+2*\xGateGap+\rSwap,-\reggap-6*\ygap-\rSwap);
\draw (20*\xgap+2*\xGateGap-\rSwap,-\reggap-8*\ygap-\rSwap)--(20*\xgap+2*\xGateGap+\rSwap,-\reggap-8*\ygap+\rSwap);
\draw (20*\xgap+2*\xGateGap-\rSwap,-\reggap-8*\ygap+\rSwap)--(20*\xgap+2*\xGateGap+\rSwap,-\reggap-8*\ygap-\rSwap);

\draw (21*\xgap+2*\xGateGap,-1*\ygap)--(21*\xgap+2*\xGateGap,-\reggap-4*\ygap);
\draw [fill = black] (21*\xgap+2*\xGateGap,-1*\ygap) circle (\rCXctrl);
\draw (21*\xgap+2*\xGateGap-\rSwap,-\reggap-2*\ygap-\rSwap)--(21*\xgap+2*\xGateGap+\rSwap,-\reggap-2*\ygap+\rSwap);
\draw (21*\xgap+2*\xGateGap-\rSwap,-\reggap-2*\ygap+\rSwap)--(21*\xgap+2*\xGateGap+\rSwap,-\reggap-2*\ygap-\rSwap);
\draw (21*\xgap+2*\xGateGap-\rSwap,-\reggap-4*\ygap-\rSwap)--(21*\xgap+2*\xGateGap+\rSwap,-\reggap-4*\ygap+\rSwap);
\draw (21*\xgap+2*\xGateGap-\rSwap,-\reggap-4*\ygap+\rSwap)--(21*\xgap+2*\xGateGap+\rSwap,-\reggap-4*\ygap-\rSwap);

\draw (22*\xgap+2*\xGateGap,-0*\ygap)--(22*\xgap+2*\xGateGap,-\reggap-6*\ygap);
\draw [fill = black] (22*\xgap+2*\xGateGap,-0*\ygap) circle (\rCXctrl);
\draw (22*\xgap+2*\xGateGap-\rSwap,-\reggap-2*\ygap-\rSwap)--(22*\xgap+2*\xGateGap+\rSwap,-\reggap-2*\ygap+\rSwap);
\draw (22*\xgap+2*\xGateGap-\rSwap,-\reggap-2*\ygap+\rSwap)--(22*\xgap+2*\xGateGap+\rSwap,-\reggap-2*\ygap-\rSwap);
\draw (22*\xgap+2*\xGateGap-\rSwap,-\reggap-6*\ygap-\rSwap)--(22*\xgap+2*\xGateGap+\rSwap,-\reggap-6*\ygap+\rSwap);
\draw (22*\xgap+2*\xGateGap-\rSwap,-\reggap-6*\ygap+\rSwap)--(22*\xgap+2*\xGateGap+\rSwap,-\reggap-6*\ygap-\rSwap);

\draw [fill=white] (-1*\xgap-0.2,-\reggap-2*\ygap+0.13) rectangle (-1*\xgap+0.2,-\reggap-2*\ygap-0.13);
\draw [fill=white] (24*\xgap+2*\xGateGap-0.2,-\reggap-2*\ygap+0.13) rectangle (24*\xgap+2*\xGateGap+0.2,-\reggap-2*\ygap-0.13);
\node at (-1*\xgap,-1*\reggap-\ygap*2) {$\neg$};
\node at (24*\xgap+2*\xGateGap,-1*\reggap-\ygap*2) {$\neg$};

\end{tikzpicture}
\captionsetup{font=small}
\captionsetup{width=0.9\textwidth}
\caption[my caption]{%
Simulation of $R_{\mathrm{Q}}$ using the one-hot encoding of $a$.
The circuit employs exactly $3M - 2$ Fredkin gates.
Note that the two negation gates can be omitted if the ancilla is initialized and restored to the state $|1\> \otimes |0\>^{\otimes M-1}$. }
\label{fig:RQ_linear_ancilla}
\end{figure}


\def\xpad{0.50} 

\begin{figure}[!h]
\centering
\begin{tikzpicture}

\foreach \i in {0,...,2} {
     \draw (-\xpad,-\ygap*\i)--(\xpad+35*\xgap,-\ygap*\i);
     \draw (-\xpad,-\reggap-\ygap*2-\ygap*\i)--(\xpad+35*\xgap,-\reggap-\ygap*2-\ygap*\i);
}
\draw (-\xpad,-2*\reggap-4*\ygap)--(\xpad+35*\xgap,-2*\reggap-4*\ygap);
\foreach \i in {0,...,7} {
     \draw (-\xpad,-3*\reggap-4*\ygap-\ygap*\i)--(\xpad+35*\xgap,-3*\reggap-4*\ygap-\ygap*\i);
}

\node [left] at (-\xpad,-0*\reggap-\ygap*1) {$|a\>$};
\node [left] at (-\xpad,-1*\reggap-\ygap*3) {$|\overline{0}\>$};
\node [left] at (-\xpad,-2*\reggap-\ygap*4) {$|b\>$};
\node [left] at (-\xpad,-3*\reggap-\ygap*7.5) {$|\overline\ell\>$};

\draw (0,0)--(0,-\reggap-\ygap*2-\rCXnot);
\draw [fill = white] (0,0) circle (\rCXctrl);
\draw (0,-\reggap-\ygap*2) circle (\rCXnot);

\draw (1*\xgap,-1*\ygap)--(1*\xgap,-\reggap-3*\ygap-\rCXnot);
\draw [fill = white] (1*\xgap,-1*\ygap) circle (\rCXctrl);
\draw [fill = black] (1*\xgap,-\reggap-2*\ygap) circle (\rCXctrl);
\draw (1*\xgap,-\reggap-3*\ygap) circle (\rCXnot);

\draw (2*\xgap,-2*\ygap)--(2*\xgap,-\reggap-4*\ygap-\rCXnot);
\draw [fill = white] (2*\xgap,-2*\ygap) circle (\rCXctrl);
\draw [fill = black] (2*\xgap,-\reggap-3*\ygap) circle (\rCXctrl);
\draw (2*\xgap,-\reggap-4*\ygap) circle (\rCXnot);

\draw (3*\xgap,-\reggap-4*\ygap)--(3*\xgap,-3*\reggap-4*\ygap);
\draw [fill = black] (3*\xgap,-\reggap-4*\ygap) circle (\rCXctrl);
\draw (3*\xgap-\rSwap,-2*\reggap-4*\ygap-\rSwap)--(3*\xgap+\rSwap,-2*\reggap-4*\ygap+\rSwap);
\draw (3*\xgap-\rSwap,-2*\reggap-4*\ygap+\rSwap)--(3*\xgap+\rSwap,-2*\reggap-4*\ygap-\rSwap);
\draw (3*\xgap-\rSwap,-3*\reggap-4*\ygap-\rSwap)--(3*\xgap+\rSwap,-3*\reggap-4*\ygap+\rSwap);
\draw (3*\xgap-\rSwap,-3*\reggap-4*\ygap+\rSwap)--(3*\xgap+\rSwap,-3*\reggap-4*\ygap-\rSwap);

\draw (4*\xgap,-2*\ygap)--(4*\xgap,-\reggap-4*\ygap-\rCXnot);
\draw [fill = white] (4*\xgap,-2*\ygap) circle (\rCXctrl);
\draw [fill = black] (4*\xgap,-\reggap-3*\ygap) circle (\rCXctrl);
\draw (4*\xgap,-\reggap-4*\ygap) circle (\rCXnot);

\draw (5*\xgap,-2*\ygap)--(5*\xgap,-\reggap-4*\ygap-\rCXnot);
\draw [fill = black] (5*\xgap,-2*\ygap) circle (\rCXctrl);
\draw [fill = black] (5*\xgap,-\reggap-3*\ygap) circle (\rCXctrl);
\draw (5*\xgap,-\reggap-4*\ygap) circle (\rCXnot);

\draw (6*\xgap,-\reggap-4*\ygap)--(6*\xgap,-3*\reggap-5*\ygap);
\draw [fill = black] (6*\xgap,-\reggap-4*\ygap) circle (\rCXctrl);
\draw (6*\xgap-\rSwap,-2*\reggap-4*\ygap-\rSwap)--(6*\xgap+\rSwap,-2*\reggap-4*\ygap+\rSwap);
\draw (6*\xgap-\rSwap,-2*\reggap-4*\ygap+\rSwap)--(6*\xgap+\rSwap,-2*\reggap-4*\ygap-\rSwap);
\draw (6*\xgap-\rSwap,-3*\reggap-5*\ygap-\rSwap)--(6*\xgap+\rSwap,-3*\reggap-5*\ygap+\rSwap);
\draw (6*\xgap-\rSwap,-3*\reggap-5*\ygap+\rSwap)--(6*\xgap+\rSwap,-3*\reggap-5*\ygap-\rSwap);

\draw (7*\xgap,-2*\ygap)--(7*\xgap,-\reggap-4*\ygap-\rCXnot);
\draw [fill = black] (7*\xgap,-2*\ygap) circle (\rCXctrl);
\draw [fill = black] (7*\xgap,-\reggap-3*\ygap) circle (\rCXctrl);
\draw (7*\xgap,-\reggap-4*\ygap) circle (\rCXnot);

\draw (8*\xgap,-1*\ygap)--(8*\xgap,-\reggap-3*\ygap-\rCXnot);
\draw [fill = white] (8*\xgap,-1*\ygap) circle (\rCXctrl);
\draw [fill = black] (8*\xgap,-\reggap-2*\ygap) circle (\rCXctrl);
\draw (8*\xgap,-\reggap-3*\ygap) circle (\rCXnot);

\draw (9*\xgap,-1*\ygap)--(9*\xgap,-\reggap-3*\ygap-\rCXnot);
\draw [fill = black] (9*\xgap,-1*\ygap) circle (\rCXctrl);
\draw [fill = black] (9*\xgap,-\reggap-2*\ygap) circle (\rCXctrl);
\draw (9*\xgap,-\reggap-3*\ygap) circle (\rCXnot);

\draw (10*\xgap,-2*\ygap)--(10*\xgap,-\reggap-4*\ygap-\rCXnot);
\draw [fill = white] (10*\xgap,-2*\ygap) circle (\rCXctrl);
\draw [fill = black] (10*\xgap,-\reggap-3*\ygap) circle (\rCXctrl);
\draw (10*\xgap,-\reggap-4*\ygap) circle (\rCXnot);

\draw (11*\xgap,-\reggap-4*\ygap)--(11*\xgap,-3*\reggap-6*\ygap);
\draw [fill = black] (11*\xgap,-\reggap-4*\ygap) circle (\rCXctrl);
\draw (11*\xgap-\rSwap,-2*\reggap-4*\ygap-\rSwap)--(11*\xgap+\rSwap,-2*\reggap-4*\ygap+\rSwap);
\draw (11*\xgap-\rSwap,-2*\reggap-4*\ygap+\rSwap)--(11*\xgap+\rSwap,-2*\reggap-4*\ygap-\rSwap);
\draw (11*\xgap-\rSwap,-3*\reggap-6*\ygap-\rSwap)--(11*\xgap+\rSwap,-3*\reggap-6*\ygap+\rSwap);
\draw (11*\xgap-\rSwap,-3*\reggap-6*\ygap+\rSwap)--(11*\xgap+\rSwap,-3*\reggap-6*\ygap-\rSwap);

\draw (12*\xgap,-2*\ygap)--(12*\xgap,-\reggap-4*\ygap-\rCXnot);
\draw [fill = white] (12*\xgap,-2*\ygap) circle (\rCXctrl);
\draw [fill = black] (12*\xgap,-\reggap-3*\ygap) circle (\rCXctrl);
\draw (12*\xgap,-\reggap-4*\ygap) circle (\rCXnot);

\draw (13*\xgap,-2*\ygap)--(13*\xgap,-\reggap-4*\ygap-\rCXnot);
\draw [fill = black] (13*\xgap,-2*\ygap) circle (\rCXctrl);
\draw [fill = black] (13*\xgap,-\reggap-3*\ygap) circle (\rCXctrl);
\draw (13*\xgap,-\reggap-4*\ygap) circle (\rCXnot);

\draw (14*\xgap,-\reggap-4*\ygap)--(14*\xgap,-3*\reggap-7*\ygap);
\draw [fill = black] (14*\xgap,-\reggap-4*\ygap) circle (\rCXctrl);
\draw (14*\xgap-\rSwap,-2*\reggap-4*\ygap-\rSwap)--(14*\xgap+\rSwap,-2*\reggap-4*\ygap+\rSwap);
\draw (14*\xgap-\rSwap,-2*\reggap-4*\ygap+\rSwap)--(14*\xgap+\rSwap,-2*\reggap-4*\ygap-\rSwap);
\draw (14*\xgap-\rSwap,-3*\reggap-7*\ygap-\rSwap)--(14*\xgap+\rSwap,-3*\reggap-7*\ygap+\rSwap);
\draw (14*\xgap-\rSwap,-3*\reggap-7*\ygap+\rSwap)--(14*\xgap+\rSwap,-3*\reggap-7*\ygap-\rSwap);

\draw (15*\xgap,-2*\ygap)--(15*\xgap,-\reggap-4*\ygap-\rCXnot);
\draw [fill = black] (15*\xgap,-2*\ygap) circle (\rCXctrl);
\draw [fill = black] (15*\xgap,-\reggap-3*\ygap) circle (\rCXctrl);
\draw (15*\xgap,-\reggap-4*\ygap) circle (\rCXnot);

\draw (16*\xgap,-1*\ygap)--(16*\xgap,-\reggap-3*\ygap-\rCXnot);
\draw [fill = black] (16*\xgap,-1*\ygap) circle (\rCXctrl);
\draw [fill = black] (16*\xgap,-\reggap-2*\ygap) circle (\rCXctrl);
\draw (16*\xgap,-\reggap-3*\ygap) circle (\rCXnot);

\draw (17*\xgap,0)--(17*\xgap,-\reggap-\ygap*2-\rCXnot);
\draw [fill = white] (17*\xgap,0) circle (\rCXctrl);
\draw (17*\xgap,-\reggap-\ygap*2) circle (\rCXnot);

\draw (18*\xgap,0)--(18*\xgap,-\reggap-\ygap*2-\rCXnot);
\draw [fill = black] (18*\xgap,0) circle (\rCXctrl);
\draw (18*\xgap,-\reggap-\ygap*2) circle (\rCXnot);

\draw (19*\xgap,-1*\ygap)--(19*\xgap,-\reggap-3*\ygap-\rCXnot);
\draw [fill = white] (19*\xgap,-1*\ygap) circle (\rCXctrl);
\draw [fill = black] (19*\xgap,-\reggap-2*\ygap) circle (\rCXctrl);
\draw (19*\xgap,-\reggap-3*\ygap) circle (\rCXnot);

\draw (20*\xgap,-2*\ygap)--(20*\xgap,-\reggap-4*\ygap-\rCXnot);
\draw [fill = white] (20*\xgap,-2*\ygap) circle (\rCXctrl);
\draw [fill = black] (20*\xgap,-\reggap-3*\ygap) circle (\rCXctrl);
\draw (20*\xgap,-\reggap-4*\ygap) circle (\rCXnot);

\draw (21*\xgap,-\reggap-4*\ygap)--(21*\xgap,-3*\reggap-8*\ygap);
\draw [fill = black] (21*\xgap,-\reggap-4*\ygap) circle (\rCXctrl);
\draw (21*\xgap-\rSwap,-2*\reggap-4*\ygap-\rSwap)--(21*\xgap+\rSwap,-2*\reggap-4*\ygap+\rSwap);
\draw (21*\xgap-\rSwap,-2*\reggap-4*\ygap+\rSwap)--(21*\xgap+\rSwap,-2*\reggap-4*\ygap-\rSwap);
\draw (21*\xgap-\rSwap,-3*\reggap-8*\ygap-\rSwap)--(21*\xgap+\rSwap,-3*\reggap-8*\ygap+\rSwap);
\draw (21*\xgap-\rSwap,-3*\reggap-8*\ygap+\rSwap)--(21*\xgap+\rSwap,-3*\reggap-8*\ygap-\rSwap);

\draw (22*\xgap,-2*\ygap)--(22*\xgap,-\reggap-4*\ygap-\rCXnot);
\draw [fill = white] (22*\xgap,-2*\ygap) circle (\rCXctrl);
\draw [fill = black] (22*\xgap,-\reggap-3*\ygap) circle (\rCXctrl);
\draw (22*\xgap,-\reggap-4*\ygap) circle (\rCXnot);

\draw (23*\xgap,-2*\ygap)--(23*\xgap,-\reggap-4*\ygap-\rCXnot);
\draw [fill = black] (23*\xgap,-2*\ygap) circle (\rCXctrl);
\draw [fill = black] (23*\xgap,-\reggap-3*\ygap) circle (\rCXctrl);
\draw (23*\xgap,-\reggap-4*\ygap) circle (\rCXnot);

\draw (24*\xgap,-\reggap-4*\ygap)--(24*\xgap,-3*\reggap-9*\ygap);
\draw [fill = black] (24*\xgap,-\reggap-4*\ygap) circle (\rCXctrl);
\draw (24*\xgap-\rSwap,-2*\reggap-4*\ygap-\rSwap)--(24*\xgap+\rSwap,-2*\reggap-4*\ygap+\rSwap);
\draw (24*\xgap-\rSwap,-2*\reggap-4*\ygap+\rSwap)--(24*\xgap+\rSwap,-2*\reggap-4*\ygap-\rSwap);
\draw (24*\xgap-\rSwap,-3*\reggap-9*\ygap-\rSwap)--(24*\xgap+\rSwap,-3*\reggap-9*\ygap+\rSwap);
\draw (24*\xgap-\rSwap,-3*\reggap-9*\ygap+\rSwap)--(24*\xgap+\rSwap,-3*\reggap-9*\ygap-\rSwap);

\draw (25*\xgap,-2*\ygap)--(25*\xgap,-\reggap-4*\ygap-\rCXnot);
\draw [fill = black] (25*\xgap,-2*\ygap) circle (\rCXctrl);
\draw [fill = black] (25*\xgap,-\reggap-3*\ygap) circle (\rCXctrl);
\draw (25*\xgap,-\reggap-4*\ygap) circle (\rCXnot);

\draw (26*\xgap,-1*\ygap)--(26*\xgap,-\reggap-3*\ygap-\rCXnot);
\draw [fill = white] (26*\xgap,-1*\ygap) circle (\rCXctrl);
\draw [fill = black] (26*\xgap,-\reggap-2*\ygap) circle (\rCXctrl);
\draw (26*\xgap,-\reggap-3*\ygap) circle (\rCXnot);

\draw (27*\xgap,-1*\ygap)--(27*\xgap,-\reggap-3*\ygap-\rCXnot);
\draw [fill = black] (27*\xgap,-1*\ygap) circle (\rCXctrl);
\draw [fill = black] (27*\xgap,-\reggap-2*\ygap) circle (\rCXctrl);
\draw (27*\xgap,-\reggap-3*\ygap) circle (\rCXnot);

\draw (28*\xgap,-2*\ygap)--(28*\xgap,-\reggap-4*\ygap-\rCXnot);
\draw [fill = white] (28*\xgap,-2*\ygap) circle (\rCXctrl);
\draw [fill = black] (28*\xgap,-\reggap-3*\ygap) circle (\rCXctrl);
\draw (28*\xgap,-\reggap-4*\ygap) circle (\rCXnot);

\draw (29*\xgap,-\reggap-4*\ygap)--(29*\xgap,-3*\reggap-10*\ygap);
\draw [fill = black] (29*\xgap,-\reggap-4*\ygap) circle (\rCXctrl);
\draw (29*\xgap-\rSwap,-2*\reggap-4*\ygap-\rSwap)--(29*\xgap+\rSwap,-2*\reggap-4*\ygap+\rSwap);
\draw (29*\xgap-\rSwap,-2*\reggap-4*\ygap+\rSwap)--(29*\xgap+\rSwap,-2*\reggap-4*\ygap-\rSwap);
\draw (29*\xgap-\rSwap,-3*\reggap-10*\ygap-\rSwap)--(29*\xgap+\rSwap,-3*\reggap-10*\ygap+\rSwap);
\draw (29*\xgap-\rSwap,-3*\reggap-10*\ygap+\rSwap)--(29*\xgap+\rSwap,-3*\reggap-10*\ygap-\rSwap);

\draw (30*\xgap,-2*\ygap)--(30*\xgap,-\reggap-4*\ygap-\rCXnot);
\draw [fill = white] (30*\xgap,-2*\ygap) circle (\rCXctrl);
\draw [fill = black] (30*\xgap,-\reggap-3*\ygap) circle (\rCXctrl);
\draw (30*\xgap,-\reggap-4*\ygap) circle (\rCXnot);

\draw (31*\xgap,-2*\ygap)--(31*\xgap,-\reggap-4*\ygap-\rCXnot);
\draw [fill = black] (31*\xgap,-2*\ygap) circle (\rCXctrl);
\draw [fill = black] (31*\xgap,-\reggap-3*\ygap) circle (\rCXctrl);
\draw (31*\xgap,-\reggap-4*\ygap) circle (\rCXnot);

\draw (32*\xgap,-\reggap-4*\ygap)--(32*\xgap,-3*\reggap-11*\ygap);
\draw [fill = black] (32*\xgap,-\reggap-4*\ygap) circle (\rCXctrl);
\draw (32*\xgap-\rSwap,-2*\reggap-4*\ygap-\rSwap)--(32*\xgap+\rSwap,-2*\reggap-4*\ygap+\rSwap);
\draw (32*\xgap-\rSwap,-2*\reggap-4*\ygap+\rSwap)--(32*\xgap+\rSwap,-2*\reggap-4*\ygap-\rSwap);
\draw (32*\xgap-\rSwap,-3*\reggap-11*\ygap-\rSwap)--(32*\xgap+\rSwap,-3*\reggap-11*\ygap+\rSwap);
\draw (32*\xgap-\rSwap,-3*\reggap-11*\ygap+\rSwap)--(32*\xgap+\rSwap,-3*\reggap-11*\ygap-\rSwap);

\draw (33*\xgap,-2*\ygap)--(33*\xgap,-\reggap-4*\ygap-\rCXnot);
\draw [fill = black] (33*\xgap,-2*\ygap) circle (\rCXctrl);
\draw [fill = black] (33*\xgap,-\reggap-3*\ygap) circle (\rCXctrl);
\draw (33*\xgap,-\reggap-4*\ygap) circle (\rCXnot);

\draw (34*\xgap,-1*\ygap)--(34*\xgap,-\reggap-3*\ygap-\rCXnot);
\draw [fill = black] (34*\xgap,-1*\ygap) circle (\rCXctrl);
\draw [fill = black] (34*\xgap,-\reggap-2*\ygap) circle (\rCXctrl);
\draw (34*\xgap,-\reggap-3*\ygap) circle (\rCXnot);

\draw (35*\xgap,0)--(35*\xgap,-\reggap-\ygap*2-\rCXnot);
\draw [fill = black] (35*\xgap,0) circle (\rCXctrl);
\draw (35*\xgap,-\reggap-\ygap*2) circle (\rCXnot);

\draw [lightgray, dash pattern=on 3pt off 1.5pt, rounded corners = \groupRound] (4*\xgap-\groupPad,-2*\ygap+\groupPad) rectangle (5*\xgap+\groupPad,-\reggap-4*\ygap-\groupPad);
\draw [lightgray, dash pattern=on 3pt off 1.5pt, rounded corners = \groupRound] (8*\xgap-\groupPad,-1*\ygap+\groupPad) rectangle (9*\xgap+\groupPad,-\reggap-3*\ygap-\groupPad);
\draw [lightgray, dash pattern=on 3pt off 1.5pt, rounded corners = \groupRound] (12*\xgap-\groupPad,-2*\ygap+\groupPad) rectangle (13*\xgap+\groupPad,-\reggap-4*\ygap-\groupPad);

\draw [lightgray, dash pattern=on .4pt off 0.6pt, rounded corners = \groupRound] (17*\xgap-\groupPad,-0*\ygap+\groupPad) rectangle (18*\xgap+\groupPad,-\reggap-2*\ygap-\groupPad);

\draw [lightgray, dash pattern=on 3pt off 1.5pt, rounded corners = \groupRound]
(22*\xgap-\groupPad,-2*\ygap+\groupPad) rectangle (23*\xgap+\groupPad,-\reggap-4*\ygap-\groupPad);
\draw [lightgray, dash pattern=on 3pt off 1.5pt, rounded corners = \groupRound]
(26*\xgap-\groupPad,-1*\ygap+\groupPad) rectangle (27*\xgap+\groupPad,-\reggap-3*\ygap-\groupPad);
\draw [lightgray, dash pattern=on 3pt off 1.5pt, rounded corners = \groupRound]
(30*\xgap-\groupPad,-2*\ygap+\groupPad) rectangle (31*\xgap+\groupPad,-\reggap-4*\ygap-\groupPad);

\end{tikzpicture}
\captionsetup{font=small}
\captionsetup{width=0.9\textwidth}
\caption[my caption]{%
Simulation of $R_{\mathrm{Q}}$ that uses $\lceil \log_2 M \rceil$ ancilla qubits to sequentially check whether $a = 0$, $a = 1$, $a = 2$, etc., and, if so, to perform the corresponding swap. When $M$ is a power of $2$, the circuit employs $M$ Fredkin gates, $4M - 8$ Toffoli gates, and $4$ CNOT gates. Note that each pair of gates enclosed in a dashed box can be replaced by a single CNOT gate, and the pair of gates enclosed in the dotted box can be replaced by a single negation gate.}
\label{fig:RQ_log_ancilla}
\end{figure}
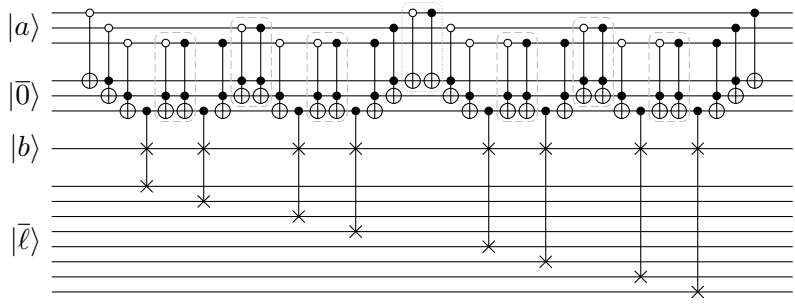

\newpage

\section{Proof of Theorem~\ref{thm:WeightsToExp}}
\label{app:expected_ratio_proof}

In this Appendix, we prove Theorem~\ref{thm:WeightsToExp}. The proof essentially follows that in \cite[Section 6.2]{jordan25:DQI-original} with some stylistic adjustments.

\bigskip

\noindent
{\bf{}Theorem~\ref{thm:WeightsToExp}} (restated){\bf{}.}
\emph{
Let $|\algostate_4\>$ be the final state of the DQI algorithm, as in \eqref{eq:psi4}.
We have 
\[
\<\algostate_4|\Phi_f|\algostate_4\> = \frac{mr}{p} + \frac{p-2r}{p}
\sum_{k=1}^\ell |w_k^2| k
+   \frac{2\sqrt{r(p-r)}}{p}
 \sum_{k=0}^{\ell-1} \Re(w_k^* w_{k+1})
\sqrt{(k+1)(m-k)}.
\]
}

\bigskip

Regarding the statement of the theorem,
recall that $\Phi_f=\sum_{x\in\F_p^n}\overline f(x)|x\>\<x|$ where $\overline f(x)=|\{i\colon b_i\cdot x\in S_i\}|$.
We have $\overline f(x)=\sum_{i=1}^m \overline f_i(x)$ for $\overline f_i(x):=f_i((Bx)_i)$. Also note that \eqref{eq:psi4} gives the final state of the algorithm as $|\algostate_4\> = (F_p^{-1})^{\otimes n}|\algostate'_4\>$
where 
\[
|\algostate'_4\>=\sum_{k=0}^\ell \sum_{\substack{y\in\F_p^m\\|y|=k}}
\frac{w_k}{\sqrt{\binom{m}{k}}}\beta_y|B^\top y\>
\]
is as in \eqref{eq:psi4prime}.
For brevity, let us denote $\<\Phi_f\>_{|\algostate_4\>}:=
\<\algostate_4|\Phi_f|\algostate_4\>$.

Let us define $\widehat\Phi_f:=F_p^{\otimes n} \Phi_f(F_p^{-1})^{\otimes n}$, and note that $\<\Phi_f\>_{|\algostate_4\>}=
\<\algostate'_4|\widehat\Phi_f |\algostate'_4\>$, which is a more convenient form of $\<\Phi_f\>_{|\algostate_4\>}$ for purposes of the proof. First, let us express $\Phi_f$ and $\widehat\Phi_f$ via generalized Pauli operators.

\subsection{Observables via generalized Pauli operators}

The generalized Pauli $X$ and $Z$ operators are defined as
\[
Z_p: = \sum_{b \in \mathbb{F}_p} \omega_p^b |b\rangle \langle b|
\qqAnd
X_p: = \sum_{b \in \mathbb{F}_p} |b+1\>\<b|,
\]
which the Fourier transform 
relates as $F_pZ_pF_p^*=X_p^*$ and $F_p^*Z_pF_p=X_p$.


\begin{clm}
\label{clm:Phi_via_Paulis}
We have
\[
\widehat\Phi_f
= \frac{1}{p}\sum_{i=1}^m\sum_{v\in S_i}\sum_{a\in\F_p} \omega_p^{-av}
\bigotimes_{j=1}^n X_p^{-a B_{ij}}.
\]
\end{clm}

\begin{proof}
Note that we can write the function
$\overline  f_i(x)\in\{0,1\}$
indicating whether $\sum_j B_{ij}x_j\in S_i$ as
\[
\overline f_i(x) = \sum_{v\in S_i}\sum_{a\in\F_p}
\omega_p^{a(\sum_j B_{ij}x_j-v)}\big/p.
\]
Given that $\overline f(x)=\sum_{i=1}^m \overline f_i(x)$, we thus get
\[
\Phi_f = \sum_{x\in\F_p^n}\overline f(x)|x\>\<x|
 = \frac1p\sum_{x\in\F_p^n}
 \sum_{i=1}^m\sum_{v\in S_i}\sum_{a\in\F_p}\omega_p^{a(\sum_j B_{ij}x_j-v)}
|x\>\<x|.
\]
Next we write 
$\omega_p^{a(\sum_j B_{ij}x_j-v)}=
\omega_p^{-av}\prod_j \omega_p^{aB_{ij}x_j}$ and we decompose the sum over $x=x_1x_2\ldots x_n\in\F_p^n$ as summing individually over each $x_j\in\F_p$. As a result, we get that 
\begin{align*}
\Phi_f &  =  \frac{1}{p}\sum_{i=1}^m\sum_{v\in S_i}\sum_{a\in\F_p} \omega_p^{-av}
\bigotimes_{j=1}^n
\sum_{x_j\in\F_p}\omega_p^{a B_{ij}x_j}
|x_j\>\<x_j|
=  \frac{1}{p}\sum_{i=1}^m\sum_{v\in S_i}\sum_{a\in\F_p} \omega_p^{-av}
\bigotimes_{j=1}^n Z_p^{a B_{ij}}.
\end{align*}
We conclude the proof by observing that $F_pZ_pF_p^{-1}=X_p^{-1}$.
\end{proof}

\subsection{Exploiting orthogonality of error syndromes }

Now, when evaluating $\<\Phi_f\>_{|\algostate_4\>} = \<\algostate'_4|\widehat\Phi_f |\algostate'_4\>$, we want to see how each term $\bigotimes_{j=1}^n X_p^{-a B_{ij}}$ of $\widehat\Phi_f$ acts on each basis state $|B^\top y\>$ in the support of $|\algostate'_4\>$.
Let $e_i\in F_p^{m}$ be the indicator vector with $1$ at index $i$ and $0$ elsewhere.
For every string $x\in\F_p^n$, we have 
$\bigotimes_{j=1}^n X_p^{-a B_{ij}} |x\> =
 |x-a B^\top e_i\>$. Therefore, for every $y'\in\F_p^m$, we have
 \[
\bigotimes_{j=1}^n X_p^{-a B_{ij}} |B^\top y'\> =
 |B^\top y'-a B^\top e_i\> 
=  |B^\top (y'-a e_i)\>.
\]
 Thus, from Claim~\ref{clm:Phi_via_Paulis} and the expression \eqref{eq:psi4prime} for $|\algostate'_4\>$, we get 
\begin{align*}
\<\Phi_f\>_{|\algostate_4\>}
&=
\frac{1}{p}
\sum_{k,k'=0}^\ell \frac{w_k^*w_{k'}}{\sqrt{\binom{m}{k}\binom{m}{k'}}}
\sum_{\substack{y,y'\in\F_p^m\\|y|=k\\|y'|=k'}}
\beta_y^*\beta_{y'}
\sum_{i=1}^m\sum_{v\in S_i}\sum_{a\in\F_p} \omega_p^{-av}
\<B^{\top}y|B^{\top}(y'-ae_i)\>.
\end{align*}

Now note that the Hamming weights of $y$ and $y'-ae_i$ are, respectively, at most $\ell$ and $\ell+1$. In addition, recall that $\ell\le d^\perp/2-1$, therefore the minimum distance $d^\perp$ of the code $C^\perp=\ker B^\top$ is at least $2\ell+2$.
This means that $B^\top y=B^\top (y'-a e_i)$ if and only if $y=y'-a e_i$, and the 
inner product $\<B^{\top}y|B^{\top}(y'-ae_i)\>$ in the above sum is $1$ if $y'=y+ae_i$ and $0$ otherwise. Hence, by taking $y'=y+a e_i$ and $k'=|y'|=|y+a e_i|$, we get
\[
\<\Phi_f\>_{|\algostate_4\>} 
= \frac1p \sum_{k=0}^\ell \frac{w_k^*}{\sqrt{\binom{m}{k}}}\sum_{\substack{y\in\F_p^m\\|y|=k}}
\beta_y^*
\sum_{i=1}^m\sum_{v\in S_i}\sum_{a\in\F_p} \omega_p^{-av}
 \frac{w_{|y+ae_i|}}{\sqrt{\binom{m}{|y+ae_i|}}} \beta_{y+ae_i}.
\]

\subsection{Four components of the expectation}

In order to evaluate the above sum for $\<\Phi_f\>_{|\algostate_4\>}$, we decompose it as a total of four other sums. The first sum corresponds to $a=0$,
 the second to $y_i = 0$ and $y'_i=a\ne 0$, the third to $y_i=-a\ne 0$ and $y'_i= 0$, and the fourth to $y_i\ne 0$, $y'_i=y_i+a\ne 0$, and $a\ne 0$.
Recall that $\beta_y = \prod_{\substack{i=1\\y_i\ne 0}}^m \widehat{g}_i(y_i)$ where $\widehat{g}_i(z)$ is given in \eqref{eq:hatg_i_def}.
Therefore, we can write 
 \begin{subequations}
 \begin{align}
\<\Phi_f\>_{|\algostate_4\>} =
& \label{eq:Phi_f_sum1}
\frac1p \sum_{k=0}^\ell \frac{|w_k^2|}{\binom{m}{k}}\sum_{\substack{y\in\F_p^m\\|y|=k}} |\beta_y^2|
\sum_{i=1}^m\sum_{v\in S_i} \omega_p^{0}
 \\ & \label{eq:Phi_f_sum2}
  +  \frac1p\sum_{k=0}^{\ell-1} \frac{w_k^*w_{k+1}}{\sqrt{\binom{m}{k}\binom{m}{k+1}}}\sum_{\substack{y\in\F_p^m\\|y|=k}} |\beta_y^2|
\sum_{\substack{i=1\\y_i=0}}^m\sum_{v\in S_i}\sum_{\substack{a\in\F_p\\a\ne 0}} \omega_p^{-av}
  \widehat g_i(a)
 \\ & \label{eq:Phi_f_sum3}
 +  \frac1p\sum_{k=1}^{\ell} \frac{w_k^*w_{k-1}}{\sqrt{\binom{m}{k}\binom{m}{k-1}}}\sum_{\substack{y\in\F_p^m\\|y|=k}} 
\sum_{\substack{i=1\\y_i\ne 0}}^m
|\beta_{y-y_ie_i}^2|
\sum_{v\in S_i}\omega_p^{y_iv} \widehat g_i^*(y_i)
 \\ & \label{eq:Phi_f_sum4}
 +  \frac1p\sum_{k=1}^\ell \frac{|w_k^2|}{\binom{m}{k}}\sum_{\substack{y\in\F_p^m\\|y|=k}} 
\sum_{\substack{i=1\\y_i\ne 0}}^m
|\beta_{y-y_ie_i}^2|
\sum_{v\in S_i}\!\sum_{\substack{a\in\F_p\\a\ne 0,-y_i}} \!\!
\omega_p^{-av}  \widehat g_i^*(y_i) \widehat g_i(y_i+a).
 \end{align}
 \end{subequations}
We will evaluate these four sums one by one.%
\footnote{\label{foot:conjugate_sums}%
However, it can be seen by symmetry that sums \eqref{eq:Phi_f_sum2} and \eqref{eq:Phi_f_sum3} are complex conjugates to one another.}
For all of them, we will use the identity
\[
\sum_{\substack{y \in \F_p^m\\|y| = k}} |\beta_y^2|
= \!\!\!\!\!\sum_{\{i_1, \ldots, i_k \}\subset \{1, \ldots, m\}}\!\!
\left( \sum_{y_{i_1} \in \mathbb{F}_p} |\widehat g_{i_1}(y_{i_1})|^2 \right)
\cdots
\left( \sum_{y_{i_k} \in \mathbb{F}_p} |\widehat g_{i_k}(y_{i_k})|^2 \right)
= \binom{m}{k},
\]
where we have used that $\widehat g_{i}(0)=0$ and, therefore,
$
\sum_{y_i \in \mathbb{F}_p^*} |\widehat g_{i}(y_i)|^2
= \sum_{y_i \in \mathbb{F}_p} |\widehat g_{i}(y_i)|^2.
$

\paragraph{First sum, \eqref{eq:Phi_f_sum1}.}

We use $\sum_{v\in S_i} \omega_p^{0}=|S_i|=r$.
The first sum thus evaluates to
\[
\frac1p \sum_{k=0}^\ell \frac{|w_k^2|}{\binom{m}{k}}\sum_{\substack{y\in\F_p^m\\|y|=k}} |\beta_y|^2 m r
= \frac{mr}{p} \sum_{k=0}^\ell |w_k^2|
= \frac{mr}{p}.
\]

\paragraph{Second sum, \eqref{eq:Phi_f_sum2}.}
By reversing the relation \eqref{eq:hatg_i_def}, which is essentially a Fourier transform, we have
$\sum_{z\in \F_p} \omega_p^{-vz}\widehat g_i(z)/\sqrt{p} = g_i(v)$.
In addition, recall from \eqref{eq:g_i_alt} that 
\[
g_i(v)=\begin{cases}
\sqrt{p-r}/\sqrt{pr}, & \text{if }v\in S_i, \\
-\sqrt{r}/\sqrt{p(p-r)}, & \text{if }v\in \F_p\setminus S_i.
\end{cases}
\]
Thus, because $\widehat g_i(0)=0$, for every $v\in S_i$ we have
 \beE
 \label{eq:hatg_to_g}
 \sum_{\substack{a\in\F_p\\a\ne 0}} \omega_p^{-av}
  \widehat g_i(a)
  =
\sum_{a\in\F_p} \omega_p^{-av}  \widehat g_i(a)  = \sqrt{p}g_i(v)
= \sqrt{p}\sqrt{\frac{p-r}{pr}}
= \sqrt{\frac{p-r}{r}}.
 \enE
Using this, we can see that the second sum evaluates to
\begin{align*}
& \frac1p\sum_{k=0}^{\ell-1} \frac{w_k^*w_{k+1}}{\sqrt{\binom{m}{k}\binom{m}{k+1}}}
\sum_{\substack{y\in\F_p^m\\|y|=k}} |\beta_y^2|
(m-k)r
\sqrt{\frac{p-r}{r}}
  \\ & \quad =
  \frac{\sqrt{r(p-r)}}{p}
 \sum_{k=0}^{\ell-1} \frac{w_k^*w_{k+1}}{\sqrt{\binom{m}{k}\binom{m}{k+1}}}
\binom{m}{k} (m-k)
  \\ & \quad =
  \frac{\sqrt{r(p-r)}}{p}
 \sum_{k=0}^{\ell-1} w_k^*w_{k+1}
\sqrt{(k+1)(m-k)}.
 \end{align*}

\paragraph{Third sum, \eqref{eq:Phi_f_sum3}.}

For the third and the fourth sums, we will use following rearrangement of summation. 
For any $\mathsf{function}(\cdot,\cdot)$ of domain $\F_p^m\times\{1,\ldots,m\}$, we have
\beE
\label{eq:sum_rewriting}
\sum_{\substack{y\in\F_p^m\\|y|=k}} 
\sum_{i\colon y_i\ne 0}
\mathsf{function}(y,i)
=
\sum_{\substack{\bar y\in\F_p^m\\|\bar y|=k-1}} 
\sum_{i\colon \bar y_i\ne 0}
\sum_{\substack{z\in\F_p\\z\ne 0}}
\mathsf{function}(\bar y+ze_i,i).
\enE
Using this identity, we can rewrite the third sum
\eqref{eq:Phi_f_sum3} as
\begin{align*}
\frac1p\sum_{k=1}^{\ell} \frac{w_k^*w_{k-1}}{\sqrt{\binom{m}{k}\binom{m}{k-1}}}\sum_{\substack{\bar y\in\F_p^m\\|\bar y|=k-1}} 
|\beta_{\bar y}^2|
\sum_{\substack{i=1\\\bar y_i= 0}}^m
\sum_{v\in S_i}
\sum_{\substack{z\in \F_p\\z\ne 0}}
\omega_p^{zv} \widehat g_i^*(z).
\end{align*}
Analogously to \eqref{eq:hatg_to_g}, we have
$
\sum_{z\in \F_p^*}\omega_p^{zv} \widehat g_i^*(z)
= (\sum_{z\in \F_p^*}\omega_p^{-zv} \widehat g_i(z))^*
= \sqrt{(p-r)/r}
$. 
Therefore the third sum evaluates to 
\begin{align*}
&
\frac1p\sum_{k=1}^{\ell} \frac{w_k^*w_{k-1}}{\sqrt{\binom{m}{k}\binom{m}{k-1}}}\sum_{\substack{\bar y\in\F_p^m\\|\bar y|=k-1}} 
|\beta_{\bar y}^2|
(m-k+1)
r
\sqrt{\frac{p-r}{r}}
\\ & \quad =
\frac{\sqrt{r(p-r)}}{p}\sum_{k=1}^{\ell} \frac{w_k^*w_{k-1}}{\sqrt{\binom{m}{k}\binom{m}{k-1}}}
\binom{m}{k-1}
(m-k+1)
\\ & \quad =
\frac{\sqrt{r(p-r)}}{p}\sum_{k=1}^{\ell}
w_k^*w_{k-1}
\sqrt{k(m-k+1)}.
\end{align*}
As noted in Footnote~\ref{foot:conjugate_sums}, we can see that this is the complex conjugate of the second sum \eqref{eq:Phi_f_sum2}.

\paragraph{Fourth sum, \eqref{eq:Phi_f_sum4}.}

As for the third sum, we use the identity  \eqref{eq:sum_rewriting} to rearrange the summation, and we get that the fourth sum equals
\beE
\label{eq:Phi_f_sum4_alt}
\frac{1}{p}
\sum_{k=1}^\ell \frac{|w_k^2|}{\binom{m}{k}}\sum_{\substack{\bar y\in\F_p^m\\|\bar y|=k-1}} 
|\beta_{\bar y}^2|
\sum_{\substack{i=1\\\bar y_i= 0}}^m
\sum_{v\in S_i}
\sum_{\substack{z\in\F_p\\z\ne 0}}
\sum_{\substack{a\in\F_p\\a\ne 0,-z}} \!\!
\omega_p^{-av}
  \widehat g_i(z+a) \widehat g_i^*(z)
\enE
Because $\widehat g_i(0)=0$, we can include $z=0$ and $a=-z$ in the above summation, and then rewrite the resulting sum over $z$ and $a$ as
\beE
\label{eq:part_of_4th} 
 \sum_{z\in \F_p}
\sum_{\substack{a\in\F_p\\a\ne 0}} 
\omega_p^{-av}  \widehat g_i(z+a) \widehat g_i^*(z)  
 =
 \sum_{z,a\in \F_p}
\omega_p^{-av}
  \widehat g_i(z+a) \widehat g_i^*(z)
  -
  \sum_{z\in \F_p}
|\widehat g_i(z)|^2.
\enE
The latter sum is easy:  
$\sum_{z}|\widehat g_i(z)|^2
=\sum_{z}|g_i(z)|^2=1$.
For the former sum, we expand both instances of $\widehat g_i(\cdot)$ using \eqref{eq:hatg_i_def}:
\begin{multline*}
 \sum_{z,a\in \F_p}
\omega_p^{-av}
  \widehat g_i(z+a) \widehat g_i^*(z)
 =
\sum_{z,a,z',z''\in \F_p}
\omega_p^{-av}
\omega_p^{(z+a)z'}g_i(z)
\omega_p^{-zz''}g_i(z'')/p
 \\ =
\sum_{z',z''\in \F_p}
g_i(z')g_i(z'')
\sum_{z\in \F_p}
\omega_p^{z(z'-z'')}
\sum_{a\in \F_p}
\omega_p^{a(z'-v)}/p
 =
p g_i(v)^2 = \frac{p-r}{r},
\end{multline*}
where the last equality is due to $v\in S_i$.
We thus get that \eqref{eq:part_of_4th} equals $(p-2r)/r$.

Hence, the fourth sum \eqref{eq:Phi_f_sum4}, when first rewritten as \eqref{eq:Phi_f_sum4_alt}, evaluates to
\begin{multline*}
\frac{1}{p}
\sum_{k=1}^\ell \frac{|w_k^2|}{\binom{m}{k}}\sum_{\substack{\bar y\in\F_p^m\\|\bar y|=k-1}} 
|\beta_{\bar y}^2|
\sum_{\substack{i=1\\\bar y_i= 0}}^m
\sum_{v\in S_i}
\frac{p-2r}{r}
=
\frac{p-2r}{p}
\sum_{k=1}^\ell \frac{|w_k^2|}{\binom{m}{k}}
(m-k+1)
\sum_{\substack{\bar y\in\F_p^m\\|\bar y|=k-1}} 
|\beta_{\bar y}^2|
\\
=
 \frac{p-2r}{p}
\sum_{k=1}^\ell \frac{|w_k^2|}{\binom{m}{k}}
(m-k+1)
\binom{m}{k-1}
=
\frac{p-2r}{p}
\sum_{k=1}^\ell |w_k^2| k.
\end{multline*}
 
\bigskip

Now that we have evaluated all the four sums \eqref{eq:Phi_f_sum1}, \eqref{eq:Phi_f_sum2}, \eqref{eq:Phi_f_sum3}, \eqref{eq:Phi_f_sum4}, we can see that their total $\<\Phi_f\>_{|\algostate_4\>}$ is as given in Theorem~\ref{thm:WeightsToExp}. This concludes the proof.

\section{Binomial distribution}
\label{app:binomial}

Recall the definition \eqref{eq:q_def} of $\nmean=\ell/m-m^{-1/2+c}$ and recall that $\ell\le d^\perp/2-1 < m/2$. Let us consider the binomial distribution with parameters $m$ and $\nmean$, and let $K$ be a random variable following this distribution. That is, for $k=0,1,\ldots,m$, we have
\[
p_k:=\Pr(K=k)= \nmean^k(1-\nmean)^{m-k}\binom{m}{k}
\]
where in this appendix we have introduced the notation $p_k$ for brevity. Note that the amplitudes $w_k'$ defined in \eqref{eq:wprime_k_def} are simply $w_k'=\sqrt{p_k}$.

\subsection{Tail Bounds}
\label{ssec:tail_bounds}

Let $\mu_K:=\mathbb{E}[K]$ denote the expected value of $K$, which for the binomial distribution is $\mu_K=m\nmean=\ell-m^{1/2+c}$. We want to bound probabilities that $K$ deviates from $\mu_k$ by more than $m^{1/2+c}$, and, intuitively, since the standard deviation of the distribution is $\sqrt{m\nmean(1-\nmean)}\le\sqrt{m}/2\ll m^{1/2+c}$, we expect this probability to be small.

We already have the assumption $\ell<m/2$ from above. In addition, we will also assume that $\ell$ is not too small, which is reasonable for our applications, where $\ell$ grows linearly with $m$.

\begin{clm}
\label{clm:tail_bound}
Assuming $\ell\ge 4m^{1/2+c}$, we have both
\[
\Pr(K > \ell)\le  \exp(-m^{2c}/2)
\qqAnd
\Pr(K < \ell-2m^{1/2+c})\le  \exp(-m^{2c}/2).
\]
\end{clm}

\begin{proof}
For any positive $\delta\le 1$, the Chernoff bound implies
\begin{align*}
& \Pr(K\ge (1+\delta)\mu_K) \le \exp(-\mu_K\delta^2/3),
\\& \Pr(K\le (1-\delta)\mu_K) \le \exp(-\mu_K\delta^2/2).
\end{align*}
(see \cite[Theorems~4.4 and 4.5]{mitzenmacher17:prob}).
We choose $\delta:=m^{1/2+c}/\ell$, for which we have both
\begin{align*}
& (1+\delta)\mu_K = 
(1+m^{1/2+c}/\ell)(\ell-m^{1/2+c})=
\ell - m^{1+2c}/\ell < \ell,
\\ &
(1-\delta)\mu_K =
(1-m^{1/2+c}/\ell)(\ell-m^{1/2+c}) = 
\ell-2m^{1/2+c}+m^{1+2c}/\ell \ge
\ell-2m^{1/2+c}.
\end{align*}
Because the claim assumes $m^{1/2+c}\le \ell/4$, we also have
$\mu_K
\ge 3\ell/4 > \ell/2$.
Therefore, by the Chernoff bound, 
\begin{align*}
&
\Pr\left(K \ge \ell \right)
\le \exp\left(-\frac{(3\ell/4)(m^{1/2+c}/\ell)^2}{3}\right)
= \exp\left(-\frac{m^{1+2c}}{4\ell}\right),
\\ &
\Pr(K < \ell-2m^{1/2+c})
\le \exp\left(-\frac{(\ell/2)(m^{1/2+c}/\ell)^2}{2}\right)
= \exp\left(-\frac{m^{1+2c}}{4\ell}\right).
\end{align*}
The proof follows by $\ell\le m/2$.
\end{proof}

\subsection{Probability mass of the mode}
\label{ssec:mode}

Here we place an upper bound on the largest probability $p_k=\Pr(K=k)$ under some reasonable assumptions. Define $\kappa:=\lceil \nmean(m+1)\rceil - 1$. The following claim states that $k=\kappa$ is a mode of the probability distribution $(p_k)_k$ by stating that $p_k$ is an increasing function for $k\le\kappa$ and a decreasing function for $k\ge\kappa$.

\begin{clm}
\label{clm:mode_k}
We have $\max_k p_k = p_\kappa$ and $p_k \ge p_{k+1}$ if and only if $k\ge \nmean(m+1)-1$.
\end{clm}

\begin{proof}
From the expression for $p_k$, we can see that for all $k=0,1,\ldots, m-1$ we have
\[
\frac{p_k}{p_{k+1}}
= \frac{(1-\nmean)(k+1)}{\nmean(m-k)}.
\]
First of all, this shows that $p_k=p_{k+1}$ if and only if $(1-\nmean)(k+1) = \nmean(m-k)$, that is, if and only if $\nmean(m+1)-1$ is an integer and $k=\nmean(m+1)-1$. Second, it shows that $p_k \ge p_{k+1}$ if and only if $k\ge \nmean(m+1)-1$. This means that
$\max_k p_k = p_\kappa$ for $\kappa:=\lceil \nmean(m+1)\rceil - 1$, which is a mode of the binomial distribution.%
\footnote{More generally, a mode of a probability distribution is the value that occurs most frequently. The binomial distribution is multimodal if and only if $\nmean(m+1)-1$ is an integer.} 
\end{proof}

\begin{clm}
\label{clm:mode_mass}
Assuming $q$ satisfies $7\le \nmean(m+7)\le m$, we have that $p_k\le \frac3{\sqrt{m\nmean(1-\nmean)}}$ for all $k=0,1,\ldots, m$.
\end{clm}

We essentially assume that $\frac{7}{m+7}\le \nmean\le 1-\frac{7}{m+7}$, which is the case in our applications, where $\nmean,1-\nmean=\Theta(1)$. Therefore, the above claim says that, in our applications $\max_kp_k= \OO(1/\sqrt{m})$. Note that this assumption also implicitly assumes $m\ge 7$, which we will find useful.

\begin{proof}
To upper bound $p_k$, let us use Stirling's formula, as given by Robbins \cite{robbins55:stirling}. For any integer $n\ge 1$, we have $1< \ee^{\frac{1}{12n+1}}$, and $\ee^{\frac{1}{12n}} < \sqrt{\pi/2}$, and therefore
\[
\sqrt{2\pi n} \left(\frac{n}{e}\right)^n  
< \sqrt{2\pi n} \left(\frac{n}{e}\right)^n \ee^{\frac{1}{12n+1}} 
< n! 
< \sqrt{2\pi n} \left(\frac{n}{e}\right)^n \ee^{\frac{1}{12n}}
< \pi\sqrt{n} \left(\frac{n}{e}\right)^n.
\]
Hence, for all $k=0,1,\ldots, m$, we have
\begin{align*}
p_k
& = \nmean^k(1-\nmean)^{m-k}\frac{m!}{k!(m-k)!}
\\& < \nmean^k(1-\nmean)^{m-k}\frac{\pi \sqrt{m} \left(\frac{m}{e}\right)^m}{\sqrt{2\pi k} \left(\frac{k}{e}\right)^k\sqrt{2\pi (m-k)} \left(\frac{m-k}{e}\right)^{m-k}}
\\& = 
\Big(\frac{\nmean m}{k}\Big)^k
\Big(\frac{(1-\nmean)m}{m-k}\Big)^{m-k} 
\cdot
\frac{\sqrt{ m} }{2\sqrt{k(m-k)} }.
\end{align*}
The following claim upper bounds the former factor for the mode $k=\kappa$.

\begin{clm}
For $\kappa=\lceil \nmean(m+1)\rceil-1$ and $q$ satisfying 
$7\le \nmean(m+7)\le m$, we have
\[
\Big(\frac{\nmean m}{\kappa}\Big)^\kappa
\Big(\frac{(1-\nmean)m}{m-\kappa}\Big)^{m-\kappa} 
\le 4.
\]
\end{clm}

\begin{proof}
Let us first consider the inverse quantity.
For the mode $\kappa$, we have $\nmean m+\nmean-1 \le \kappa\le \nmean m+\nmean$, and therefore
\begin{align*}
\Big(\frac{\kappa}{\nmean m}\Big)^\kappa
\Big(\frac{m-\kappa}{(1-\nmean)m}\Big)^{m-\kappa} 
%
& \ge
\Big(\frac{\nmean m+\nmean-1}{\nmean m}\Big)^\kappa
\Big(\frac{m-\nmean m-\nmean}{m-\nmean m}\Big)^{m-\kappa} 
\\ & \ge
\Big(\frac{\nmean m+\nmean-1}{\nmean m}\Big)^{\nmean m+\nmean}
\Big(\frac{m-\nmean m-\nmean}{m-\nmean m}\Big)^{m-\nmean m-\nmean+1} 
\\ & =
\bigg(1-\frac{1}{\frac{\nmean m}{1-\nmean}}\bigg)^{\nmean(m+1)}
\bigg(1-\frac{1}{\frac{(1-\nmean)m}{q}}\bigg)^{(1-\nmean)(m+1)} .
\end{align*}
Note that our assumptions that $7\le \nmean(m+7)$ and $\nmean(m+7)\le m$ are equivalent to, respectively, 
$\frac{\nmean m}{1-\nmean}\ge 7$ and $\frac{(1-\nmean)m}{q}\ge 7$, which bound the two denominators in the above expression. For any real $x\ge 7$ we have $(1-1/x)^x \ge 1/3$, and thus $1-1/x \ge 1/3^{1/x}$.

Now, by returning from the inverse of the quantity that we want to bound back to the quantity itself, we get
\[
\Big(\frac{\nmean m}{\kappa}\Big)^\kappa
\Big(\frac{(1-\nmean)m}{m-\kappa}\Big)^{m-\kappa} 
\le
3^{\frac{1-\nmean}{\nmean m}\nmean(m+1)} 3^{\frac{q}{(1-\nmean)m}(1-\nmean)(m+1)} 
= 3^{1+1/m} \le 4,
\]
where the last inequality is due to $m\ge 7$.
\end{proof}

Hence, by incorporating this claim in the earlier upper bound on $p_k$, for the mode 
$\kappa=\lceil \nmean(m+1)\rceil-1$ we get $p_\kappa \le  2\sqrt{m}/\sqrt{\kappa(m-\kappa)}$. 
It is left to bound $\kappa(m-\kappa)$.

First, because $\nmean m+\nmean-1 \le \kappa\le \nmean m+\nmean$, we get that
\[
\kappa(m-\kappa)
\ge (\nmean m+\nmean-1)(m-\nmean m-\nmean)
 =\nmean(1-\nmean)(m+1)^2 - m,
\]
We decompose $\nmean(1-\nmean)(m+1)^2$ into two parts as
\[
\nmean(1-\nmean)(m+1)^2 =
\underbrace{7/16\cdot \nmean(1-\nmean)(m+1)^2}_{\ge m}
+ \underbrace{9/16\cdot \nmean(1-\nmean)(m+1)^2}_{\ge \nmean(1-\nmean)m^2 /2},
\]
where the latter lower bound is trivial, and we obtain the former as follows.

From the assumption that both $q,1-\nmean\ge 7/(m+7)$,
we get that
\[
\nmean(1-\nmean)\ge \frac{7}{m+7}\left(1-\frac{7}{m+7}\right) =\frac{7m}{(m+7)^2},
\]
and therefore
\[
\frac{7}{16}\nmean(1-\nmean)(m+1)^2
 \ge \frac{49}{16} m \frac{(m+1)^2}{(m+7)^2}.
\]
This, because for $m\ge 7$ we have $(m+1)^2/(m+7)^2\ge 16/49$, is at least $m$, as required.

We have obtained that $\kappa(m-\kappa)\ge \nmean(1-\nmean)m^2/2$.
Hence $p_\kappa \le  \sqrt{8}/\sqrt{m\nmean(1-\nmean)}$, and the proof is concluded by taking $\sqrt{8}\le 3$.\end{proof}

\end{document}